\documentclass[12pt,9.5pt]{article}

\usepackage{units}
\usepackage{enumerate}
\usepackage{booktabs}
\usepackage{tabularx}
\usepackage{arydshln} 
\usepackage{amssymb}
\usepackage{amsmath}
\usepackage{xcolor}
\usepackage{graphicx}
\usepackage{subfigure}
\usepackage{tikz}
\usepackage{pgfplots}
\usepackage{bbm}

    \usepackage{caption}
\usetikzlibrary{decorations.pathreplacing,calligraphy}
\pgfplotsset{width=10cm,compat=1.9}
\definecolor{green}{HTML}{30AE17}
\usepgfplotslibrary{fillbetween}
\usetikzlibrary{graphs}
\usepackage{comment}
\usepackage{tikz}
\usepackage{amsfonts}
\usepackage{amsmath}
\usepackage{amsthm}
\usepackage{amssymb}
\usepackage{graphicx}
\usepackage{natbib}
\usepackage{url}
\usepackage{subcaption}
\usepackage{graphicx}
\usepackage{epstopdf}
\usepackage{tabularx,ragged2e,booktabs,caption}
\usepackage[margin=1in]{geometry}
\usepackage{color, colortbl}
\usepackage{array}
\usepackage[skip=1pt]{caption}

\newtheorem{theorem}{Theorem}

\newtheorem{assumption}[theorem]{Assumption}

\newtheorem{claim}[theorem]{Claim}

\newtheorem{definition}{Definition}

\newtheorem{lemma}{Lemma}

\newtheorem{proposition}{Proposition}
\newtheorem{remark}{Remark}

\newcommand{\savvy}{selective}
\newcommand{\Savvy}{Selective}
\newcommand{\inexperienced}{casual}
\newcommand{\Inexperienced}{Casual}
\newcommand{\inexp}{casual}
\newcommand{\CSsav}{CS_{sel}}
\newcommand{\CSin}{CS_{cas}}

\let\svmaketitle\maketitle
\def\maketitle{\svmaketitle\thispagestyle{empty}}

\begin{document}

\title{Limits of Disclosure in Search Markets\thanks{%
 We thank Yunus Aybas, Mike Baye, Alessandro Bonatti, Piotr Dworczak, Huiyi Guo, Emir Kamenica, Navin Kartik, Aaron Kolb, R. Vijay Krishna, Jorge Lemus, Jordan Martel, Eric Maskin, Marilyn Pease, Priyanka Sharma, Mark Whitmeyer, Kun Zhang, conference participants at the 2023 Society for Advancement of Economic Theory, the 2023 European Association for Research in Industrial Economics conference, the 2023 Southern Economic Meetings, the 2024 American Economic Association Meetings, the 2025 Society for Advancement of Economic Theory, the 2026 IIOC, and seminar participants at Arizona State University, Florida State University, HEC Paris, Kelley School of Business, ROY (Sciences Po/Paris School of Economics), Texas A\&M University, the University of Illinois Urbana-Champaign, the University of Iowa, the University of Queensland for helpful comments and suggestions. Boleslavsky thanks the Weimer Faculty Fellowship at Indiana University for financial support. An earlier version of this paper was circulated under the title ``Information and Pricing in Search Markets.''}}

\author{Raphael Boleslavsky\footnote{Kelley School of  Business, Indiana University. E-mail: rabole@iu.edu.} \  and  Silvana Krasteva\footnote{Department of Economics, Texas A\&M University. E-mail: ssk8@tamu.edu.}\\ }
\maketitle

\begin{abstract}

This paper examines competitive information disclosure in search markets with a mix of \savvy{} consumers, who search costlessly, and \inexperienced{} consumers, who incur positive search costs. \Savvy{} consumers incentivize truthful disclosure; \inexperienced{} consumers, concealment. With both types, the equilibrium features partial disclosure, which persists despite intense competition. \Inexperienced{} consumers may search actively, but only in small markets. While \savvy{} consumers benefit from increased competition, \inexperienced{} consumers may be harmed. Changes in search costs have non-monotone effects: the equilibrium approaches full disclosure as the cost approaches zero or the mean valuation. When costs are high, further increases raise equilibrium informativeness and the welfare of \savvy{} consumers, while reducing the welfare of \inexperienced{} consumers.

\textbf{JEL Classifications: D4, D21, D83, L13}
	
\textbf{Keywords: information design, search market, heterogeneous search cost} 

\end{abstract}

\newpage

\begin{flushleft}As the number of choices grows... choice no longer liberates, but debilitates.
\end{flushleft}
\begin{flushright}
― Barry Schwartz, \textit{The Paradox of Choice: Why More Is Less}
\end{flushright}

\section{Introduction}

\baselineskip16pt

Online markets increasingly monetize attention. Streaming platforms, for example, reward producers based on user engagement, either directly through payments or indirectly through renewals, sequels, or other forms of project continuation. Successfully engaging users can be challenging, however, because such platforms offer many competing alternatives, which are often free to access. To persuade users to engage with their content, producers design trailers, previews, samples, or other demos that convey information about how well the content matches individual users’ tastes. Thus, competition on such platforms is primarily through the provision of information, rather than price. At the same time, users differ in their willingness and ability to process such information. For some users, exploring the platform's offerings is easy, or even enjoyable, and they evaluate many options before deciding. Others may find the process of exploration taxing, and their searches tend to be less exhaustive. Thus, identical information can lead to different search paths and choices depending on a user's willingness to explore, shaping producers' incentives to reveal information.

A variety of other markets, both online and offline, have similar features. Software companies offer trial versions that convey information about product attributes, while often facing price constraints.\footnote{Third-party developers on platforms such as the Apple App Store must comply with rules that constrain prices and prevent sales outside the App Store. See Epic Games, Inc. v. Apple Inc., No. 4:20-cv-05640-YGR (N.D. Cal.)} In cryptocurrency markets, fundamentals are difficult to assess, and new projects are hard to evaluate using traditional methods.  Because returns are highly skewed, small differences in initial prices are often negligible relative to the potential upside. Consequently, investment decisions are guided largely by information that shapes investors’ perceptions of projects, with prices playing a more limited role. Public schools are free to attend, so they compete for students by providing families with information about curriculum, culture,  values, and the overall learning environment. Similar forces arise in nonprofit fundraising, where organizations selectively emphasize aspects of their mission to attract donors, and in political or cultural movements, where leaders garner support by shaping followers' perceptions of the movement and its goals. In all of these settings, individuals vary in patience, preferences, or access to technology, resulting in heterogeneous incentives for exploration of different alternatives.

This paper studies markets that feature competitive information provision, search, and consumer heterogeneity. Our parsimonious model builds on the canonical framework  of \citet{stahl1989oligopolistic}, but instead of prices, firms choose signal structures that convey information about consumers’ valuations. The market consists of a finite number of horizontally differentiated firms and a continuum of consumers with unit demand and uncertain valuations. As in \citet{stahl1989oligopolistic}, consumers search sequentially with free recall, choosing which firms to visit and when to stop. They differ in their search costs: \inexperienced{} consumers incur a positive cost to visit and may stop early, while \savvy{} consumers have zero (or negative) cost and visit all firms. Upon visiting a firm, a consumer observes additional information, learning about her valuation for the firm's product. The information structure is designed by the firm with no restrictions, as in the Bayesian persuasion literature. The firms' information structures are not observed by consumers before they visit, i.e., search is undirected. In this market, a consumer's incentive to search and a firm's incentive to provide information are intertwined. We characterize the unique symmetric Perfect Bayesian equilibrium and study how market structure and search frictions shape information provision, consumer behavior, and welfare, contrasting our findings with more traditional search markets where firms compete in price.

The two market segments generate opposing incentives---in isolation, each leads to an extreme equilibrium outcome. Indeed, when all consumers are \savvy{}, full disclosure is the unique equilibrium in sufficiently competitive markets \citep{hwang2019competitive}. In contrast, when all consumers are \inexperienced{}, the equilibrium features an ``informational Diamond paradox,'' in which each firm conceals any information that would trigger consumers to search, and they select the first product they encounter  \citep{au2023attraction}.\footnote{These authors also consider a separate case in which firms commit to information provision before being visited, which is a directed search.}

The mix of \savvy{} and \inexperienced{} consumers in the market creates tension between disclosure and concealment.  While \savvy{} consumers visit all firms before purchase, \inexperienced{} consumers follow a standard reservation strategy.\footnote{In a symmetric setting, \inexperienced{} consumers stop searching when they encounter an expected valuation weakly above their reservation value or after they have visited all firms. If all firms are visited, they select the product with the highest expected valuation.} In response, firms conceal some true valuations just below the reservation value. Revealing such valuations truthfully incentivizes a \inexperienced{} consumer to continue searching. Pooling them with valuations above the reservation value increases the posterior mean, which can induce a \inexperienced{} consumer to stop. At the same time, for true valuations above the reservation value, such pooling lowers the posterior mean and reduces the probability that a \savvy{} consumer chooses the firm's product. Thus, the \savvy{} consumers introduce a cost of concealment that tempers the incentive to pool valuations and foreclose the \inexperienced{} consumers' search. We show that the unique equilibrium always exhibits such pooling in an interval around the reservation value. In contrast, sufficiently low or high valuations, far from the reservation value, may be revealed truthfully.

Crucially, the extent of information distortion depends on the market's structure and the search technology. In a small market, a firm must carefully balance the incentives for disclosure and concealment. With few competitors, a firm has a meaningful chance to eventually convert a \inexperienced{} consumer even if she leaves initially. Simultaneously, even moderate valuations are competitive among \savvy{} consumers when disclosed truthfully. As a result, firms have relatively little to gain and much to lose from distorting information. Indeed, we show that in many small markets, only a narrow interval of valuations just below the reservation value is concealed, while lower valuations are revealed truthfully. The incentives are reversed in large markets. With many competitors, a \inexperienced{} consumer who leaves is unlikely to return, and only very high valuations can realistically compete for \savvy{} consumers. Consequently, the benefit of concealing low valuations is substantial while the cost of concealment is relatively small, leading to a highly distorted equilibrium. We show that in sufficiently large markets, firms conceal all valuations below the reservation value. Even in the limit as the number of firms grows large, the incentive to capture \inexperienced{} consumers prevents the market from approaching full disclosure. This stands in contrast to the existing literature on competitive information provision, where strong competition leads to full disclosure \citep{ivanov2013information,au2020competitive,hwang2019competitive}.

The distinct patterns of equilibrium disclosure in small and large markets have intriguing implications. In small markets, some valuations below the reservation value are disclosed truthfully, even though they induce further search by \inexperienced{} consumers. In large markets, however, all such valuations are concealed. Consequently, \inexperienced{} consumers search actively in small markets, but not in large ones, a result we refer to as the \textit{paradox of choice}. By implication, \inexperienced{} consumers may be better off in less competitive markets, where a balanced disclosure fosters active search. Neither of these findings arises in standard search models, but they relate to a growing literature documenting that an abundance of options can impede evaluation and leave decision makers worse off, a phenomenon commonly attributed to cognitive ``choice overload'' \citep{IL2000,CBG2015}. Our results highlight a complementary supply-side force: strategic information provision can reduce or eliminate search as competition intensifies. In this respect, our analysis is related to \citet{K2008} and \citet{nocke2024consumer}, who show that  distortions resembling choice overload can arise from a monopolist's profit-maximization, rather than cognitive constraints.

Changes in the search cost also affect the positive and normative properties of the equilibrium. We establish that information provision is non-monotone in the \inexperienced{} consumers' search cost, converging to full disclosure as the search cost approaches zero or the prior mean. Notably, we show that when the search cost is high, a further increase raises equilibrium informativeness, contrasting with standard intuition. This is because \inexperienced{} consumers can be persuaded to stop searching with a small distortion, and as the search cost grows, the required distortion shrinks, increasing signal informativeness. Consequently, the increase in search cost creates a conflict of interest across market segments. \Savvy{} consumers are not affected by the search cost, but they benefit indirectly from the more informative signal. Meanwhile, \inexperienced{} consumers are directly harmed by the higher search cost, and they do not benefit from increased informativeness, because they select the first product they encounter.

We extend our model in two dimensions. First, we consider a two-stage search process, whereby consumers initially hone in on a subset of the available options, and then search sequentially within it. This extension establishes that the key dimension of competition is not the absolute size of the market, but the size of the consideration set. Second, we analyze further cost heterogeneity. Rather than a single cost for the \inexperienced{} type, we allow for a more general cost distribution, showing that our equilibrium characterization in a large economy is robust while exploring the boundaries of existing results.

Our analysis highlights that equilibrium information provision is shaped by different underlying forces than prices, upending familiar findings. In our model, each firm designs its information structure freely, but it cannot tamper with the realization, as this is a consumer's subjective impression after sampling or evaluating the product. In contrast, with price competition, a firm can play a mixed strategy, but it can always deviate from the realization. Therefore, equilibrium requires that every price in the support is optimal for the firm, given the other firms' mixed strategies. In this sense, equilibrium prices are optimal ex post, while equilibrium information structures are optimal ex ante. As we discuss in Section \ref{infovsprices}, this distinction delivers sharply contrasting findings. In particular, in \citet{stahl1989oligopolistic} (and related papers), \inexperienced{} consumers stop at the first visited firm, regardless of market size. Furthermore, as search costs or competition increase, the equilibrium becomes more monopolistic, converging to the monopoly price in a large market. In contrast, our analysis links active search by \inexperienced{} consumers to market size (via the paradox of choice). As search costs increase, our equilibrium becomes more informative, approaching full disclosure, the frictionless outcome (akin to Bertrand competition). Similarly, as competition increases, our equilibrium becomes more informative (and Bertrand-like), without approaching full disclosure in the limit.

\paragraph{Related Literature.}

A growing literature examines information disclosure in competitive environments. Some papers take a centralized perspective, studying optimal disclosure by a planner. \citet{armstrong2022consumer} study centralized information provision in a horizontally differentiated Bertrand game, while \citet{dogan2022consumer} explore related issues in a search market. In a similar vein, \cite{zhou2022improved} studies how exogenous increases in information provision (via the convex order) affect search intensity, price competition, and welfare.\footnote{See also \citet{bergemann2015limits} and \citet{roesler2017buyer} who study flexible information production in a monopoly market.}

Our paper is more closely related to models in which disclosure is decentralized, arising from a strategic interaction between firms. In frictionless models, consumers observe all of the information disclosed by firms with no cost \citep{ivanov2013information, hwang2019competitive, au2020competitive}.\footnote{A number of papers study frictionless settings where competing senders reveal information about a common state \citep{kamenica2017competition, li2021sequential}.} 
In such settings, competition typically strengthens equilibrium information provision, resulting in full disclosure in large markets. In contrast, we show that the possibility of even small search frictions incentivizes concealment of low values, preventing the market from approaching full disclosure as competition grows. \cite{board2018competitive} consider strategic disclosure in a market with search frictions, but unlike our setting, sellers offer homogeneous products and consumers have an identical search cost. They derive conditions under which full disclosure and monopoly disclosure arise in equilibrium; in both cases, consumers never search past the first firm. In contrast, our equilibrium always features partial disclosure and active search by one or both consumer types. \cite{au2023attraction} also study a search model with horizontally differentiated products and a homogeneous search cost. As discussed above, when firms are unable to commit to an information disclosure strategy, such markets exhibit an informational Diamond paradox. Building on this finding, these authors  incorporate an ``attraction motive'': by making a commitment to provide more information ahead of the consumers' search decision, firms incentivize subsequent visits. The equilibrium features randomization over information structures in less competitive markets and full disclosure in large markets, contrasting with our main findings. 

In a concurrent paper, \citet{HH2025} analyze strategic disclosure in a search market where every consumer's search cost is drawn from a continuous distribution supported on a weakly positive interval, focusing mainly on the boundary case where the lowest cost in the support is zero. In this environment, the authors characterize a continuum of equilibria in which firms use ``upper-censorship'' strategies, revealing all valuations below a threshold and pooling all valuations above it into a single ``high'' realization.\footnote{Upper-censorship equilibria have an intriguing implication for consumer search behavior. In particular, in a typical upper-censorship equilibrium all consumers (regardless of search cost) stop only when they observe the ``high'' realization or sample all firms \citep[pg. 3]{HH2025}. In other words, all consumers follow the same search strategy, as if they all had the highest cost in the support. To an outside observer, therefore, the pattern of consumer search might suggest that all consumers have an identical cost.} This characterization relies on specific features of the environment that \citet{HH2025} study. In particular, if the lowest cost in the support of the search cost distribution is not zero, or if a mass of consumers visits all firms, then an upper-censorship distribution is no longer an equilibrium.\footnote{See \citet{HH2025}, April 2025, Appendix C page 1, and July 2025, Proposition 3.} In contrast, we characterize the unique equilibrium of the disclosure game with binary types without restricting attention to any specific class, with qualitatively different findings. Indeed, firms disclose high values and conceal low ones in large markets, the opposite of \citet{HH2025}. In addition, we find that equilibrium informativeness is non-monotone in search cost, rather than decreasing. In Section 6.2, we show that our characterization of disclosure in large markets is robust to additional cost heterogeneity and revisit the link to \citet{HH2025}. Overall, our analyses focus on related issues, but they adopt different modeling approaches and deliver sharply contrasting findings; they should therefore be seen as complementary. 

An important point of contrast is a significant literature on pricing in search markets, including  foundational papers that consider a mix of \savvy{} consumers, who fully learn the state of the market, and \inexperienced{} consumers, who have positive (or even infinite) search costs \citep{V1980,stahl1989oligopolistic,BM2001}. More recently, \citet{EW2012} study a model with similar structure to \citet{stahl1989oligopolistic}, in which firms can manipulate the \inexperienced{} consumers' search cost by ``obfuscating'' price.  \citet{DDZ2017} also build on \citet{stahl1989oligopolistic}, studying how production cost transparency affects prices. In Section \ref{infovsprices}, we discuss the link to \citet{stahl1989oligopolistic} and the broader literature in more detail.

\section{Setup}
\label{Setup}

\paragraph{Market.}  Consider a horizontally differentiated market with $n\geq 2$ competing firms and a unit mass of consumers. Each firm offers a single product, and each consumer eventually chooses one of them. All parties are risk-neutral, and it is common knowledge that each product's valuation $V_i\in [0,1]$ is independently and identically distributed across consumers and products, following an atomless distribution $F(\cdot)$ with mean $\mu$, and a bounded, continuous density $f(\cdot)$ with full support. To highlight the incentives most clearly, we focus on strictly convex $F(\cdot)^{n-1}$.\footnote{Strict convexity of $F(\cdot)^{n-1}$ provides a sharp characterization
in the absence of search frictions. If $F(\cdot)\in C^2([0,1])$ and $f(\cdot)>0$ on $[0,1]$, then this condition holds for all sufficiently large $n$.}
 Neither firms nor consumers know the realized valuations at the beginning of the game---consumers learn about their valuations by searching, as described below. Each firm maximizes the mass of consumers who choose its product (i.e., its demand), or equivalently, the probability that an individual consumer chooses it. Consumers maximize the value of their chosen alternatives, net of search costs.

\paragraph{Information Design.} 
Each firm controls the process by which consumers learn about their valuations for its product, which allows the firm to strategically disclose or conceal information that shapes the distribution of consumers' posterior beliefs. In particular, each firm designs a signal space, $S_i$, and a joint distribution function of valuation and signal, over $[0,1] \times S_i$. Since consumers are risk-neutral, their purchase decisions depend solely on the posterior mean, $E[V_i | s_i]$. Thus, the firms' information design problem can be recast as choosing a distribution, $G_i(\cdot)$, of the posterior mean, subject to the standard Bayes plausibility constraint that $G_i(\cdot)$ is a mean-preserving contraction (MPC) of the prior $F(\cdot)$.

\paragraph{Search Cost Heterogeneity.} Following \cite{stahl1989oligopolistic}, consumers are heterogeneous and privately informed about their search costs. A fraction $\alpha\in [0,1]$ of consumers are \inexperienced{} and face a high search cost $s\in(0,\mu)$, while the remaining fraction, $1-\alpha$, are \savvy{} and visit all firms. The two extreme cases, $\alpha=0$ and $\alpha=1$, correspond to a market with costless information acquisition, as in \cite{hwang2019competitive}, and a market with a single strictly positive search cost, as in \cite{au2023attraction}. We revisit these cases in Section \ref{sec_benchmarks}.

\paragraph{Search Process and Timing.} Firms simultaneously design information structures in an initial stage, thereby inducing posterior mean distributions $G_i(\cdot)$. The firms' chosen information structures are not observed (but are anticipated) by consumers before they search, i.e., search is undirected. Subsequently, consumers search sequentially, as in \citet{weitzman1979optimal}. In each stage, a consumer chooses some firm, $i$, to visit. Upon visiting firm $i$, the consumer incurs the privately known search cost. The consumer observes firm $i$'s information design and draws a realization of the posterior mean, $v_i\sim G_i$. If she has not visited all firms, then the consumer decides whether to stop the search or continue. If she stops, then she selects the product with the highest realization among those that have been visited thus far. If she continues, then she selects which firm to visit next. Search continues in this way until the consumer either decides to stop, or until all firms have been visited. Once a consumer visits all firms, she chooses the product with the highest posterior mean realization. 

 \begin{remark} Our modeling of the timing and search process, whereby firms simultaneously design their information structures ahead of consumers' undirected, sequential search, is motivated by the examples in the introduction, and it mirrors the framework in \citet{stahl1989oligopolistic} and the broader price-search literature \citep{EW2012, DDZ2017, zhou2022improved}.\footnote{For example, it is plausible that movie trailers or software demos are designed before product launch, with minimal or no  revision for a period of time thereafter.} Provided firms have limited information about consumers' search history, this timing can be relaxed with no changes to the equilibrium characterization. Suppose that each firm designs a consumer's information structure upon her visit. It does so knowing only that it has been visited. In particular, a firm does not observe how many rivals the consumer has visited, and additional hints (e.g., calendar time or others' past visits) are unobserved or uninformative.\footnote{For example, in the limited record  framework of \citet{KR2022} (record length 0), duration is unobserved, and the calendar time of a consumer's arrival reveals nothing about her search history or type (or those of others). Alternatively, consider a steady state in which departing consumers are replaced by new arrivals of the same type. The market's stationarity implies that calendar time and the mass of past visits are uninformative. Though replacement changes its calculation, the firm's payoff is qualitatively similar.} In this case, each firm chooses its disclosure strategy to maximize the probability its product is chosen conditional on visit, and its updated belief depends only on the prior and search strategy. As we show below, this setting is strategically equivalent.

 \end{remark}

\paragraph{Equilibrium.} The analysis focuses on a pure-strategy symmetric Perfect Bayesian equilibrium of this game.\footnote{With undirected search, allowing firms to randomize has no effect on the characterization.} In particular: (1) each firm's posterior mean distribution $G_i(\cdot)$ is identical, and it is optimal given that its rivals select $G_i(\cdot)$ and the search strategy of each consumer type; (2) each consumer type's search strategy is optimal given the identical distribution of the posterior mean $G_i(\cdot)$ that each firm is expected to select.

\section{Preliminary Analysis}
\label{Preliminary}

Consider a possible symmetric equilibrium, in which each firm selects the same distribution of posterior means $G(\cdot)$. Given this conjectured strategy, we derive consumers' optimal search strategies and formulate each firm's best response as an information design problem.

\paragraph{Optimal Search Strategy.} While \savvy{} consumers visit all firms, the \inexperienced{} consumer's  problem is identical to the classic ``Pandora's Box'' problem \citep{weitzman1979optimal}, except that the distribution of rewards (posterior means) inside each box (firm) is endogenous. In particular, before visiting, a \inexperienced{} consumer anticipates that a firm will select the equilibrium posterior belief distribution, $G(\cdot)$, and therefore, the consumer's search decision is based on this conjectured distribution of posterior means from unvisited firms. We focus on passive beliefs: a consumer's conjecture about the distribution of posteriors at unvisited firms is not affected if the current firm deviates. In other words, consumers always anticipate $G(\cdot)$ from unvisited firms. Exploiting the equivalence  with the \citet{weitzman1979optimal} search, a \inexperienced{} consumer's optimal search strategy follows ``Pandora's Rule.''  In particular, each firm is assigned an identical reservation value $r$ that depends only on the conjectured distribution of the posterior mean $G(\cdot)$ and the search cost,
\begin{equation}
\int_{r}^{1} (v-r) dG=s.
\label{eq_r_j}
\end{equation} 
Because all firms have the same reservation value, consumers are indifferent among all search orders and randomize uniformly over all possible sequences, so that each firm is equally likely to appear at each position. The \inexperienced{} type searches until she draws a posterior mean that exceeds $r$, or until all firms have been visited, in which case she selects the highest realization.

\paragraph{Visits and Beliefs.} The probability of visit in each market segment is important for calculating a firm's payoff. To calculate this probability, consider firm $i$ and suppose that the $n-1$ other firms select posterior mean distributions $G(\cdot)$ and \inexperienced{} consumers use a symmetric search strategy, with reservation value $r$. Firm $i$ is always visited if the consumer is \savvy{}, but it is visited with probability $\eta$ if the consumer is \inexperienced{}, i.e.,
\begin{equation}   \eta\equiv\Pr(\text{visit}\mid\text{\inexp{}})=\frac{1}{n}(1+G(r)+G(r)^2+...G(r)^{n-1})=\frac{1-G(r)^n}{n(1-G(r))}.
    \label{eta}
\end{equation}
 To see this, note that because the \inexperienced{} consumer randomizes fairly over the search order, firm $i$ appears in each position with probability $1/n$. At position $k$, firm $i$ is visited if the consumer did not stop at any of the previous $k-1$ firms, probability $G(r)^{k-1}$. Accounting for each possible position results in the above summation. Building on this calculation, a firm's posterior belief that a consumer who visits is \inexperienced{} is given by
\begin{align}\label{alphatilde}
\widetilde{\alpha}\equiv\Pr(\text{\inexp{}}\mid\text{visit})=\frac{\alpha\eta}{\alpha\eta+(1-\alpha)}.
\end{align}
Note that $\eta$ and $\widetilde{\alpha}$ depend on the \inexperienced{} consumers' search strategy (via $r$), and on rivals' disclosure strategies (via $G(\cdot)$), but not on a firm's own disclosure strategy.

\paragraph{Firm's Decision.} To formulate firm $i$'s information design problem, consider a possible symmetric equilibrium in which all firms select atomless distribution of posteriors $G(\cdot)$, which is associated with reservation value $r$, visit probability $\eta$, and belief $\widetilde{\alpha}$.\footnote{We treat $G(\cdot)$ as atomless. With a type that searches all firms, this follows from standard tie-breaking arguments (see footnote 25). }  As described above, the consumer randomizes the order of search uniformly, because she anticipates that all firms offer the same distribution of posteriors. Thus, $r$, $\eta$, $\widetilde{\alpha}$, and the \inexperienced{} consumer's search order are not affected if firm $i$ were to deviate from $G(\cdot)$. A fundamental implication comes from decomposing the firm's payoff,
\begin{equation*}   \Pr(\text{chosen})=\Pr(\text{visit})\Pr(\text{chosen}\mid\text{visit})\quad\text{and}\quad \Pr(\text{visit})=\alpha\eta+(1-\alpha).
\end{equation*}
The probability a firm is visited does not depend on its disclosure strategy. Thus, each firm designs its disclosure strategy to maximize the probability that its product is chosen, conditional on visit. Furthermore, once firm $i$ is visited and the posterior mean is realized, only the realization $v$ affects a consumer's choice.

 To formulate the firm's information design problem, we must compute the probability that its product is chosen conditional on visit, for each possible realization of the posterior mean, $v$. As we will see, this probability, $u(v)$, is the firm's payoff function in this problem,
 \begin{align*}
u(v)&\equiv\Pr(\text{chosen}\mid\text{visit}\cap v)\\&=
\widetilde{\alpha}\Pr(\text{chosen}\mid\text{visit}\cap v\cap \text{\inexp{}})+(1-\widetilde{\alpha})\Pr(\text{chosen}\mid\text{visit}\cap v\cap\text{\savvy{}}).
\end{align*}
Because the \savvy{} consumers visit all firms and choose the best realization,
\begin{align*}
\Pr(\text{chosen}\mid\text{visit}\cap v\cap\text{\savvy{}})=G(v)^{n-1},
\end{align*}
i.e., the probability of being chosen by a \savvy{} type is the probability that $v$ is largest.

For the \inexperienced{} type, the situation is more subtle. Note first that if $v\geq r$, then 
\begin{equation*}
\Pr(\text{chosen}\mid\text{visit}\cap v\cap\text{\inexp{}})=1.
\end{equation*}
This observation is an immediate consequence of the \inexperienced{} consumer's search rule: upon visiting a firm, she stops her search immediately if  $v\geq r$. For $v<r$, however, she chooses the firm's product if and only if it has the highest posterior mean realization: the consumer initially leaves, visits all other firms, and then returns. Thus, if $v<r$, then $\Pr(\text{chosen}\cap\text{visit}\mid v\cap\text{\inexp{}})=G(v)^{n-1}$. Furthermore, the \inexperienced{} consumer does not necessarily visit all firms, and the probability a firm is visited is independent of its own $v$:  $\Pr(\text{visit}|v\cap\text{\inexp{}})=\Pr(\text{visit}|\text{\inexp{}})=\eta$. Altogether, with $v<r$, 
\begin{align*}
\Pr(\text{chosen}\mid\text{visit}\cap v\cap\text{\inexp{}})=\frac{\Pr(\text{chosen}\cap\text{visit}\mid v\cap\text{\inexp{}})}{\Pr(\text{visit}\mid v\cap\text{\inexp{}})}=\frac{G(v)^{n-1}}{\eta}.
\end{align*} 

Combining these observations, conditional on firm $i$ being visited and a posterior mean realization $v$, the probability its product is chosen is 
\begin{equation}
u(v)=
\begin{cases}
(\frac{\widetilde{\alpha}}{\eta}+(1-\widetilde{
\alpha}))G(v)^{n-1} & \text{if } v< r\\
\widetilde{\alpha}+(1-\widetilde{\alpha})G(v)^{n-1} & \text{if } v\geq r.
\end{cases}
\label{eq_u}
\end{equation}
The key feature of a firm's payoff function $u(\cdot)$ is the upward jump at $r$, reflecting the discrete change in the search behavior of \inexperienced{} consumers.\footnote{The size of the jump at $v=r$ is $\widetilde{\alpha}(1-G(r)^{n-1}/\eta)$. To see immediately that $G(r)^{n-1}/\eta<1$, note that $\eta$ is the average value of a decaying geometric sequence with last term $G(r)^{n-1}$.} As we highlight throughout, it is this discontinuity that creates incentives for firms to garble information.

Building on the previous logic, given a conjectured symmetric equilibrium in which all firms select $G(\cdot)$ and the \inexperienced{} consumers' reservation value is $r$, firm $i$'s best responses solve the following maximization problem,
\begin{align}
\underset{\widehat{G}_i \in \text{MPC}(F)}{\max} \ \int_0^1 u(v) d\widehat{G}_i.
\label{eq_optim_G}
\end{align}
By implication, $G(\cdot)$ is an equilibrium if and only if it solves the optimization defined by \eqref{eq_u} and \eqref{eq_optim_G}, where $r$ satisfies \eqref{eq_r_j}, $\eta$ satisfies \eqref{eta}, and $\widetilde{\alpha}$ satisfies \eqref{alphatilde}.
Optimality conditions for such optimization problems are provided by \citet{dworczak2019simple}.

\subsection{Benchmarks}
\label{sec_benchmarks}

To glean some insight into the role of the discontinuity at the reservation value, it is instructive to briefly contrast the extremes of frictionless search ($\alpha=0$) and the informational Diamond paradox ($\alpha=1$). The former corresponds to the model studied by \cite{hwang2019competitive} while the latter is discussed in \cite{au2023attraction}.  Past this subsection, we restrict attention to $\alpha\in (0,1)$.

\paragraph{Frictionless Search.} First, consider the case in which all consumers are \savvy{} ($\alpha=0$). 

\begin{lemma}(Frictionless Search, \cite{hwang2019competitive}).
Suppose all consumers are \savvy{}, $\alpha=0$. The unique symmetric equilibrium is full disclosure, $G(\cdot)=F(\cdot)$.
\label{lemma_alpha_0}
\end{lemma}   

When all consumers are \savvy{}, the payoff given by eq. (\ref{eq_u}) is continuous and increasing in $v$. Furthermore, if all other firms fully disclose, then the payoff reduces to $u(v)=F(v)^{n-1}$, which is convex. Consequently, a firm benefits from spreading maximally, and its best response is to disclose fully. The argument for uniqueness is more involved. Roughly speaking, any mean-preserving contraction of $F(\cdot)$ generates a payoff function that has a convex part within the support of the contraction, and the firm can profitably deviate by spreading mass within this region.

\paragraph{Informational Diamond Paradox.} The result above contrasts sharply with the scenario where all consumers are \inexperienced{} ($\alpha=1$).

\begin{lemma}(Informational Diamond Paradox, \citet{W2021,au2023attraction}). Suppose all consumers are \inexperienced{}, $\alpha=1$. In any symmetric equilibrium the reservation value is $r^*=\mu-s$, the equilibrium disclosure $G(\cdot)$ puts zero mass on $v<r^*$, and consumers visit only one firm.
\label{lemma_alpha_1}
\end{lemma}     

When $\alpha=1$, the equilibrium payoff given by eq. (\ref{eq_u}) jumps to $1$ at $r$. If $r$ were an exogenous threshold, the firm's problem would resemble well-known models in the literature, and its optimal disclosure would take a simple form: either partial disclosure of low values and all remaining mass concentrated on $r$, or all mass weakly above $r$ (e.g., no disclosure) \citep{KG2011,dworczak2019simple}. In other words, a firm's optimal disclosure either allocates no mass strictly below $r$, or no mass strictly above. When $r$ is endogenous and governed by the firms' disclosure strategies, active search requires positive mass strictly above and below the reservation value, which is inconsistent with optimal disclosure.\footnote{With a positive search cost, the search equation \eqref{eq_r_j} requires positive mass strictly above the reservation value. Furthermore, for the consumer to search beyond the first firm with positive probability, there must be positive mass strictly below.} In equilibrium, all posterior mean realizations must therefore exceed the reservation value, giving rise to an informational Diamond paradox in which firms provide no useful information, and consumers stop searching after a single draw. 

 \paragraph{Deviations.} 
 
 It is also helpful to see why the equilibrium in our setting departs from these benchmarks. Consider first a severe version of the informational Diamond paradox in which signals are uninformative. In this case, all mass is concentrated at $\mu$, and all firms tie in the \savvy{} market. This gives rise to an obvious deviation: assigning an arbitrarily small mass below $\mu$ allows the firm to shift its mass point slightly higher, effectively breaking all ties in its favor. Thus, the existence of the \savvy{} consumers ensures that some information is revealed in equilibrium, tempering the informational Diamond paradox.

 \begin{figure}
\begin{minipage}{0.5\textwidth}
\begin{tikzpicture}[scale=0.6]

\draw[color=blue, domain=0:4, line width=1pt] 
    plot (\x, {0.1*\x^2});

\draw[color=blue, domain=4:10, line width=1pt] 
    plot (\x, {0.005*\x^3 + 10 - 0.005*10^3});

\draw[line width=0.5pt, black, dotted] 
    (4, {0.005*4^3 + 10 - 0.005*10^3}) -- (4, 0);

\draw[line width=0.5pt] 
    (0, 0) -- (0, 10.5);  

\draw[line width=0.5pt] 
    (0, 0) -- (10.5, 0);  

\fill (0, 0) node[left] {\footnotesize{$0$}};
\fill (0, 10) node[left] {\footnotesize{$1$}};

\fill (4, 0) node[below] {\footnotesize{$r$}};

\fill (8, 9) node[above] {\footnotesize{$u(\cdot)$}};

\end{tikzpicture}\end{minipage}
\begin{minipage}{0.5\textwidth}
\begin{tikzpicture}[scale=0.6]

\draw[color=blue, domain=0:4, line width=1pt] 
    plot (\x, {0.1*\x^2});

\draw[color=blue, domain=4:10, line width=1pt] 
    plot (\x, {0.005*\x^3 + 10 - 0.005*10^3});

\draw[color=red, line width=1pt] 
    (2, {0.1*2^2}) -- (4, {0.005*4^3 + 10 - 0.005*10^3}); 

\draw[color=red, line width=1pt] 
    (4, {0.005*4^3 + 10 - 0.005*10^3}) -- (8, {0.005*8^3 + 10 - 0.005*10^3});

\draw[line width=0.5pt, black, dotted] 
    (2, {0.1*2^2}) -- (2, 0);

\draw[line width=0.5pt, black, dotted] 
    (4, {0.005*4^3 + 10 - 0.005*10^3}) -- (4, 0);

\draw[line width=0.5pt, black, dotted] 
    (8, {0.005*8^3 + 10 - 0.005*10^3}) -- (8, 0);

    \draw[line width=3pt, red, opacity=0.3] (2,0) -- (8,0);
\draw[->] (2.3,-0.4) -- (3.7,-0.4);
\draw[->] (7.7,-0.4) --(4.3,-0.4);

\draw[line width=0.5pt] 
    (0, 0) -- (0, 10.5);  

\draw[line width=0.5pt] 
    (0, 0) -- (10.5, 0);  

\fill (0, 0) node[left] {\footnotesize{$0$}};
\fill (0, 10) node[left] {\footnotesize{$1$}};

\fill (2, 0) node[below] {\footnotesize{$\underline{v}$}};
\fill (4, 0) node[below] {\footnotesize{$r$}};
\fill (8, 0) node[below] {\footnotesize{$\overline{v}$}};

\fill (3, 5) node[below] {\footnotesize{$h(\cdot)$}};
\fill (8, 9) node[above] {\footnotesize{$u(\cdot)$}};

\end{tikzpicture}
\end{minipage}

\caption{\label{fig_h_v}\textit{Deviating From Full Disclosure}. }
\end{figure}
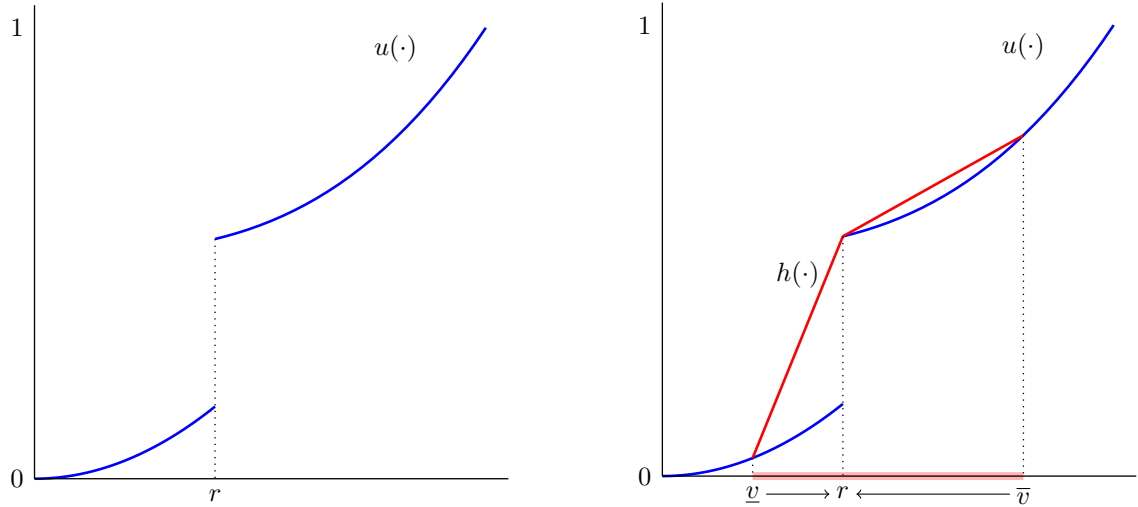

 Next, suppose that all firms fully disclose, as in frictionless search.   
 Then, the payoff function is strictly increasing and piecewise convex, with a jump up at $r$, as illustrated in the left panel of Figure \ref{fig_h_v}. Reallocating mass onto $r$ from an interval around it (while maintaining Bayes plausibility) benefits the firm.  Intuitively, the firm discontinuously increases the probability that \inexperienced{} consumers choose its product at the cost of a relatively small (continuous) distortion in the \savvy{} market segment.\footnote{Our formal argument is similar to the one \citet{W2021} and \citet{HH2025} use to prove the informational Diamond paradox, but it does not follow immediately because our payoff $u(\cdot)$ is not flat above $r$. The construction of the auxiliary function $h(\cdot)$ is also reminiscent of the multiplier (or ``price function'') in \citet{dworczak2019simple}.} 
 
 To see this formally, note that for any $\underline{v}<r$ sufficiently close to $r$, there exists $\overline{v}>r$ such that the conditional expectation $E_F[v|v\in I]=r$, where $I\equiv[\underline{v},\overline{v}]$. Next, construct an auxiliary function $h(\cdot)$ by linearly interpolating $(\underline{v},u(\underline{v})),\,(r,u(r))$ and $(\overline{v},u(\overline{v}))$, as illustrated in Figure \ref{fig_h_v}. By construction, $h(v)>u(v)$ for almost all $v\in I$, which implies $E_F[h(v)|v\in I]>E_F[u(v)|v\in I]$. Furthermore, the slope of $h(\cdot)$ to the right of $r$ is finite, but the slope to the left is arbitrarily large for $\underline{v}$ sufficiently close to $r$, which implies that $h(\cdot)$ is concave. Using Jensen's inequality, $h(r)>E_F[h(v)|v\in I]$. In other words, collapsing all mass inside $I$ to its mean, $r$, increases the expected $h(\cdot)$. By construction $u(r)=h(r)>E_F[h(v)|v\in I]>E_F[u(v)|v\in I]$, implying that such a deviation is beneficial.

\section{Equilibrium Information Disclosure}
\label{Info_disclosure}

\subsection{Exogenous Reservation Value}\label{sec_exog}  As an intermediate step in our analysis, we characterize the equilibrium for exogenous reservation value $r\in(0,1)$. Thus, we relax the restriction that $r$ satisfies \eqref{eq_r_j}, while maintaining all other equilibrium conditions (the firm's payoff $u(\cdot)$ is endogenous).\footnote{This part of our characterization is reminiscent of frictionless search with an exogenous outside option, analyzed in \citet{hwang2019competitive}. Nevertheless, there are crucial differences. In particular, in \citet{hwang2019competitive} the payoff below the reservation value is always 0, while in our model it  is determined endogenously by the equilibrium distribution $G(\cdot)$ (consult (\ref{eq_u})). This endogeneity introduces a number of complications not present in \citet{hwang2019competitive}. In addition, the full model (with endogenous $r$) produces substantially different findings, which we discuss later.} In the next subsection, we build on these results to characterize the equilibrium of the full model.

We now examine the overall shape of equilibrium disclosure. In particular, we show that the equilibrium distribution must conform to a specific structure, which effectively balances the competing incentives originating in the \savvy{} and \inexperienced{} market segments.

\begin{proposition}\label{Gstruc}(Structure of Disclosure).
Suppose a symmetric equilibrium $G$ exists, and let $r\in(0,1)$ be the exogenous reservation value. Values $\beta>0$, $v_L\in [0,r)$, and $v_H\in (r,1]$ exist such that (i) $G\in \text{MPC}(F)$, (ii) $G(v_H)=F(v_H)$, and (iii) $G$ takes the following form: 

\begin{equation}
G(v)=
\begin{cases}
F(v) & \text{for } v\leq  v_L\\
F(v_L) & \text{for } v\in (v_L,r)\\
\min\left\{\left(F(v_L)^{n-1}+\beta(v-r)\right),1\right\}^{\frac{1}{n-1}} & \text{for } v\in[r,v_H]\\
F(v) & \text{for } v\in (v_H,1].
\end{cases}
\label{eq_G}
\end{equation}
\label{prop_G_shape}
\end{proposition}  

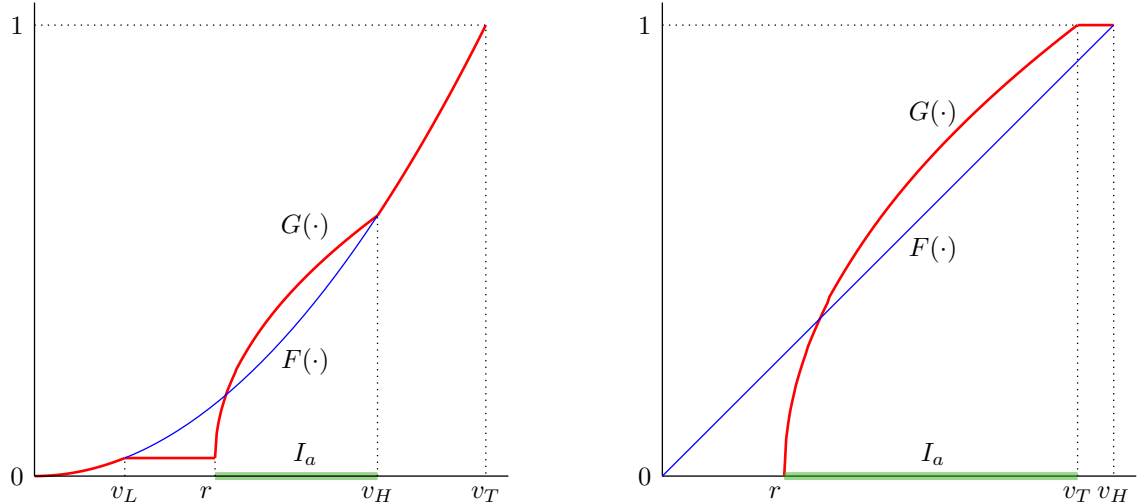
\begin{figure}
\begin{minipage}{0.5\textwidth}
\begin{tikzpicture}[scale=0.6]

\pgfmathsetmacro{\vL}{2}
\pgfmathsetmacro{\r}{4}
\pgfmathsetmacro{\vH}{7.6}

\draw[color=red, domain=0:\vL, line width=1pt] 
    plot (\x, {0.1*\x^2});

\draw[color=red, line width=1pt] 
    (\vL, {0.1*\vL^2}) -- (\r, {0.1*\vL^2});

\draw[color=red, domain=\r:\vH, line width=1pt,samples=100] 
    plot (\x, {0.1*\vL^2 + 0.1*(\vH^2 - \vL^2)*((\x - \r)/(\vH - \r))^0.5});

\draw[color=red, domain=\vH:10, line width=1pt] 
    plot (\x, {0.1*\x^2});

\draw[color=blue, domain=\vL:\vH, line width=0.5pt] 
    plot (\x, {0.1*\x^2});

\draw[line width=3pt, green, opacity=0.5] 
    (\r, 0) -- (\vH, 0);

\draw[line width=0.5pt] (0,0) -- (0,10.5);     
\draw[line width=0.5pt] (0,0) -- (10.5,0);     
\draw[line width=0.5pt, dotted] (0,10) -- (10,10) -- (10,0);  

\draw[line width=0.5pt, dotted] (\r,0) -- (\r, {0.1*\vL^2});
\draw[line width=0.5pt, dotted] (\vL,0) -- (\vL, {0.1*\vL^2});
\draw[line width=0.5pt, dotted] (\vH,0) -- (\vH, {0.1*\vH^2});

\fill (\r-0.2,0) node[below] {\footnotesize{$r$}};
\fill (\vL,0) node[below] {\footnotesize{$v_L$}};
\fill (\vH,0) node[below] {\footnotesize{$v_H$}};
\fill (10,0) node[below] {\footnotesize{$v_T$}};
\fill (0,0) node[left] {\footnotesize{$0$}};
\fill (0,10) node[left] {\footnotesize{$1$}};
\fill (6,0) node[above] {\footnotesize{$I_{a}$}};
\fill (6,5) node[above] {\footnotesize{$G(\cdot)$}};
\fill (6,2) node[above] {\footnotesize{$F(\cdot)$}};

\end{tikzpicture}
\end{minipage}
\begin{minipage}{0.5\textwidth}
\begin{tikzpicture}[scale=0.6]

\pgfmathsetmacro{\vL}{0}
 \pgfmathsetmacro{\r}{2.7}
 \pgfmathsetmacro{\a}{1}

\pgfmathsetmacro{\vH}{10.45}

\draw[line width=0.5pt] (0,0) -- (0,10.5);     
\draw[line width=0.5pt] (0,0) -- (10.5,0);     
\draw[line width=0.5pt, dotted] (0,10) -- (10,10) -- (10,0);  

\draw[line width=0.5pt, dotted] (\r,0) -- (\r, {0.1*\vL^2});
\draw[line width=0.5pt, dotted] (\vL,0) -- (\vL, {0.1*\vL^2});
 \draw[line width=0.5pt, dotted] (9.2,0) -- (9.2,10);

\draw[color=red, domain=\r:9.2, line width=1pt,samples=150] 
    plot (\x, {0.1*\vL^2 + 0.1*(\vH^2 - \vL^2)*((\x - \r)/(\vH - \r))^0.5});

\draw[color=red, line width=1pt, smooth] 
    (9.2,10) -- (10,10);

\draw[color=blue, domain=0:10, line width=0.5pt] 
    plot (\x, {10*(0.1*\x)^\a});

\draw[line width=3pt, green, opacity=0.5] 
    (\r, 0) -- (9.2, 0);

\fill (\r-0.2,0) node[below] {\footnotesize{$r$}};
\fill (9.2,0) node[below] {\footnotesize{$v_T$}};
\fill (10,0) node[below] {\footnotesize{$v_H$}};
\fill (0,0) node[left] {\footnotesize{$0$}};
\fill (0,10) node[left] {\footnotesize{$1$}};
\fill (6,0) node[above] {\footnotesize{$I_{a}$}};
 \fill (6,7.5) node[above] {\footnotesize{$G(\cdot)$}};
 \fill (6,4.5) node[above] {\footnotesize{$F(\cdot)$}};

\end{tikzpicture}
\end{minipage}

\caption{\label{fig_G_star}\textit{Structure of Disclosure.} $G(\cdot)$ red, $F(\cdot)$ blue. Left panel: Disclosure at top and bottom ($v_L>0$, $v_H<1$). Right panel: no disclosure at the top or bottom ($v_L=0$, $v_H=1$). }
\end{figure}

Proposition \ref{Gstruc} reveals that the distortion from full disclosure always manifests in an interval $(v_L,v_H)$ around the reservation value. Furthermore, this distortion may extend to the extremes ($v_L=0$ and $v_H=1$ are possible), but if it does not, then sufficiently high or low valuations are disclosed accurately. Qualitatively, a firm never wants a consumer who continues searching to be near indifference, as a small distortion to its disclosure strategy would induce the consumer to stop. Consequently, equilibrium disclosures feature a gap below the reservation value and, if $v_L>0$, a bifurcated distribution of posterior means.  

Figure \ref{fig_G_star} depicts the two polar cases of disclosure at both extremes and no disclosure of extreme values. Examining it gives a better understanding of the required structure. From \eqref{eq_G}, an interval of valuations below $v_L$ is disclosed accurately to the consumer. Thus, we refer to $v_L$ as the \textit{lower disclosure threshold}. With $v_L=0$, no such interval exists, and there is \textit{no disclosure at the bottom}. The interpretation of the second threshold $v_H$ is more nuanced. Formally, $v_H$ is the second point above the reservation value at which $G(\cdot)$ and $F(\cdot)$ intersect. When $v_H<1$ (as in the left panel), all values above $v_H$ are disclosed truthfully. In contrast, when $v_H=1$ (as in the right panel), no values above $r$ are disclosed truthfully. In either case, we denote the top of the support of $G(\cdot)$ by $v_T$. Next, notice that the equilibrium structure requires that $G(\cdot)^{n-1}$ is affine over interval $I_a\equiv [r,\min\{v_H,v_T\}]$, which we refer to as the \textit{affine interval}. By implication, $G(\cdot)$ is concave over $I_a$. The slope of $G(\cdot)^{n-1}$ over the affine interval is denoted $\beta$.

\begin{figure}
\begin{minipage}{0.5\textwidth}
\begin{tikzpicture}[scale=0.6]

\pgfmathsetmacro{\al}{0.1}
\pgfmathsetmacro{\et}{0.4}
\pgfmathsetmacro{\vL}{2}
\pgfmathsetmacro{\r}{4}
\pgfmathsetmacro{\vH}{7.6}

\draw[line width=0.5pt] (0,0) -- (0,10.5);
\draw[line width=0.5pt] (0,0) -- (10.5,0);
\draw[dotted, line width=0.5pt] (0,10) -- (10,10) -- (10,0);

\draw[dotted, line width=0.5pt] (\vL,0) -- (\vL,{0.1*\vL^2});
\draw[dotted, line width=0.5pt] (\r,0) -- (\r,{0.1*\vL^2});
\draw[dotted, line width=0.5pt] (\vH,0) -- (\vH,{0.1*\vH^2});

\draw[line width=3pt, green, opacity=0.5] (\r,0) -- (\vH,0);
\fill (6,0) node[above] {\footnotesize{$I_{a}$}};

\draw[teal, line width=1pt] (\vL,{0.1*(\al/\et + 1 - \al)*\vL^2}) -- (\r,{0.1*(\al/\et + 1 - \al)*\vL^2});
\draw[teal, line width=1pt, dashed] 
  (\r,{0.1*(\al/\et + 1 - \al)*\vL^2}) -- 
  (\r,{0.1*(\al/\et + 1 - \al)*\vL^2 + 
        (10*\al + 10*(1-\al)*(0.1*\vH)^2 - 0.1*(\al/\et + 1 - \al)*\vL^2)/(\vH - \vL)*(\r - \vL)});

\draw[teal, line width=1pt, domain=0:\vL] 
  plot (\x,{0.1*(\al/\et + 1 - \al)*\x^2});

\draw[teal, line width=1pt, domain=\r:\vH] 
  plot (\x,{0.1*(\al/\et + 1 - \al)*\vL^2 + 
            (10*\al + 10*(1 - \al)*(0.1*\vH)^2 - 0.1*(\al/\et + 1 - \al)*\vL^2)/(\vH - \vL)*(\x - \vL)});

\draw[teal, line width=1pt, domain=\vH:10] 
  plot (\x,{10*\al + 10*(1 - \al)*(0.1*\x)^2});

\draw[decorate, decoration={brace, amplitude=10pt}] 
  (\r,{
    0.1*(\al/\et + 1 - \al)*\vL^2 +
    (10*\al + 10*(1 - \al)*(0.1*\vH)^2 - 0.1*(\al/\et + 1 - \al)*\vL^2)/(\vH - \vL)*(\r - \vL)
  }) -- (\r,{0.1*\vL^2});

\fill (\r,0) node[below] {\footnotesize{$r$}};
\fill (\vL,0) node[below] {\footnotesize{$v_L$}};
\fill (\vH,0) node[below] {\footnotesize{$v_H$}};
\fill (0,0) node[left] {\footnotesize{$0$}};
\fill (0,10) node[left] {\footnotesize{$1$}};
\fill (5,1) node[above] {\footnotesize{$J$}};
 \fill (8.5,8.2) node[above] {\footnotesize{$u(\cdot)$}};
\node[rotate=45] at (5.9,3.7) {\footnotesize{slope $(1-\widetilde{\alpha})\beta$}};

\end{tikzpicture}
\end{minipage}
\begin{minipage}{0.5\textwidth}
\begin{tikzpicture}[scale=0.6]

\pgfmathsetmacro{\vL}{0}
\pgfmathsetmacro{\r}{2.7}
\pgfmathsetmacro{\vH}{11}
\pgfmathsetmacro{\vT}{9.2}
\pgfmathsetmacro{\h}{7}

\draw[line width=0.5pt] (0,0) -- (0,10.5);
\draw[line width=0.5pt] (0,0) -- (10.5,0);
\draw[dotted, line width=0.5pt] (0,10) -- (10,10) -- (10,0);

\draw[dotted, line width=0.5pt] (\vL,0) -- (\vL,{0.1*\vL^2});
\draw[dotted, line width=0.5pt] (9.2,0) -- (9.2,10);

\draw[teal, line width=1.5pt] (0,0) -- (\r,0);
\draw[teal, line width=1pt] 
  [domain=\r:\vT] plot (\x,{\h + (\x - \r)*(10 - \h)/(\vT - \r)});
\draw[teal, line width=1pt, dashed] (\r,0) -- (\r,\h);
\draw[teal, line width=1pt] (\vT,10) -- (10,10);

\draw[line width=3pt, green, opacity=0.5] (\r,0) -- (9.2,0);
\fill (7,0) node[above] {\footnotesize{$I_{a}$}};

\draw[decorate, decoration={brace, amplitude=10pt}] 
  (\r,\h) -- (\r,0) node[midway,xshift=15pt] {\footnotesize{$\widetilde{\alpha}$}};

\fill (\r,0) node[below] {\footnotesize{$r$}};
\fill (9.2,0) node[below] {\footnotesize{$v_T$}};
\fill (10,0) node[below] {\footnotesize{$v_H$}};
\fill (0,0) node[left] {\footnotesize{$0$}};
\fill (0,10) node[left] {\footnotesize{$1$}};
 \fill (5.2,8.2) node[above] {\footnotesize{$u(\cdot)$}};

\node[rotate=25] at (5.5,7.5) {\footnotesize{slope $(1-\widetilde{\alpha})\beta$}};

\end{tikzpicture}
\end{minipage}

\caption{\label{fig_payoff}\textit{Payoff function}. Payoff teal. Left panel: Disclosure at top and bottom ($v_L>0$, $v_H<1$). Right panel: no disclosure at the top or bottom ($v_L=0$, $v_H=1$). In the left panel, the jump $J = \widetilde{\alpha}(1 - F(v_L)^{n-1}/\eta) > 0$.}
\end{figure}
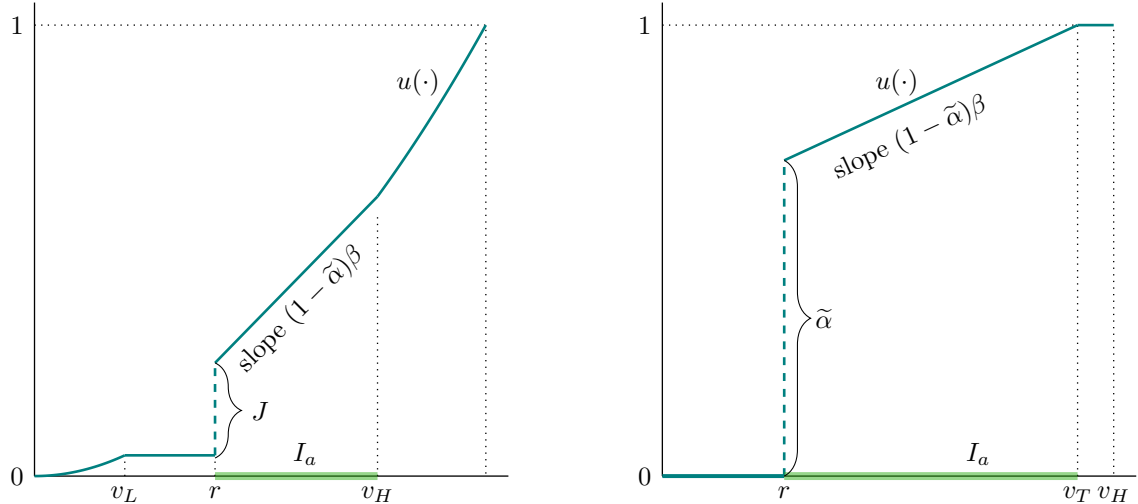

As discussed above, the presence of \inexperienced{} consumers causes firms to depart from full disclosure around the reservation value $r$. Proposition \ref{Gstruc} states that this departure has to take a specific form. An interval of valuations just below the reservation value must be pooled with values above it, which creates an interval $(v_L,r)$ on which the equilibrium distribution (and the endogenous payoff) is flat. In an interval above the reservation value ($I_a$) the (true) valuations from this gap are absorbed, distorting the posterior mean. In equilibrium, $u(\cdot)$ must be affine on $I_a$, which leaves the firm indifferent among all mean preserving shifts of mass within this interval. In fact, a non-affine payoff cannot be sustained in equilibrium for any region of partial disclosure: strict convexity would incentivize greater disclosure, while strict concavity would incentivize less. An affine $u(\cdot)$, in turn, requires that $G(\cdot)^{n-1}$ is affine on $I_a$, explaining the functional form of $G(\cdot)$ in this interval. These connections between the equilibrium disclosure strategies and the equilibrium payoffs are illustrated in Figure \ref{fig_payoff}, which shows the payoffs associated with the disclosure strategies in Figure \ref{fig_G_star}.

Proposition \ref{Gstruc} establishes necessary conditions for an equilibrium. To complete the characterization, we identify disclosure strategies that are mutual best responses. In particular, the competitors' strategies determine a firm's endogenous payoff $u(\cdot)$, not only through the term $G(\cdot)^{n-1}$, but also through $\eta$ and $\widetilde{\alpha}$ (see \eqref{eta}, \eqref{alphatilde},  and \eqref{eq_u}). To be an equilibrium, a disclosure strategy must be optimal for each firm, given its associated payoff function.  Next, we establish the existence and uniqueness of an equilibrium $G(\cdot)$ and describe its key properties.

\begin{proposition}(Exogenous Reservation Value). Consider the game with exogenous reservation value $r\in(0,1)$. 
\begin{itemize}
\item For each $r$, the equilibrium exists and is unique. In particular, the equilibrium strategy takes the form of \eqref{eq_G}, with unique values $\{v^{eq}_L(r),v^{eq}_H(r),\beta^{eq}(r)\}$. 

\item A unique $\underline{r}(n,\alpha)\in(0,\mu)$ exists such that 
\begin{enumerate}
\item[(i)] for $r\leq \underline{r}$, the equilibrium features no disclosure at the bottom ($v^{eq}_L(r)=0$)
 \item[(ii)] for $r>\underline{r}$ the equilibrium features disclosure at the bottom ($0<v_L^{eq}(r)<r$),  and the lower disclosure threshold $v^{eq}_L(\cdot)$ is increasing in $r$.
\end{enumerate}
\item Threshold $\underline{r}(n,\alpha)\in(0,\mu)$ is weakly increasing in $n$, and it approaches $\mu$ as $n\rightarrow \infty$. 
\item As $r$ approaches zero or one, the equilibrium converges to full disclosure in distribution ($G(\cdot) \rightarrow F(\cdot)$ pointwise). 

\end{itemize}
\label{prop_r_exog}
\end{proposition}

In addition to establishing equilibrium existence and uniqueness with exogenous $r$, this proposition links the structure of equilibrium disclosure to the reservation value: disclosure at the bottom arises if and only if it is sufficiently large ($r>\underline{r}$). To understand this link, it is useful to distinguish two cases. First, consider $r>\mu$. In this case, Bayes plausibility alone requires that the equilibrium disclosure has positive mass below the reservation value, and Proposition \ref{Gstruc} implies that this mass must arise from disclosure at the bottom.\footnote{Bayes plausibility requires that $G\in\text{MPC}(F)$, which implies that $G$ assigns positive mass weakly below its mean $\mu$, which itself is smaller than $r$.} Thus, when $r>\mu$, disclosure at the bottom is an inherent feature of competitive information provision. 

Second, consider $r<\mu$, so that a distribution with no mass below $r$ is consistent with Bayes plausibility. To see why disclosure at the bottom nevertheless arises, suppose all other firms use identical disclosures consistent with the equilibrium structure (Proposition \ref{Gstruc}),  with no mass below the reservation value ($v_L=0$). When the reservation value is low ($r\leq \underline{r}$), the prior mass below it is small, and shifting all of it into the affine interval generates only a small distortion. In this case, responding with a strategy that has no disclosure at the bottom allows a firm to sell to all \inexperienced{} consumers who visit without substantially reducing its competitiveness in the \savvy{} market, and doing so is a best response. When the reservation value is larger ($r\in(\underline{r},\mu)$), however, such a shift requires that a substantial amount of mass is absorbed above $r$, generating a large distortion that weakens the firm in the \savvy{} market. By deviating and disclosing the smallest, most damaging valuations truthfully, the firm relaxes the Bayes plausibility constraint, which allows it to allocate mass more aggressively above the reservation value. Although such low valuations never lead to its product being chosen directly, these losses are offset by additional gains in the \savvy{} market, resulting in a beneficial deviation. In other words, for $r\in(\underline{r},\mu)$ disclosure at the bottom arises from the interplay of Bayes plausibility and competition for \savvy{} consumers.

The preceding logic also suggests that when $r<\mu$ and competition for \savvy{} consumers is intense, a firm has less incentive to disclose low values. As discussed above, if others do not disclose at the bottom, a firm might wish to deviate in order to relax Bayes plausibility and gain an advantage in the \savvy{} market. If many firms are competing for \savvy{} consumers, however, then it is unlikely that the firm will win them over, even if it gains such an advantage. Thus, there is little to gain from such a deviation. In other words, when competition for \savvy{} consumers is fierce, it is optimal for a firm to ensure that all \inexperienced{} consumers that visit choose its product, resulting in an equilibrium with no disclosure at the bottom. This intuition underlies the third point of Proposition \ref{prop_r_exog}, which implies that for any exogenous $r\in(0,\mu)$, the unique equilibrium of the benchmark has no disclosure at the bottom if the number of firms is sufficiently large.

Proposition \ref{prop_r_exog} also reveals a non-monotonic relationship between the reservation value and disclosure. In particular, as the reservation value approaches the extremes of $0$ and $1$, the equilibrium converges to full disclosure, as in frictionless search. As described above, with $r$ approaching 0, all probability mass below the reservation value is pooled with higher valuations, and there is no disclosure at the bottom. But because there is not much probability mass below $r$ in the prior distribution of values, the distortion caused by this pooling is small. In other words, when $r$ is small, it is virtually costless for the firm to deter \inexperienced{} consumers' search. Thus, the firm introduces only a small distortion to capture the \inexperienced{} consumer, and otherwise focuses on competing for \savvy{} consumers by fully disclosing their values. When $r$ is high, on the other hand, there is little mass above the reservation value with which to pool lower valuations. Thus, there is relatively little scope to manipulate the consumer's posterior mean. In other words, when $r$ is close to 1, it is almost impossible for the firm to deter search, and the equilibrium therefore approaches full disclosure, as in frictionless search. This non-monotonicity suggests that when the reservation value is endogenous, with extreme search costs the equilibrium approaches full disclosure, while intermediate search costs limit equilibrium disclosure.

\paragraph{Equilibrium Characterization.} For the interested reader, we provide additional details behind the proof of Proposition \ref{prop_r_exog}. The three requirements of Proposition \ref{Gstruc} restrict the mutually admissible values of $v_L$, $v_H$, and $\beta$. For example, a distribution that satisfies conditions (ii) and (iii) of Proposition \ref{Gstruc}, may not be a mean preserving contraction of $F(\cdot)$, thereby failing (i). When a distribution satisfies all three conditions of Proposition \ref{Gstruc}, we say it is an \textit{equilibrium candidate}. The following lemma gives necessary and sufficient conditions for the existence of an equilibrium candidate with lower disclosure threshold $v_L$.

\begin{lemma}\label{feas}(Equilibrium Candidates.) For each $\{v_L,r\}$, a distribution $G(\cdot)$ that satisfies conditions (i)-(iii) of Proposition \ref{Gstruc} exists if and only if $E_F[v|v>v_L]>r$. If it exists, such $G(\cdot)$ is unique, with slope

\begin{align}
\beta^*=\frac{E_F[F(v)^{n-1}|v\in(v_L, v_H)]-F(v_L)^{n-1}}{E_F[v|v\in(v_L,v_H)]-r}.
\label{eq_BET}
 \end{align}
\end{lemma}

  Lemma \ref{feas} provides insight into a link between $r$ and $v_L$ imposed solely by the equilibrium structure characterized in Proposition \ref{Gstruc}. First, note that if $r$ and $v_L$ are too far apart ($r-v_L$ large, $v_L$ low), then the conditional mean in the denominator of \eqref{eq_BET} places too much weight on low realizations: even with $v_H=1$, the denominator is negative. In other words, even an infinite slope ($\beta=\infty$),  does not allocate enough mass above $r$ for the distribution $G(\cdot)$ to match the mean of $F(\cdot)$ in the partial disclosure region. On the other hand, whenever $E_F[v|v>v_L]>r$, the required structure leads to a unique equilibrium candidate. Indeed, the slope $\beta^*$ and the requirement $G(v_H)=F(v_H)$ uniquely determine $v_H$.

Having derived conditions that ensure existence of a unique equilibrium candidate for each $v_L$, we turn to optimality. We apply the characterization in \citet{dworczak2019simple}, who study the decision problem in \eqref{eq_optim_G} for exogenous payoff function $u(\cdot)$.

\medskip
\noindent\textit{Optimality Conditions \citep{dworczak2019simple}. }{\itshape 
Consider problem \eqref{eq_optim_G} for an exogenous payoff function $u(\cdot)$, which satisfies  regularity conditions.\footnote{$u(\cdot)$ is upper semicontinuous with at most finitely many one-sided jump discontinuities at interior points, and has bounded slope in each interval between them.} $\widehat{G}_i\in\text{MPC}(F)$ is a solution if and only if an auxiliary function $\phi(\cdot)$ exists such that: (DM1) $\phi(\cdot)$ is continuous and weakly convex, (DM2) $\phi(v)\geq u(v)$ for all $v\in[0,1]$, (DM3) $\text{support}(\widehat{G}_i)\subseteq \{v\in [0,1]: u(v)=\phi(v)\}$, (DM4) $\int_0^1 \phi(v)d\widehat{G}_i=\int_0^1 \phi(v)dF$.}
\medskip

\noindent If $(G(\cdot),\phi(\cdot))$ satisfy these optimality conditions then we refer to $\phi(\cdot)$ as a \textit{multiplier that supports} $G(\cdot)$.  In the proof of Proposition \ref{Gstruc}, we show that any candidate $G(\cdot)$ generates a regular $u(\cdot)$, and thus existence of a multiplier is necessary and sufficient for optimality.

\begin{figure}
\begin{minipage}{0.5\textwidth}
\begin{tikzpicture}[scale=0.6]

\pgfmathsetmacro{\al}{0.1}
\pgfmathsetmacro{\et}{0.4}
\pgfmathsetmacro{\vL}{2}
\pgfmathsetmacro{\r}{4}
\pgfmathsetmacro{\vH}{7.6}

\fill (\r,0) node[below] {\footnotesize{$r$}};
\fill (\vL,0) node[below] {\footnotesize{$v_L$}};
\fill (\vH,0) node[below] {\footnotesize{$v_H$}};

\draw[dotted, line width=0.5pt] (\r,0)--(\r,0.1*\vL^2);
\draw[dotted, line width=0.5pt] (\vL,0)--(\vL,0.1*\vL^2);
\draw[dotted, line width=0.5pt] (\vH,0)--(\vH,0.1*\vH^2);

\draw[line width=1pt, teal] (\vL,{0.1*(\al/\et + 1 - \al)*\vL^2})--(\r,{0.1*(\al/\et + 1 - \al)*\vL^2});
\draw[line width=1pt, teal, dashed] 
  (\r,{0.1*(\al/\et + 1 - \al)*\vL^2}) --
  (\r,{0.1*(\al/\et+1-\al)*\vL^2+(10*\al+10*(1-\al)*(0.1*\vH)^2 - 0.1*(\al/\et+1-\al)*\vL^2)/(\vH-\vL)*(\r-\vL)});

\draw[decorate, decoration={brace,amplitude=10pt}]
  (\r,{0.1*(\al/\et+1-\al)*\vL^2+(10*\al+10*(1-\al)*(0.1*\vH)^2 - 0.1*(\al/\et+1-\al)*\vL^2)/(\vH-\vL)*(\r-\vL)})
  -- (\r,0.1*\vL^2);

\draw[violet, line width=1pt, domain=0:\vL] 
  plot (\x,{0.1*(\al/\et+1-\al)*\x^2});
\draw[violet, line width=1pt, domain=\vL:\vH] 
  plot (\x,{0.1*(\al/\et+1-\al)*\vL^2 + (10*\al + 10*(1-\al)*(0.1*\vH)^2 - 0.1*(\al/\et+1-\al)*\vL^2)/(\vH-\vL)*(\x-\vL)});
\draw[violet, line width=1pt, domain=\vH:10] 
  plot (\x,{10*\al + 10*(1-\al)*(0.1*\x)^2});

\draw[line width=0.5pt] (0,0) -- (0,10.5); 
\draw[line width=0.5pt] (0,0) -- (10.5,0); 
\draw[dotted, line width=0.5pt] (0,10) -- (10,10) -- (10,0); 

\fill (0,0) node[left] {\footnotesize{$0$}};
\fill (0,10) node[left] {\footnotesize{$1$}};
\fill (5,1) node[above] {\footnotesize{$J$}};
 \fill (3.3,0.4) node[above] {\footnotesize{$u(\cdot)$}};
 \fill (8,8.2) node[above] {\footnotesize{$\phi(\cdot)$}};
\fill (6,0) node[above] {\footnotesize{$I_{a}$}};

\draw[green, line width=3pt, opacity=0.5] (\r,0)--(\vH,0);

\end{tikzpicture}
\end{minipage}
\begin{minipage}{0.5\textwidth}
\begin{tikzpicture}[scale=0.6]

\pgfmathsetmacro{\vL}{0}
\pgfmathsetmacro{\r}{2.7}
\pgfmathsetmacro{\vH}{11}
\pgfmathsetmacro{\vT}{9.2}
\pgfmathsetmacro{\h}{7}

\fill (\r,0) node[below] {\footnotesize{$r$}};
\fill (9.2,0) node[below] {\footnotesize{$v_{T}$}};
\fill (10,0) node[below] {\footnotesize{$v_H$}};

\draw[violet, line width=1pt, domain=0:\vT+0.5]
  plot (\x,{\h + (\x-\r)*(10-\h)/(\vT-\r)});

\draw[decorate, decoration={brace,amplitude=10pt}]
  (\r,\h)--(\r,0) node[midway,xshift=15pt]{\footnotesize{$\widetilde{\alpha}$}};

\draw[line width=0.5pt] (0,0) -- (0,10.5); 
\draw[line width=0.5pt] (0,0) -- (10.5,0); 
\draw[dotted, line width=0.5pt] (0,10) -- (10,10) -- (10,0); 
\draw[dotted, line width=0.5pt] (9.2,0) -- (9.2,10); 

\draw[teal, line width=1pt, dashed] (\r,0)--(\r,\h);
\draw[teal, line width=1pt] (\vT,10)--(10,10);
\draw[teal, line width=1.5pt] (\r,0)--(0,0);
\fill (0,0) node[left] {\footnotesize{$0$}};
\fill (0,10) node[left] {\footnotesize{$1$}};
 \fill (1.5,0) node[above] {\footnotesize{$u(\cdot)$}};
\fill (5.5,8.5) node[above] {\footnotesize{$\phi(\cdot)$}};
\fill (7,0) node[above] {\footnotesize{$I_{a}$}};

\draw[green, line width=3pt, opacity=0.5] (\r,0)--(9.2,0);

\end{tikzpicture}
\end{minipage}

\caption{\label{fig_phi}\textit{Multiplier}. $\phi(\cdot)$ in violet, $u(\cdot)$ in teal. 
Left panel: Disclosure at top and bottom ($v_L>0$, $v_H<1$). 
Right panel: no disclosure at the top or bottom ($v_L=0$, $v_H=1$). 
In left panel, multiplier and payoff coincide on $[0,v_L]\cup[r,1]$; 
in right panel, on $I_a=[r,v_T]$. 
In the left panel, the jump $J = \widetilde{\alpha}(1-F(v_L)^{n-1}/\eta)>0$.}
\end{figure}

Proposition \ref{Gstruc} imposes a rigid structure on an equilibrium candidate, and the optimality conditions (DM1-4) require that a multiplier supporting it inherit some of its key features. In particular, given a $G(\cdot)$ which satisfies the conditions of Proposition \ref{Gstruc}, let
\begin{align}\label{eq_PHI}
    \phi(v)\equiv
    \begin{cases}
(\frac{\widetilde{\alpha}}{\eta}+(1-\widetilde{
\alpha}))F(v)^{n-1} & \text{if } v<v_L\\
\widetilde{\alpha}+(1-\widetilde{\alpha})(F(v_L)^{n-1}+\beta(v-r)) & \text{if } v\in [v_L,v_H]\\
\widetilde{\alpha}+(1-\widetilde{\alpha})F(v)^{n-1} & \text{if } v\in (v_H,1]\end{cases}
\end{align}
where $\eta$ is from \eqref{eta} and $\widetilde{\alpha}$ from \eqref{alphatilde}.
The optimality conditions require that the multiplier must coincide with the (endogenous) convex payoff up to $v_L$ and past $v_H$ (if $v_H<1$). The multiplier must also coincide with the payoff function over interval $I_a$, where both $G(\cdot)^{n-1}$ and the payoff function are affine (all follow from DM3). In fact, we show that the multiplier must be affine (with no kinks) over the entire interval of partial disclosure, $[v_L,v_H]$.\footnote{This follows from convexity and (DM4). In particular, recall that the multiplier must be weakly convex on $[v_L,v_H]$. Furthermore, because $G\in\text{MPC}(F)$ and they coincide outside the interval in question, $E_G[v|v\in(v_L,v_H)]=E_F[v|v\in(v_L,v_H)]$. Moreover, because $\phi(\cdot)$ is convex, and the integrated difference of cdfs is strictly positive in the interior of the interval, we have $\int_{v_L}^{v_H}\phi(v)dF\geq \int_{v_L}^{v_H}\phi(v)dG$, with equality if and only if $\phi(\cdot)$ is affine. Equality of these integrals is required by (DM4).} Together, these conditions characterize a unique multiplier that could support a candidate distribution $G(\cdot)$ as an equilibrium. Furthermore, with the structure inherited from $G(\cdot)$, it satisfies nearly all of the remaining optimality conditions. In particular, the next lemma shows that the candidate $\phi(\cdot)$ in \eqref{eq_PHI} can fail to support the equilibrium candidate $G(\cdot)$ in only two ways.

\begin{lemma}(Equilibrium Conditions).\label{optdu1} Given an equilibrium candidate $G(\cdot)$ and the associated multiplier $\phi(\cdot)$ in \eqref{eq_PHI}:
 (i) If $G(\cdot)$ has no disclosure at the bottom ($v_L=0$), then $G(\cdot)$ is an equilibrium if and only if $\phi(0)\geq 0$. (ii) If $G(\cdot)$ has disclosure at the bottom ($v_L>0$), then $G(\cdot)$ is an equilibrium if and only if $\phi(\cdot)$ is continuous  and convex at  $v_L$.

\end{lemma}

 To see how these failures can occur, consider a candidate distribution with no disclosure at the bottom, $v_L=0$. The associated payoff function is 0 on $[0,r)$ (see right panel of Figure \ref{fig_phi}). If the slope of the payoff is relatively low on $I_a$, then so is the slope of the multiplier. In this case, it does not cross the horizontal axis, i.e., $\phi(0)\geq 0$, and it is above the payoff on the entire interval $[0,r)$ as required by (DM2) (see right panel of Figure \ref{fig_phi}). However, if the multiplier is too steep, then its horizontal intercept is strictly positive, leading to a violation of (DM2) at valuations below the intercept. With disclosure at the bottom, a failure can also occur at the lower threshold $v_L$, though it takes on a more complex form: rather than just a requirement on the sign of $\phi(0)$, the optimality conditions must be strengthened to require continuity and convexity at $v_L$. Lemma \ref{optdu1} shows that these are the only possible failures of (DM1-4), while Proposition \ref{prop_r_exog} shows that the equilibrium conditions identified in Lemma \ref{optdu1} always have a unique solution. The unique threshold $\underline{r}(n,\alpha)$ defines the boundary between these cases: it is the exogenous reservation value at which (i) the equilibrium has no disclosure at the bottom, and (ii) the multiplier crosses the horizontal axis at $0$.

\subsection{Endogenous Reservation Value}\label{sec_endog}
Building on the results of the previous section, we characterize the equilibrium with endogenous reservation value and its properties. To do so, we must impose the remaining equilibrium condition \eqref{eq_r_j}, which links the disclosure strategy used by firms and the \inexperienced{} type's reservation value. Thus, the characterization of the previous section links an exogenous $r$ to the equilibrium distribution $G_r(\cdot)$, while the search equation links an exogenous distribution $G(\cdot)$ to the \inexperienced{} type's endogenous reservation value. To complete the characterization, we seek a fixed point where the equilibrium distribution $G_r(\cdot)$ is consistent (via \eqref{eq_r_j}) with the reservation value it induces.

In principle, finding such a fixed point could be complicated, because the endogenous reservation value depends (in general) on the entire distribution $G_r(\cdot)$. However, exploiting the fact that $G_r(\cdot)$ is an equilibrium with exogenous $r$, the search equation simplifies significantly. Exploiting the equilibrium structure (Proposition \ref{Gstruc}), the search equation becomes

\begin{align}\label{eq_search2}
    \int_{v_L}^1(v-r)dF=s.
\end{align}
Formally, this simplification exploits three features of the equilibrium structure: (i) $F(\cdot)$ and $G_r(\cdot)$ have the same mean, (ii) $G_r(\cdot)=F(\cdot)$ below $v_L$, and (iii) $G_r(\cdot)$ has no mass in $[v_L,r]$.\footnote{Consider the left hand side of \eqref{eq_r_j}, $\int_r^{1}(v-r)dG_r=\int_0^{1}vdG_r-\int_0^{r}vdG_r-(1-G_r(r))r$. Using (i), $\int_0^{1}vdG_r=\int_0^{1}vdF$. Using (ii) and (iii), $\int_0^{r}vdG_r=\int_0^{v_L}vdG_r=\int_0^{v_L}vdF$. Furthermore, (ii) implies, $G_r(r)=F(v_L)$. Substituting, we have $\int_r^{1}(v-r)dG_r=\int_{v_L}^1vdF-(1-F(v_L))r=\int_{v_L}^1(v-r)dF$. } Intuitively, $G_r(\cdot)$ is constructed from $F(\cdot)$ by shifting all mass in $[v_L,r]$ above $r$ in a mean preserving way. Thus, the expected gain of sampling from $G_r(\cdot)$ and keeping the draw when it exceeds $r$ (the left hand side of \eqref{eq_r_j}) is equal to the expected gain of sampling from $F(\cdot)$ and keeping the draw when it exceeds $v_L$ (the left hand side of \eqref{eq_search2}). Thus, the search equation resembles its counterpart in a benchmark with full information, but with a crucial difference: the lower bound of integration is $v_L$, rather than $r$.

In this form, the search equation relates the equilibrium values of $v_L$ and $r$ to the (exogenous) prior distribution $F(\cdot)$, rather than the equilibrium distribution $G_r(\cdot)$, allowing us to analyze it in a more straightforward way. Viewing the search equation \eqref{eq_search2} as an implicit function $r^{se}(\cdot)$ that relates $v_L\in[0,1)$ to $r$, we find that it has a relatively simple shape: $r^{se}(\cdot)$ starts at positive value $r^{se}(0)=\mu-s$, increases until it hits the 45\textdegree{}-line at the full information reservation value, $r^{fi}$, and it decreases thereafter (see Lemma \ref{lemma_search_eq}). Meanwhile, our characterization of the equilibrium with exogenous reservation value in Proposition \ref{prop_r_exog}, relates $r$ and $v_L$ via function $v_L^{eq}(r)$ or its inverse correspondence $r^{eq}(v_L)$.  

The following proposition links $r^{se}(\cdot)$ and $r^{eq}(\cdot)$, characterizing the equilibrium with endogenous reservation value.

\begin{figure}
\begin{minipage}{0.5\textwidth}
\begin{tikzpicture}[scale=0.65]
x={(0,1)}, y={(1,0)}
\pgfmathsetmacro{\al}{0.95}
\pgfmathsetmacro{\s}{0.1}
 \pgfmathsetmacro{\veq}{10*(1-2*sqrt{\s})}
 \pgfmathsetmacro{\rF}{10*(1-sqrt{2*\s})}
 \pgfmathsetmacro{\root}{10*sqrt{(1-2*\s)}}

\draw[line width=0.5pt] (0,0) -- (0,10);     
\draw[line width=0.5pt] (0,0) -- (10,0);     
\draw[line width=0.5pt, dotted] (0,10) -- (10,10) -- (10,0);  
\draw[line width=0.5pt] (0,0) -- (10,10);     
\draw[line width=0.5pt, dotted] 
(5.52, 0) -- (5.52, 5.52);

\draw[line width=1pt,red] (0,0) -- (0,10*0.5*\al)--(10,10);
\draw[color=blue, domain=0:\root, line width=1pt,samples=200] 
    plot (\x, {10*(1-\s-0.5*(1-(\x/10)^2)-(\x/10)^2)/(1-(\x/10))});

\fill (10,0) node[below] {\footnotesize{$1$}};
\fill (10,0) node[right] {\footnotesize{$v_L$}};
\fill (0,10) node[above] {\footnotesize{$r$}};
\fill (0,10*0.5*\al) node[left] {\footnotesize{$\underline{r}$}};
\fill (5,9) node[below] {\footnotesize{$r^{eq}(\cdot)$}};
\fill (5.52,-0.10) node[below] {\footnotesize{$r^{fi}$}};
\fill (0,0) node[below] {\footnotesize{$v_L^*$}};

\fill (0,5-10*\s) node[left] {\footnotesize{$r^{se}(0)$}};
\fill (10.2,3) node[left] {\footnotesize{$r^{se}(\cdot)$}};

\filldraw[black] (0, 4) circle (3pt);

\end{tikzpicture}
\end{minipage}
\begin{minipage}{0.5\textwidth}
\begin{tikzpicture}[scale=0.65]
\pgfmathsetmacro{\al}{0.65}
\pgfmathsetmacro{\s}{0.1}
 \pgfmathsetmacro{\veq}{10*(1-sqrt{2*\s/(1-\al))})}
 \pgfmathsetmacro{\req}{10*(1-(2-\al)*sqrt{2*\s/(1-\al)/2})}
 \pgfmathsetmacro{\rF}{10*(1-sqrt{2*\s})}
 \pgfmathsetmacro{\root}{10*sqrt{(1-2*\s)}}

\draw[line width=0.5pt] (0,0) -- (0,10);     
\draw[line width=0.5pt] (0,0) -- (10,0);     
\draw[line width=0.5pt, dotted] (0,10) -- (10,10) -- (10,0);  
\draw[line width=0.5pt] (0,0) -- (10,10);     
\draw[line width=0.5pt, dotted] 
(5.52, 0) -- (5.52, 5.52);

 \draw[line width=0.5pt,dotted] (2.4,0) -- (2.4,4.9)--(0,4.9);

 \draw[line width=0.5pt,dotted] (5.52,5.52) -- (3.3,5.52)--(3.3,0);
\draw[line width=0.5pt,dotted] (3.3,5.52)--(0,5.52);

\draw[decorate,decoration={brace,amplitude=10pt}] 
    (3.3,0) -- (5.52,0) ;

\draw[line width=1pt,red] (0,0) -- (0,10*0.5*\al)--(10,10);
\draw[color=blue, domain=0:\root, line width=1pt, samples=200] 
    plot (\x, {10*(1-\s-0.5*(1-(\x/10)^2)-(\x/10)^2)/(1-(\x/10))});

\fill (10,0) node[below] {\footnotesize{$1$}};
\fill (10,0) node[right] {\footnotesize{$v_L$}};
\fill (0,10) node[above] {\footnotesize{$r$}};
\fill (0,10*0.5*\al) node[left] {\footnotesize{$\underline{r}$}};
\fill (5,8) node[below] {\footnotesize{$r^{eq}(\cdot)$}};
\fill (5.52,0) node[below] {\footnotesize{$r^{fi}$}};
\fill (0,5.52) node[left] {\footnotesize{$r^{fi}$}};
\fill (2.4,0) node[below] {\footnotesize{$v_L^*$}};
\fill (0,4.9) node[left] {\footnotesize{$r^*$}};
\fill (0,5-10*\s) node[left] {\footnotesize{$r^{se}(0)$}};
\fill (10.2,3) node[left] {\footnotesize{$r^{se}(\cdot)$}};
\fill (3.8,0) node[below] {\scriptsize{$v_L^{eq}(r^{fi})$}};
\fill (4.4,0.5) node[above] {\scriptsize{gap}};

\filldraw[black] (2.4, 4.9) circle (3pt);

\end{tikzpicture}
\end{minipage}

\caption{\label{fig_endog}\textit{Equilibrium With Endogenous Reservation Value}. Equilibrium $(r^*,v_L^*)$ shown as black dot. Left panel: No disclosure at the bottom. Right panel: Disclosure at the bottom. Both panels drawn for the uniform distribution; $\alpha$ higher in left panel.}
\end{figure}

\begin{proposition}\label{prop_r_endog}(Endogenous Reservation Value). Consider the game with endogenous reservation value. For each $(n,\alpha,s,F(\cdot))$, an equilibrium exists and is unique. In particular, the equilibrium strategy $G^*(\cdot)$ takes the form of \eqref{eq_G}, with unique values $\{v_L^*,v_H^*,\beta^*,r^*\}$. Furthermore, the equilibrium has the following properties, 
\begin{itemize}
    \item if $\underline{r}(n,\alpha)\geq r^{se}(0)=\mu-s$, then there is no disclosure at the bottom ($v_L^*=0$, $r^*=\mu-s$).
    \item if $\underline{r}(n,\alpha)< r^{se}(0)=\mu-s$, then there is disclosure at the bottom ($v_L^*>0$, $r^*>\mu-s$). 
     \item  the reservation value is strictly smaller than under full information  ($r^*<r^{fi}$).
\end{itemize}
\end{proposition}

Proposition \ref{prop_r_endog} is illustrated in Figure \ref{fig_endog}. In each panel, the red curve represents the correspondence $r^{eq}(\cdot)$, while the blue curve represents $r^{se}(\cdot)$. The equilibrium combination $(r^*,v_L^*)$ occurs at the intersection, illustrated by a black dot. Several findings of Proposition \ref{prop_r_endog} can be deduced from this figure. First, notice that whenever $\underline{r}\geq r^{se}(0)$, an equilibrium with no disclosure at the bottom exists ($v_L^*=0,r^*=r^{se}(0)$).  Second, notice that in the right panel, where $\underline{r}<r^{se}(0)$, a crossing occurs at an interior $v_L^*\in(0,v_L^{eq}(r^{fi}))$, implying existence of an equilibrium with disclosure at the bottom. To see the general logic, consider the equilibrium with exogenous reservation value $r=r^{fi}$. The lower disclosure threshold in such an equilibrium must be strictly smaller than the full information reservation value, $v_L^{eq}(r^{fi})<r^{fi}$, as illustrated in the right panel of Figure \ref{fig_endog}. Furthermore, point $(r^{fi},r^{fi})$ is the peak of $r^{se}(\cdot)$. Therefore,  at $v_L^{eq}(r^{fi})<r^{fi}$, the blue curve $r^{se}(\cdot)$ must lie below the red curve $r^{eq}(\cdot)$. Of course with $\underline{r}<r^{se}(0)$, the blue curve lies above the red at $v_L=0$. By implication, a crossing must exist at some $(v^*_L,r^*)$, with $v_L^*<v_L^{eq}(r^{fi})$ and $r^*<r^{fi}$. 

We have illustrated graphically that an equilibrium always exists, and that it has the properties described in Proposition \ref{prop_r_endog}. In the proof, we also show that the equilibrium is unique. To do so, we establish that at any interior crossing, $r^{eq}(\cdot)$ must cross through $r^{se}(\cdot)$ from below. This property rules out an interior crossing when $\underline{r}\geq r^{se}(0)$, and it rules out multiple crossings when $\underline{r}<r^{se}(0)$.

\section{Limits of Disclosure}
In this section, we explore the positive and normative implications of changes in the market structure ($n$) and the search technology ($s$). With respect to market structure, we compare a small oligopoly and a larger, more competitive market. We also examine how changes in search cost affect equilibrium informativeness and welfare in the limiting cases $s\rightarrow \{0,\mu\}$. 

\paragraph{Informativeness and Welfare.} To compare informativeness of equilibrium disclosure strategies, we use the notion of \textit{integral precision} \citep{ganuza2010signal}. This criterion ranks informativeness of signals according to the dispersion of the posterior mean, where less ``dispersion'' is interpreted as a mean-preserving contraction.

\begin{definition}(Informativeness). Posterior mean distribution $G_0\in\text{MPC}(F)$ is less informative than $G_1\in\text{MPC}(F)$, if and only if $G_0\in \text{MPC}(G_1)$ and $G_0\neq G_1$. 
\end{definition}

To compare the normative properties of disclosure strategies, we identify the expected surplus of each consumer type. A \savvy{} consumer searches all firms; therefore, the \savvy{} type's expected payoff is the expected value of the maximum order statistic. In contrast, a \inexperienced{} consumer searches until she finds a value that exceeds the reservation value. Existing results in the search literature show that a \inexperienced{} consumer's surplus can also be represented as an order statistic, but from an adjusted distribution that accounts for the cost and optimal search strategy \citep{KWW2016, armstrong2017ordered, choi2018consumer}. In particular, if the posterior mean $v_i\sim G$ and $r$ is the \inexperienced{} consumer's optimal reservation value (from \eqref{eq_r_j}), then the \savvy{} consumer's expected surplus is $\CSsav=E_G[\max\{v_1,...,v_n\}]$, and the \inexperienced{} consumer's expected surplus is $\CSin=E_{G}[\max\{\tilde{v}_1,...,\tilde{v}_n\}]$, where $\tilde{v}_i=\min\{v_i,r\}$, which is sometimes called the ``effective'' or ``capped'' value. It is straightforward to show that \savvy{} consumers strictly benefit from increases in informativeness,  that \inexperienced{} consumers weakly benefit, and that they strictly benefit unless the two distributions differ only above their identical reservation value.

\subsection{Market Structure}

The following proposition links key properties of the equilibrium to the structure of the market, establishing substantial differences between large and small markets.

\begin{proposition}(Market Structure).
For each $(s,\alpha,F(\cdot))$, a finite $ \underline{n}\geq 2$ exists with the following properties:
\begin{itemize}

        \item (Competition and Disclosure). In a large market, $n\geq \underline{n}$, the unique equilibrium has no disclosure at the bottom ($v_L^*=0$, $r^*=\mu-s$). In a small market, $n<\underline{n}$, the unique equilibrium has disclosure at the bottom ($v^*_L>0$, $r^*>\mu-s$).
    
    \item (Competition and Informativeness). For $n>\underline{n}$, equilibrium informativeness increases with $n$, but the equilibrium does not approach full disclosure as $n\to \infty$.
    
    \item (Paradox of Choice). In equilibrium, the \inexperienced{} consumer visits more than one firm with positive probability if and only if $n<\underline{n}$.

    \item (Rich Get Richer). For $n>\underline{n}$, the \savvy{} type's equilibrium surplus increases in $n$. 
    
    \item (Poor Get Poorer) The \inexperienced{} type's equilibrium surplus is larger with any $n<\underline{n}$ than with any $n'\geq \underline{n}$. For all $n'\geq \underline{n}$, the \inexperienced{} type's equilibrium surplus is $\mu-s$, the expected value of a single visit.

\end{itemize}
\label{prop_large_n}
\end{proposition}

Proposition \ref{prop_large_n} is a significant departure from the literature on frictionless search, where competitive pressure leads to full disclosure in a large economy \citep{ivanov2013information,au2020competitive,hwang2019competitive}.\footnote{As long as search is frictionless, this is also true if consumers have an exogenous reservation value; see \citet{hwang2019competitive} Theorem 3 and the discussion that follows. }
The difference must be rooted in the search friction of the \inexperienced{} consumers, but it may be initially unclear why the \inexperienced{} consumers have any effect at all in a large market. Indeed, with many firms, it is extremely unlikely that any particular firm will be visited by a \inexperienced{} consumer, and it infers that a consumer who visits is almost certainly \savvy{}. It is therefore plausible that the equilibrium might converge to full disclosure, the unique equilibrium when all consumers are known to be \savvy{}. 

This logic is incomplete however, because a firm's disclosure incentives are determined by the probability that its product is chosen, not just the probability of visit or its belief about a consumer's type. As discussed after Proposition \ref{prop_r_exog}, the extent of disclosure at the bottom is determined by a tradeoff across market segments. By shifting all mass above the reservation value, the firm induces immediate stopping by the \inexperienced{}, while reducing the probability of sale to \savvy{} consumers in the affine interval, where the shifted mass is absorbed. In large markets, the benefits of such concealment are substantial, since a consumer who leaves is extremely unlikely to return. Simultaneously, the costs are small, as the (moderate) valuations in the affine interval are unlikely to convert the \savvy{} consumers, even if disclosed truthfully. Consequently, when competition is sufficiently intense, the equilibrium has no disclosure at the bottom. Although a visit by a \inexperienced{} consumer is improbable ex ante, in a large market it is crucial that the firm persuades all \inexperienced{} consumers who visit to select its product, since persuading a \savvy{} consumer is far more difficult. A similar force arises in \citet{stahl1989oligopolistic}, whereby firms focus more on the \inexperienced{} (high-cost) market segment as competition increases. As we discuss in Section \ref{infovsprices}, however, this force manifests differently in our model, leading to a  distinct comparative static with respect to market structure.

Obviously, with $r^*=\mu-s$ and no disclosure at the bottom, the equilibrium could not possibly approach full disclosure as competition grows. Nevertheless, part of the intuition of frictionless models carries over. Indeed, an increase in the number of firms does result in a more informative equilibrium, but the spread in mass is confined to valuations above $r^*=\mu-s$. In particular, the proof of Proposition \ref{prop_large_n} establishes that $v_H^*(n)$ is decreasing in $n$, approaching a limit $v^\infty_H\in(\mu-s,1)$---thus, the limiting distribution has an interval of truthful disclosure at the top, but it does not approach full disclosure. Intuitively, as competition intensifies, firms focus on highlighting the most attractive features of their products while downplaying, concealing, or manipulating less favorable information in hopes of converting (\inexperienced{}) first-time visitors. In other words, a firm expands the interval of high values it discloses, which requires a downward shift in mass (confined to $v\geq r$) to compensate. 

The implications are evident in the left panel of Figure \ref{fig_large_n}, which illustrates two typical posterior mean distributions, with $n'>n>\underline{n}$. Both distributions have no disclosure at the bottom, and the reservation value is $r^*=\mu-s$. In the more competitive market, $v_H^*$ is lower, more mass is shifted left, and the posterior mean is more informative.

\begin{figure}
\begin{minipage}{0.5\textwidth}
\begin{tikzpicture}[scale=0.65]

\definecolor{violet}{RGB}{120, 40, 160}

\pgfmathsetmacro{\vL}{0}
\pgfmathsetmacro{\r}{4}
\pgfmathsetmacro{\vH}{8.5}
\pgfmathsetmacro{\vHH}{7}

\draw[color=red, domain=\r:\vH, line width=1pt] 
    plot (\x, {0.1*\vL^2 + 0.1*(\vH^2 - \vL^2)*((\x - \r)/(\vH - \r))^0.5});

\draw[color=violet, domain=\r:\vHH, line width=1pt,samples=100] 
    plot (\x, {0.1*\vL^2 + 0.1*(\vHH^2 - \vL^2)*((\x - \r)/(\vHH - \r))^0.2});

\draw[color=red, domain=\vH:10, line width=1pt] 
    plot (\x, {0.1*\x^2});

\draw[color=blue, domain=\vL:\vHH, line width=0.5pt] 
    plot (\x, {0.1*\x^2});
\draw[color=violet, domain=\vHH:\vH, line width=1pt] 
    plot (\x, {0.1*\x^2});

\draw[line width=0.5pt] (0,0) -- (0,10.5);     
\draw[line width=0.5pt] (0,0) -- (10.5,0);     
\draw[line width=0.5pt, dotted] (0,10) -- (10,10) -- (10,0);  

\draw[line width=0.5pt, dotted] (\r,0) -- (\r, {0.1*\vL^2});
\draw[line width=0.5pt, dotted] (\vL,0) -- (\vL, {0.1*\vL^2});
\draw[line width=0.5pt, dotted] (\vH,0) -- (\vH, {0.1*\vH^2});
\draw[line width=0.5pt, dotted] (4.3,0) -- (4.3,1.8);  
\draw[line width=0.5pt, dotted] (\vHH,0) -- (\vHH, {0.1*\vHH^2});

\fill (\r-0.2,0) node[below] {\footnotesize{$r^*=\mu-s$}};
\fill (\vH+0.1,0) node[below] {\scriptsize{$v^*_H(n)$}};
\fill (\vHH-0.1,0) node[below] {\scriptsize{$v^*_H(n')$}};
\fill (10,0) node[below] {\footnotesize{$1$}};
\fill (0,0) node[left] {\footnotesize{$0$}};
\fill (0,10) node[left] {\footnotesize{$1$}};
 \fill (4,3.8) node[above] {\footnotesize{$G_{n'}(\cdot)$}};
 \fill (8,8.3) node[above]{\footnotesize{$G_n(\cdot)$}};
\end{tikzpicture}
\end{minipage}
\begin{minipage}{0.5\textwidth}
\begin{tikzpicture}[scale=0.65]

\definecolor{violet}{RGB}{120, 40, 160}

\pgfmathsetmacro{\vL}{2}
\pgfmathsetmacro{\r}{4}
\pgfmathsetmacro{\vH}{7.6}
\pgfmathsetmacro{\rH}{3.2}
\pgfmathsetmacro{\vHH}{6.2}

\draw[color=red, domain=0:\vL, line width=1pt] 
    plot (\x, {0.1*\x^2});

\draw[color=red, line width=1pt] 
    (\vL, {0.1*\vL^2}) -- (\r, {0.1*\vL^2});

\draw[color=red, domain=\r:\vH, line width=1pt] 
    plot (\x, {0.1*\vL^2 + 0.1*(\vH^2 - \vL^2)*((\x - \r)/(\vH - \r))^0.5});

\draw[color=red, domain=\vH:10, line width=1pt] 
    plot (\x, {0.1*\x^2});

\draw[color=violet, domain=\rH:\vHH, line width=1pt,samples=100] 
    plot (\x, {0.1*(\vHH^2 )*((\x - \rH)/(\vHH - \rH))^0.35});
\draw[color=violet, line width=1pt] 
    (0,0)--(\rH,0);
\draw[color=violet, domain=\vHH:\vH, line width=1pt] 
    plot (\x, {0.1*\x^2});

\draw[color=blue, domain=\vL:\vHH, line width=0.5pt] 
    plot (\x, {0.1*\x^2});

\draw[line width=0.5pt] (0,0) -- (0,10.5);     
\draw[line width=0.5pt] (0,0) -- (10.5,0);     
\draw[line width=0.5pt, dotted] (0,10) -- (10,10) -- (10,0);  

\draw[line width=0.5pt, dotted] (\r,0) -- (\r, {0.1*\vL^2});
\draw[line width=0.5pt, dotted] (\vL,0) -- (\vL, {0.1*\vL^2});
\draw[line width=0.5pt, dotted] (\vH,0) -- (\vH, {0.1*\vH^2});
\draw[line width=0.5pt, dotted] (\vHH,0) -- (\vHH, {0.1*\vHH^2});

\fill (\r+0.5,0) node[below] {\scriptsize{$r^*(n)$}};
\fill (\rH,0) node[below] {\scriptsize{$r^*(n')$}};
\fill (\vL-0.2,0) node[below] {\scriptsize{$v^*_L(n)$}};
\fill (\vH+0.3,0) node[below] {\scriptsize{$v^*_H(n)$}};
\fill (\vHH,0) node[below] {\scriptsize{$v^*_H(n')$}};
\fill (10,0) node[below] {\footnotesize{$1$}};
\fill (0,0) node[left] {\footnotesize{$0$}};
\fill (0,10) node[left] {\footnotesize{$1$}};
\fill (2.8,2) node[above] {\footnotesize{$G_{n'}(\cdot)$}};
\fill (5,4.5) node[above] {\footnotesize{$G_{n}(\cdot)$}};

\end{tikzpicture}
\end{minipage}

\caption{\label{fig_large_n} Left panel: \textit{Competition and Informativeness} ($n'>n>\underline{n}$).  Right panel: \textit{Paradox of Choice} ($n'>\underline{n}>n$).}
\end{figure}

This characterization of the equilibrium has intriguing implications for \inexperienced{} consumers' search behavior. Whenever $\underline{n}>2$, firms' equilibrium disclosure strategies assign positive probability to values below $r$ if the market is small ($n<\underline{n}$), but not if the market is large ($n\geq \underline{n}$). Therefore, in a small market a realized value below $r$ leads to additional search, but in a large market, the probability of such a draw is zero, and a \inexperienced{} consumer always stops at the first firm she visits. In other words, a \inexperienced{} consumer actively searches in a small market, but not in a large market. Given its connection to the epigraph, we refer to this finding as the \textit{paradox of choice}. 

The right panel of Figure \ref{fig_large_n} illustrates the paradox of choice. It presents two typical equilibrium distributions, with $n'>\underline{n}>n$. The distribution in the smaller market has disclosure at the bottom, which leads to an active search with positive probability, while the larger market does not.  Notably, these distributions cross twice and are not ranked by informativeness: in the larger market, active search and \inexperienced{} consumers' welfare are undermined by the concealment of low values, not by an overall reduction in equilibrium informativeness.\footnote{This observation complicates the analysis of \savvy{} consumers' welfare. When $n\geq \underline{n}$, greater competition increases equilibrium informativeness, which increases \savvy{} consumers' welfare. However, for $n'>\underline{n}>n$, the equilibrium distribution for $n'$ and $n$ are not necessarily ranked by informativeness.}

The paradox of choice is a novel finding, arising from the opposing disclosure incentives across market segments. It is distinct from the informational Diamond paradox that arises when search costs are homogeneous, whereby consumers always stop at the first firm regardless of how many firms are in the market \citep{au2023attraction}. Indeed, the paradox of choice connects the \inexperienced{} consumers' search behavior to market size, with active search in small markets but not large ones. The paradox of choice also does not arise in a standard sequential search model \citep{weitzman1979optimal}. If the distribution of prizes is fixed, then the optimal sequential search strategy depends only on the reservation value derived from the distribution of prizes, not on the number of boxes available. In other words, it is the strategic interaction between firms (boxes) that links the distribution of posterior means (prizes) to the size of the market, thereby generating the paradox of choice. Despite a mix of \savvy{} and \inexperienced{} consumers and competition between firms, the paradox of choice also does not emerge in \citet{stahl1989oligopolistic}'s study of pricing in search markets. A major finding of Stahl's paper is that the support of the equilibrium price distribution is confined below the reservation price, and \inexperienced{} (high cost) consumers stop at the first firm they visit. We elaborate on the contrast with \citet{stahl1989oligopolistic} in Section \ref{infovsprices}.

The paradox of choice has significant implications for welfare. In particular, \inexperienced{} consumers might be strictly harmed when the market becomes more competitive. When the paradox of choice is at play ($\underline{n}>2$) \inexperienced{} consumers engage in active search in a small market. Such consumers always have the option not to search past the first firm, but because they strictly prefer to do so (when $v<r$), active search must make them better off. In other words, in a small market, \inexperienced{} consumers' expected surplus exceeds the payoff of taking a single draw ($\mu-s$). In a large market, however, firms' disclosure strategies preempt \inexperienced{} consumers' search, reducing their expected surplus to the payoff of a single draw. By implication, \inexperienced{} consumers are strictly better off in a small market than a large one.

\subsection{Search Technology}

 The link between the search technology,   informativeness, and the welfare of each market segment is also worthy of analysis. Recall from Proposition \ref{prop_r_exog} that when the exogenous reservation value approaches the extremes ($r\rightarrow \{0,1\}$) the equilibrium approaches full disclosure. Furthermore, from the search equation \eqref{eq_r_j} it is immediate that $r\rightarrow 0$ when $s\rightarrow \mu$, and $r\rightarrow v_T$ as $s\rightarrow 0$. Furthermore, for $r\rightarrow v_T$, the \inexperienced{} consumer always searches, just like the \savvy{} consumer. We therefore might expect that as the search cost approaches 0 or $\mu$, the equilibrium approaches full disclosure. The following proposition confirms this intuition and shows that for sufficiently high search costs, changes in the search cost have a differential impact on consumers.

 \begin{figure}
\begin{minipage}{0.5\textwidth}
     \begin{tikzpicture}[scale=0.65]

\definecolor{violet}{RGB}{120, 40, 160}

\pgfmathsetmacro{\vL}{0}
\pgfmathsetmacro{\r}{4}
\pgfmathsetmacro{\vH}{8.5}
\pgfmathsetmacro{\rH}{3.2}
\pgfmathsetmacro{\vHH}{6.2}

\draw[color=red, domain=0:\vL, line width=1pt, samples=100] 
    plot (\x, {0.1*\x^2});

\draw[color=red, line width=1pt] 
    (\vL, {0.1*\vL^2}) -- (\r, {0.1*\vL^2});

\draw[color=red, domain=\r:\vH, line width=1pt,samples=100] 
    plot (\x, {0.1*\vL^2 + 0.1*(\vH^2 - \vL^2)*((\x - \r)/(\vH - \r))^0.5});

\draw[color=red, domain=\vH:10, line width=1pt] 
    plot (\x, {0.1*\x^2});

\draw[color=violet, domain=\rH:\vHH, line width=1pt,samples=100] 
    plot (\x, {0.1*(\vHH^2 )*((\x - \rH)/(\vHH - \rH))^0.35});
\draw[color=violet, line width=1pt] 
    (0,0)--(\rH,0);
\draw[color=violet, domain=\vHH:\vH, line width=1pt] 
    plot (\x, {0.1*\x^2});

\draw[color=blue, domain=\vL:\vHH, line width=0.5pt] 
    plot (\x, {0.1*\x^2});

\draw[line width=0.5pt] (0,0) -- (0,10.5);     
\draw[line width=0.5pt] (0,0) -- (10.5,0);     
\draw[line width=0.5pt, dotted] (0,10) -- (10,10) -- (10,0);  

\draw[line width=0.5pt, dotted] (\r,0) -- (\r, {0.1*\vL^2});
\draw[line width=0.5pt, dotted] (\vL,0) -- (\vL, {0.1*\vL^2});
\draw[line width=0.5pt, dotted] (\vH,0) -- (\vH, {0.1*\vH^2});
\draw[line width=0.5pt, dotted] (\vHH,0) -- (\vHH, {0.1*\vHH^2});

\fill (\r+0.5,0) node[below] {\scriptsize{$\mu-s_0$}};
\fill (\rH-0.5,0) node[below] {\scriptsize{$\mu-s_1$}};
\fill (\vH+0.3,0) node[below] {\scriptsize{$v^*_H(s_0)$}};
\fill (\vHH,0) node[below] {\scriptsize{$v^*_H(s_1)$}};
\fill (10,0) node[below] {\footnotesize{$1$}};
\fill (0,0) node[left] {\footnotesize{$0$}};
\fill (0,10) node[left] {\footnotesize{$1$}};
\fill (2.8,2) node[above] {\footnotesize{$G_{1}(\cdot)$}};
\fill (5,4.5) node[above] {\footnotesize{$G_{0}(\cdot)$}};

\end{tikzpicture}
\end{minipage}
\begin{minipage}{0.5\textwidth}
     \begin{tikzpicture}[scale=0.65]

\definecolor{violet}{RGB}{120, 40, 160}

\pgfmathsetmacro{\vL}{0}
\pgfmathsetmacro{\r}{4}
\pgfmathsetmacro{\vH}{8.5}
\pgfmathsetmacro{\rH}{1.5}
\pgfmathsetmacro{\vHH}{2.8}

\draw[color=red, domain=0:\vL, line width=1pt, samples=100] 
    plot (\x, {0.1*\x^2});

\draw[color=red, line width=1pt] 
    (\vL, {0.1*\vL^2}) -- (\r, {0.1*\vL^2});

\draw[color=red, domain=\r:\vH, line width=1pt,samples=100] 
    plot (\x, {0.1*\vL^2 + 0.1*(\vH^2 - \vL^2)*((\x - \r)/(\vH - \r))^0.5});

\draw[color=red, domain=\vH:10, line width=1pt] 
    plot (\x, {0.1*\x^2});

\draw[color=violet, domain=\rH:\vHH, line width=1pt,samples=100] 
    plot (\x, {0.1*(\vHH^2 )*((\x - \rH)/(\vHH - \rH))^0.35});
\draw[color=violet, line width=1pt] 
    (0,0)--(\rH,0);
\draw[color=violet, domain=\vHH:\vH, line width=1pt] 
    plot (\x, {0.1*\x^2});

\draw[color=blue, domain=\vL:\vHH, line width=0.5pt] 
    plot (\x, {0.1*\x^2});

\draw[line width=0.5pt] (0,0) -- (0,10.5);     
\draw[line width=0.5pt] (0,0) -- (10.5,0);     
\draw[line width=0.5pt, dotted] (0,10) -- (10,10) -- (10,0);  

\draw[line width=0.5pt, dotted] (\r,0) -- (\r, {0.1*\vL^2});
\draw[line width=0.5pt, dotted] (\vL,0) -- (\vL, {0.1*\vL^2});
\draw[line width=0.5pt, dotted] (\vH,0) -- (\vH, {0.1*\vH^2});
\draw[line width=0.5pt, dotted] (\vHH,0) -- (\vHH, {0.1*\vHH^2});

\fill (\r+0.5,0) node[below] {\scriptsize{$\mu-s_0$}};
\fill (\rH-0.5,0) node[below] {\scriptsize{$\mu-s_2$}};
\fill (\vH+0.3,0) node[below] {\scriptsize{$v^*_H(s_0)$}};
\fill (\vHH,0) node[below] {\scriptsize{$v^*_H(s_2)$}};
\fill (10,0) node[below] {\footnotesize{$1$}};
\fill (0,0) node[left] {\footnotesize{$0$}};
\fill (0,10) node[left] {\footnotesize{$1$}};
\fill (2.8,1) node[above] {\footnotesize{$G_{2}(\cdot)$}};
\fill (5,4.5) node[above] {\footnotesize{$G_{0}(\cdot)$}};

\end{tikzpicture}
\end{minipage}

\caption{\label{fig_change_s}\textit{Increases in Search Cost.} Both panels have $s_0>\bar{s}$. Left panel: a small increase from $s_0$ to $s_1$. Right panel: a large increase to $s_2$.}
\end{figure}

\begin{proposition} (Search Cost).
For each $(n,\alpha,F(\cdot))$, as $s\rightarrow 0$ or $s\rightarrow \mu$, the equilibrium converges to full disclosure in distribution ($G(\cdot)\rightarrow F(\cdot)$ pointwise). Furthermore, there exists a threshold $ \overline{s}\in (0,\mu)$ such that if $s>\overline{s}$, then equilibrium informativeness and the surplus of the \savvy{} consumer ($\text{CS}_{sel}$) are increasing in $s$, but the surplus of the \inexperienced{} consumer ($\text{CS}_{cas}$) is decreasing in $s$.

\label{com_stat_s}
\end{proposition}

Broadly speaking, this proposition confirms that the consequences of changes in search cost depend crucially on the search cost's magnitude. In the Online Appendix, we show that when the cost is small, reducing it sufficiently increases equilibrium informativeness and the surplus of both consumer segments. In contrast, Proposition \ref{com_stat_s} shows that when the search cost is substantial, a reduction decreases equilibrium informativeness, thereby reducing the surplus of the \savvy{} segment. Conversely, an increase in search cost increases informativeness and the \savvy{} type's surplus, while reducing the \inexperienced{} type's. To see this effect, note that for any $s>\bar{s}\equiv \mu-\underline{r}(n,\alpha)$, the equilibrium has no disclosure at the bottom (Proposition \ref{prop_r_endog}). As the search cost increases, the equilibrium reservation value $r^*=\mu-s$ decreases. Consequently, less mass must be shifted into the affine interval, resulting in a smaller distortion. The increase in equilibrium informativeness benefits \savvy{} consumers, but not the \inexperienced{}, who stop at the first firm they visit.  Meanwhile, the increase in search cost is borne solely by the \inexperienced{} consumers. Thus, when the search cost exceeds $\bar{s}$, further increases benefit \savvy{} consumers, at \inexperienced{} consumers' expense.

The effect of an increase in search cost for $s>\bar{s}$ is depicted in Figure \ref{fig_change_s}. The left panel illustrates a small increase in search cost from $s_0$ to $s_1$. The resulting distributions cross once in the region of partial disclosure and are ordered by informativeness, with $G_1(\cdot)$ more informative. For a large increase in search cost, depicted in the right panel, the effect is even more profound. The interval of partial disclosure collapses and shifts toward lower values. Thus, $G_2(\cdot)$ coincides with $F(\cdot)$ over most of its support.

\section{Extensions}\label{extensions}
\subsection{Partial Search}
\label{Search_capacity}

In our base model, \savvy{} consumers visit all firms before making a purchase. In this section, we extend the model to allow for the possibility that even \savvy{} consumers engage in a partial search. In particular, we model search as a two-stage process in which consumers first narrow the set of available alternatives to a smaller consideration set, and then explore and choose within it. Such a process might arise from time or resource constraints, limited awareness, or filtering based on easily observed characteristics. A number of empirical studies document similar behavior across a range of markets, including the PC market (\citealp{goeree2008limited}), the banking industry (\citealp{honka2017advertising}), and online retailing (\citealp{kim2010online}).

To understand how such limitations affect the equilibrium, consider an extension of our base model in which consumers first choose a consideration set, consisting of $k\in [2,n)$ firms, where  $F(\cdot)^{k-1}$ is strictly convex. In a symmetric equilibrium with identical firms, each firm is equally likely to be part of a consumer’s consideration set. \Savvy{} consumers inspect all products within their consideration set. Thus, the probability of a visit by a \savvy{} consumer is $\frac{k}{n}$ and the conditional probability of being chosen is $G(v)^{k-1}$. In contrast, \inexperienced{} consumers conduct a sequential search within their consideration sets, stopping when they encounter a draw above their reservation value. Thus, a \inexperienced{} consumer visits firm $i$ with probability

\[
\overline{\eta}_k=\frac{1}{n}(1+G(r)+G(r)^2+...+G(r)^{k-1})=(\frac{k}{n})\frac{1-G(r)^{k}}{k(1-G(r))}=\frac{k}{n}\eta(k),
\]
where $\eta(k)$ is defined by eq. (\ref{eta}) for $k$ firms. In other words, since each firm is equally likely to be in the consumer's consideration set, the likelihood of a visit is simply the probability of being in the set ($\frac{k}{n}$) times the likelihood of being one of the visited firms within the set ($\eta(k)$). Consequently, the firm's belief that a consumer who visits is \inexperienced{} is given by

\[
\overline{\alpha}_k=\frac{\alpha \overline{\eta}_k}{\alpha \overline{\eta}_k+\frac{k}{n}(1-\alpha)}=\frac{\alpha\eta(k)}{\alpha\eta(k)+(1-\alpha)}=\widetilde{\alpha}(k),
\]
where $\widetilde{\alpha}(k)$ is defined by eq. (\ref{alphatilde}) for $k$ firms. As in our main model, the conditional probability of being chosen by a \inexperienced{} consumer is $1$ for $v\geq r$ and $\frac{\frac{k}{n}G^{k-1}(v)}{\overline{\eta}_k}=\frac{G(v)^{k-1}}{\eta(k)}$ for $v<r$. 

Combining the above observations, we find that the firm's payoff function is

\begin{equation}
\overline{u}_k(v)=
\begin{cases}
(\frac{\widetilde{\alpha}(k)}{\eta(k)}+(1-\widetilde{
\alpha}(k)))G(v)^{k-1} & \text{if } v< r\\
\widetilde{\alpha}(k)+(1-\widetilde{\alpha}(k))G(v)^{k-1} & \text{if } v\geq r.
\end{cases}
\label{eq_u_cap}
\end{equation}
Comparing eqs. (\ref{eq_u}) and (\ref{eq_u_cap}), the firm's payoff in this extension coincides with its payoff in the baseline model, with the number of firms in the market ($n$) replaced by the size of the consideration set ($k$). In other words, it is the size of the consideration set that determines competitive intensity, not the size of the market. With undirected search, firms cannot affect the probability that a consumer visits. Thus, they choose their disclosure strategy to maximize the probability that their product is chosen conditional on visit, which depends on the number of products under consideration, not the total number available. In essence, partial search restricts competition to a subset of $k-1$ rivals, without otherwise changing the firms' incentives. Thus, our insights and results are valid, even when search is less extensive than in the main model.

\subsection{Additional Cost Heterogeneity}\label{sec_hetero}
In this section, we consider more general forms of cost heterogeneity among \inexperienced{} consumers, showing that our results for large markets extend in a natural way. Suppose that as in the main model, the consumer is \savvy{} with probability $1-\alpha$ and visits all firms. With probability $\alpha$, the consumer is \inexperienced{}. Conditional on being \inexperienced{}, the search cost is distributed according to distribution function $K(\cdot)$ whose support lies below $\mu$. To ease the exposition, we focus on distributions satisfying the following assumption.

\begin{assumption} Distribution function $K(\cdot)$ either (a) has finite support $0<s_1<...<s_k$, or (b) admits a continuous density $k(\cdot)$ and has support on interval $[s_1,s_k]$. If $K(\cdot)$ is continuous, then (i) the lower bound of the support is strictly positive, i.e., $s_1>0$, and (ii) the lowest cost in the support has non-zero density, $k(s_1)>0$.
\end{assumption}

For a continuous distribution, point (i) requires that search costs are bounded away from zero, so that all \inexperienced{} consumers face a non-negligible cost of search. Furthermore, point (ii) restricts the density $k(\cdot)$ to be non-zero at a single point, which is a generic feature. Both of these properties are automatically satisfied when the distribution has finite support, and our characterization applies to any finite distribution. In the Online Appendix (pg. 10), we prove the following result.
\begin{proposition}\label{prop_hetero}(Cost Heterogeneity). Suppose that conditional on being \inexperienced{}, the consumer's search cost is distributed according to $K(\cdot)$, consistent with Assumption 1. A finite $\underline{n}_K$ exists such that for $n>\underline{n}_K$, an equilibrium with no disclosure at the bottom exists.
\end{proposition}

To see the idea, suppose for a moment that there is no heterogeneity among the \inexperienced{} consumers: a \inexperienced{} consumer's cost is $s_1$ with probability 1. For all sufficiently large $n$ ($n\geq \underline{n}$ computed from $s_1$), the unique equilibrium $G(\cdot|s_1)$ has no disclosure below the reservation value $r_1=\mu-s_1$, and the \inexperienced{} consumer (cost $s_1$) stops at the first visited firm. Now suppose that firms continue to provide information according to $G(\cdot|s_1)$, but that the \inexperienced{} consumer's search cost is distributed according to $K(\cdot)$. With disclosure strategy $G(\cdot|s_1)$ a \inexperienced{} consumer stops at the first firm regardless of her realized search cost (even with the lowest cost $s_1$ the consumer stops). By implication, cost heterogeneity has no effect on the payoff function above $r_1=\mu-s_1$. Below $r_1$, however, the payoff function is no longer zero; it is essentially a reflection of the distribution of costs, scaled by $\widetilde{\alpha}$. In the proof, we show that for sufficiently large $n$, providing information according to $G(\cdot|s_1)$ is also optimal for the adjusted payoff, and $G(\cdot|s_1)$ remains an equilibrium.

Proposition \ref{prop_hetero} suggests that when the market contains both \savvy{} consumers (who always search) and \inexperienced{} consumers (with positive search cost) concealment of low values in large markets is a robust equilibrium phenomenon, arising generically when the smallest cost among \inexperienced{} consumers is strictly positive. In this case, it is possible for each firm to simultaneously persuade all \inexperienced{} consumers to stop, while revealing some information to compete for the \savvy{} consumers, as in the construction outlined above. In contrast, when the support of the search cost distribution extends to zero, \inexperienced{} consumers with sufficiently small costs search actively whenever the signal is informative. In other words, it is not possible to simultaneously foreclose search by all \inexperienced{} consumers and compete for \savvy{}s by providing information. While no general characterization of the equilibrium with $s_1=0$ exists, one might naturally expect that if the search cost distribution is sufficiently concentrated near zero, then the equilibrium would resemble full disclosure, as in frictionless search. Indeed, in an extension that includes a mass of \savvy{} consumers, \citet{HH2025} show that when search costs are most-likely negligible ($s_1=0$ and $k(\cdot)$ decreasing) full disclosure is an equilibrium when competition is intense. In contrast, when search costs are always non-negligible ($s_1>0$), intense competition  incentivizes concealment of low values, sustaining an equilibrium that does not approach full disclosure (Proposition \ref{prop_hetero}). Together, these findings provide an intriguing taxonomy of equilibrium disclosure in highly competitive search markets, broadly reminiscent of the ``gap/no-gap'' distinction in the literature on Coasian dynamics \citep{FLT1985,GSW1986}.\footnote{Analysis of the intermediate case, where $s_1=0$ and $k(\cdot)$ is not monotone decreasing, would further advance our understanding of such markets and could be interesting for future research. One advantage of $s_1>0$ is that the qualitative properties of our equilibrium do not depend on the global curvature of the search cost distribution, which seems to matter when $s_1=0$.}

\section{Information vs. Prices}\label{infovsprices}

Our analysis highlights key conceptual differences between information and prices as competitive instruments. When setting prices, a firm has complete control: it can choose any price, deviating from any conjectures if it is  beneficial to do so. With mixed strategies, in particular, if a firm is unhappy with a price realization, it can simply replace it with another price that it prefers. Consequently, mixed strategy equilibria are characterized by an indifference condition which requires that all prices in the support of the equilibrium mixed strategy must be optimal against rivals' randomization. In contrast, when designing information structures, the firm has no ability to change the consumer's draw, which is a (possibly private) forecast of her valuation based on the information disclosed by the firm. In other words, the firm can manipulate the process of learning or disclosure to maximize the probability of being chosen ex ante, recognizing that it is unable to change the consumer's forecast ex post. Furthermore, because of the fundamental nature of information, the consumer's distribution of forecasted valuations must be Bayes plausible: broadly speaking, mass on good forecasts must be offset by mass on bad forecasts to maintain the prior mean, and the overall distribution cannot be more spread than the prior.

The Bayes plausibility constraint on the distribution of consumer forecasts, coupled with a firm's inability to alter the realization, creates a different set of economic forces. To highlight these differences, we contrast our findings with \citet{stahl1989oligopolistic}, which studies price competition under equivalent assumptions about consumers, the search process, and timing. \citet{stahl1989oligopolistic} delivers a sharp finding: in equilibrium, the \inexperienced{} (high search cost) consumers always purchase after a single visit \citep[see also][]{EW2012,DDZ2017}; consequently, the number of firms has no bearing on their search behavior. In contrast, our model exhibits the paradox of choice (Proposition \ref{prop_large_n}), whereby the \inexperienced{} consumers search actively in a small market but not in a large one. This possibility of active search by the \inexperienced{} is broadly consistent with casual empiricism and allows for a more nuanced description of search behavior.

 Active search, in turn, arises from a key feature of our equilibrium characterization: disclosure below the reservation value in small markets. In particular, even if rivals use disclosure strategies supported above the reservation value, thereby always inducing a \inexperienced{} consumer to stop, a firm may still prefer to deviate by disclosing values below it. Although such disclosures do not lead to sales directly, they ensure that other signal realizations are interpreted more favorably, relaxing Bayes plausibility and creating an advantage in the \savvy{} market. In other words, when firms compete via information provision, a firm's inability to tamper with the consumer's valuation forecast ex post implies that even realizations with low probability of sale can arise in equilibrium, and the global Bayes plausibility constraint implies that assigning mass to such realizations can be beneficial indirectly. With price competition, in contrast, these forces are absent: the firm is free to set or adjust its price as it wishes with no distributional constraints. Consequently, there is no analogous incentive to deviate: if rivals' strategies induce the \inexperienced{} consumers to stop, then prices above the reservation level never lead to sales, and deviating to such prices is never beneficial.

The same economic forces also shape the comparative statics, generating intriguing contrasts. In both models, a higher search cost reduces \inexperienced{} consumers’ willingness to search. With price competition, firms set higher prices while still inducing \inexperienced{} consumers to stop, resulting in a first-order dominance shift up of the equilibrium strategy \citep{stahl1989oligopolistic}.  In our model, the \inexperienced{} consumers' diminished willingness to search allows them to be captured with a smaller distortion. In particular, increases in search cost reduce the reservation value, implying that less probability mass must be shifted above it to preempt their search. Once search costs exceed the threshold ($\overline{s}$) at which \inexperienced{} consumers stop searching actively, further increases lead to a less distorted and more informative equilibrium. Unlike price competition, where the equilibrium becomes more monopolistic, with information competition the equilibrium converges pointwise to full disclosure, the frictionless outcome akin to Bertrand competition.

The two modes of competition also respond differently to increases in competition. With fierce price competition, converting \savvy{} consumers is difficult, requiring the firm to set a price near the bottom of the support. Meanwhile, in equilibrium \inexperienced{} consumers always buy from the first firm they visit. As competition grows, selling to \savvy{} consumers becomes even more difficult. Firms concentrate more on extracting surplus from the \inexperienced{} consumers, resulting in a distribution that becomes more monopolistic, converging pointwise to the monopoly price as the market becomes large. In our model, in contrast, once the market is large enough that \inexperienced{} consumers stop at the first firm ($n>\underline{n}$), firms compete more aggressively for the ``residual'' \savvy{} market. Thus, further increases in competition result in a more informative, Bertrand-like equilibrium. At the same time, the incentive to foreclose search by the \inexperienced{} consumers prevents disclosure at the bottom, bounding the equilibrium away from full disclosure. Thus, increases in competition have a differential effect: firms disclose more information at the top to compete for \savvy{} consumers, while censoring at the bottom to capture \inexperienced{} consumers. Unlike price competition, where the incentive to capture \inexperienced{} consumers dominates in large markets, with information competition both consumer types shape the equilibrium disclosure.

 \section{Concluding Remarks}\label{discussion}

This paper analyzes a model of competitive disclosure in a search market with cost heterogeneity, based on \citet{stahl1989oligopolistic}, highlighting the critical role of consumer heterogeneity in shaping firms’
incentives to provide information. The tension between \savvy{} consumers and \inexperienced{} consumers generates a novel disclosure pattern: firms distort information around the \inexperienced{} consumer's reservation value, while possibly revealing extreme values truthfully. The extent of this distortion depends crucially on the market size. In small markets, distortion is confined to a narrow region around the reservation value, with full disclosure of low values inducing search by \inexperienced{} consumers. In contrast, in large markets, all valuations below the reservation value are concealed, preventing the market from approaching full disclosure as competition intensifies. By implication, \inexperienced{} consumers search actively in a small market, but not a large one (the paradox of choice), and they may therefore be better off in less competitive markets that encourage some degree of exploration. Moreover, changes in search costs have a non-monotonic effect on disclosure, as firms reallocate their competitive focus across consumer types in response to varying search frictions.    While conventional wisdom suggests that increasing competition or reducing search frictions should improve transparency and market efficiency, our
analysis reveals a more nuanced relationship between search costs, market competition, search, and
disclosure.

These findings contrast sharply with the \citet{stahl1989oligopolistic} model of price competition, in which \inexperienced{} consumers never search beyond the first firm regardless of market size, and higher search costs always lead to stochastically higher prices. One might wonder how the two modes of competition interact in a unified framework where firms choose both prices and information.\footnote{Prices could also be incorporated in other ways. For example, firms or intermediaries might charge consumers for access to product information, as in \citet{mekonnen2025persuaded}. Alternatively, an intermediary might charge \inexperienced{} consumers a fee to search on their behalf, similar to  \citet{BM2001}.} Unlike Stahl's model, in which the equilibrium features mixing over price, the product differentiation induced by endogenous disclosure may support a symmetric pure-strategy equilibrium in both information and prices. In such an equilibrium with full market coverage, the equilibrium disclosure strategy coincides with the one characterized in this paper. In contrast, endogenous disclosure introduces two key challenges when characterizing the equilibrium price, similar to the frictionless benchmark \citet{hwang2019competitive}. First, the characterization needs to rule out simultaneous deviations in price and information provision, which are absent from standard price competition models. Second, since the demand function is determined endogenously by the firm's disclosure, providing sufficient conditions for the existence of a pure-strategy equilibrium is more demanding than in standard models. Developing methods to address these challenges is a promising direction for future research.

A number of other directions are also possible. Given the prominence of \citet{weitzman1979optimal} and \citet{stahl1989oligopolistic}, our analysis focuses on sequential search. A natural direction for future work is to study competitive information provision under alternative search protocols, including clearinghouse models \citep{BM2001}, directed search, prominence, and ordered search  \citep{armstrong2017ordered,choi2018consumer}. In such settings, firms' disclosures may shape not only consumers' stopping decisions, but also which alternatives they inspect and in what order.  Ordered search is particularly relevant for online platforms, where algorithms determine how products are displayed and may allow the platform to influence consumers' search sequence. It could also be interesting to consider firms that are differentiated both horizontally and vertically, with consumer valuations consisting of commonly known vertical quality and an idiosyncratic match value that is disclosed strategically. Another intriguing possibility is examining ``information discrimination,'' whereby firms offer a menu of information structures intended to screen consumer types. In principle, such screening is possible because \savvy{} and \inexperienced{} consumers follow different search strategies, which creates a wedge in their value of information.   Such extensions introduce additional strategic considerations, with the potential for novel insights.

{\centering
\bibliographystyle{chicago}
\bibliography{references}
}

\section{Appendix}

\renewcommand{\thetheorem}{A\arabic{theorem}}
\renewcommand{\thelemma}{A\arabic{lemma}}
\renewcommand{\theproposition}{A\arabic{proposition}}
\renewcommand{\theequation}{A\arabic{equation}}
\renewcommand{\thefigure}{A\arabic{figure}}

\setcounter{theorem}{0}
\setcounter{lemma}{0}
\setcounter{proposition}{0}
\setcounter{equation}{0}
\setcounter{figure}{0}

We use \textit{MPS} and \textit{MPC} to denote mean-preserving spread and contraction, respectively. Consider atomless distributions $F(\cdot)$ and $G(\cdot)$ on $[0,1]$, where $F(\cdot)$ has full support. Let $
H(v)\equiv\int_0^v[F(x)-G(x)]\,dx$. Then $G\in MPC(F)$ (equivalently, $F\in MPS(G)$) if and only if $H(1)=0$ and $H(v)\geq0$ for all $v\in[0,1]$. Define $\mathcal D\equiv\{v:H(v)=0\}$ and $\mathcal P\equiv\{v:H(v)>0\}$. The following lemma is  proved in the Online Appendix (pg. 6).

\begin{lemma}\label{lem:spread} Suppose $F(\cdot)$ and $G(\cdot)$ are as above, and $G\in MPC(F)$. (i) $G(v)=F(v)$ for every $v\in\mathcal D$, (ii) $\mathcal P$ is a union of disjoint open intervals $I^k=(v^k_l,v^k_h)$, each of which satisfies $H(v^k_l)=H(v^k_h)=0$. 
\end{lemma}
\noindent Given (ii), we refer to each connected component of $\mathcal{P}$ (i.e., open interval $I^k$) as a \textbf{partial disclosure interval}. Thus, $(v_l,v_h)$ is a partial disclosure interval if and only if $H(v_l)=H(v_h)=0$ and $H(v)>0$ for all $v\in(v_l,v_h)$.

\begin{proposition}
Suppose that a symmetric equilibrium $G\in MPC(F)$ exists with reservation value $r\in (0,1)$. Then, there exist $0\leq v_{\dagger}\leq v_L<r<v_H\leq 1$, $\gamma$, and $\beta>0$ such that $G(\cdot)$ takes the following form:
\begin{equation}
G(v)=
\begin{cases}
F(v) & \text{for } v\leq v_{\dagger}\\
\left(F(v_{\dagger})^{n-1}+\gamma(v-v_{\dagger})\right)^{\frac{1}{n-1}} & \text{for } v\in(v_{\dagger},v_L)\\
G(v_L) & \text{for } v\in [v_L,r)\\
\min\left\{\left(G(v_L)^{n-1}+\beta(v-r)\right),1\right\}^{\frac{1}{n-1}} & \text{for } v\in[r,v_H]\\
F(v) & \text{for } v\in (v_H,1].
\end{cases}
\label{eq_G_0}
\end{equation}
Furthermore, if $v_{\dagger}<v_L$, then $0<\gamma\leq [F(v_\dag)^{n-1}]'$. Moreover, exactly one partial disclosure interval exists, i.e. $\mathcal{P}=(v_\dag,v_H)$, and $H(v)>0$ iff $v\in(v_\dag,v_H)$. Finally, $G(v_H)=F(v_H)$.

\label{proposition_A1}
\end{proposition}

\noindent \begin{proof}[\textbf{Proof of Proposition \ref{proposition_A1}}]
We prove the propositional statement by establishing a series of equilibrium properties of $G(\cdot)$ and $u(\cdot)$, which are related by equation \eqref{eq_u}, with exogenous reservation value $r\in(0,1)$ and $\alpha\in(0,1)$. With a positive mass of searchers, a standard equilibrium argument establishes that $G(\cdot)$ is continuous on its support.\footnote{An atom at either extreme $v\in\{0,1\}$ is ruled out by the atomless prior distribution. An atom at an interior point $v\in (0,1)$ produces ties with positive probability. This allows for a profitable deviation. By reallocating an arbitrarily small mass from the mass point to values below it, the firm can shift the mass point up by a small amount in a mean-preserving way. In so doing, the firm ensures that all ties break in its favor. A formal proof is available upon request.} We denote the bottom and top of the support of $G(\cdot)$ by $v_B$ and $v_T$, i.e., $v_B\equiv \max\{v:G(v)=0\}$ and $v_T\equiv \min\{v:G(v)=1\}$. Because $G\in MPC(F)$ its support cannot extend beyond that of $F(\cdot)$ in either direction, i.e $v_B\geq 0$ and $v_T\leq 1$.
\smallskip

\noindent \textbf{\underline{Claim 1}}: Suppose an interval $I=[\underline{v},\overline{v}]$ exists such that $G(\underline{v})<G(v)<G(\overline{v})$ for any $v\in (\underline{v},\overline{v})$. Let $h(\cdot)$ be the concave envelope of $u(\cdot)$ on $I$. Then $h(\cdot)$ must be affine on $I$.
\begin{proof}[Proof Sketch] The complete proof is in the Online Appendix (pg. 1).  Suppose, to the contrary, that $h(\cdot)$ is non-affine. Let $l(\cdot)$ be the chord joining $(\underline v,u(\underline v))$ and $(\overline v,u(\overline v))$,
\begin{equation}
l(v)=\frac{\overline{v}-v}{\overline{v}-\underline{v}}u(\underline{v})+\frac{v-\underline{v}}{\overline{v}-\underline{v}}u(\overline{v}).
 \label{eq_chord}
\end{equation}  
Let $M\equiv\max_{v\in I} \{u(v)-l(v)\}$ be the maximum difference between $u(\cdot)$ and its chord $l(\cdot)$. If $h(\cdot)$ is not affine, then $M>0$. Furthermore, any $v^*$ at which $M$ is attained must satisfy $h(v^*)=u(v^*)$, and it must lie in the interval's interior. With $G(\underline{v})<G(v^*)<G(\overline{v})$, it is possible to construct a set $S=(\underline{v},\underline{v}+\epsilon_1)\cup(\overline{v}-\epsilon_2,\overline{v})$ such that $E_G[v|v\in S]=v^*$. Consider a mean preserving contraction $\widetilde{G}\in MPC(G)$ that departs from $G(\cdot)$ only by reallocating all mass in $S$ to the conditional mean $v^*$.  By construction, the convex hull of the support of $G(\cdot|v\in S)$ is $[\underline{v},\overline{v}]$. Because $h(\cdot)$ is concave and non-affine on $[\underline{v},\overline{v}]$, Jensen's inequality is strict. Combined with $u(\cdot)\leq h(\cdot)$, we have $
E_G[u(v)|v\in S]\leq E_G[h(v)|v\in S]<h(v^*)=u(v^*)$. Because collapsing the mass in $S$ to $v^*$ is strictly beneficial, it follows that $E_G[u(v)]<E_{\widetilde{G}}[u(v)]$. Thus, a deviation to $\widetilde{G}(\cdot)$ is strictly beneficial.\end{proof}

\noindent \textbf{\underline{Claim 2}}: On any interval where  $G(\cdot)$ is strictly increasing, $u(\cdot)$ must be weakly convex.

\begin{proof}
Suppose, to the contrary, that there exists an interval on which $G(\cdot)$ is strictly increasing but $u(\cdot)$ is not weakly convex. Then, there exists a sub-interval $I=[\underline{v},\overline{v}]$ and $\widehat{v}\in (\underline{v},\overline{v})$ such that $u(\widehat{v})>l(\widehat{v})$, where $l(\cdot)$ is the chord joining $(\underline{v},u(\underline{v}))$ and $(\overline{v},u(\overline{v}))$ defined in eq. (\ref{eq_chord}). It follows that (i) the concave envelope $h(\cdot)$ of $u(\cdot)$ on $I$ is non-affine as it satisfies $h(v)=u(v)=l(v)$ for $v\in\{\underline{v},\overline{v}\}$ and $h(\widehat{v})\geq u(\widehat{v})>l(\widehat{v})$. Furthermore, (ii) $G(v)-G(\underline{v})>0$ and $G(\overline{v})-G(v)>0$ for any $v\in (\underline{v},\overline{v})$ since $G(\cdot)$ is strictly increasing. Together, points (i) and (ii) contradict \underline{Claim 1}.\end{proof}

\noindent\textbf{\underline{Claim 3}:} Let $I=(v_l,v_h)$ be an open subinterval of a partial-disclosure interval (possibly the whole) satisfying $G(v_h)>G(v_l)$. If $u(\cdot)$ is weakly convex on $I$, then $u(\cdot)$ is affine on $I$.

\begin{proof}[Proof Sketch] The complete proof is in the Online Appendix (pg. 2-3). Consider first a closed, strict subinterval $I'=[v_l',v_h']\subset I$ satisfying $G(v_h')>G(v_l')$. Construct $\widetilde G_\epsilon(\cdot)$ from $G(\cdot)$ by reallocating a fraction $\epsilon\in(0,1)$ of the mass in $I'$ to its endpoints in a mean-preserving way, leaving $G(\cdot)$ otherwise unchanged. By construction, (1) $\widetilde{G}_\epsilon(v)\neq G(v)$ only on interval $I'$, where $H(\cdot)$ is bounded strictly away from 0, and (2) $\widetilde{G}_\epsilon(\cdot)\rightarrow G(\cdot)$ uniformly on $I'$. Together (1) and (2) imply $\widetilde{G}_\epsilon\in MPC(F)$ for small $\epsilon>0$, and thus $\widetilde{G}_\epsilon$ is  feasible.

Let $L(\cdot)$ denote the chord of $u(\cdot)$ on $[v_l',v_h']$ (similar in form to \eqref{eq_chord}). Since $E_{\widetilde G_\epsilon}[v]=E_G[v]$ and $L(\cdot)$ is affine, we have $E_{\widetilde G_\epsilon}[L(v)]=E_G[L(v)]$. Direct substitution yields $E_{\widetilde G_\epsilon}[u(v)]-E_G[u(v)]=E_{\widetilde G_\epsilon}[u(v)-L(v)]-E_G[u(v)-L(v)]=\epsilon\int_{v_l'}^{v_h'}[L(v)-u(v)]dG(v)$. Since $G(\cdot)$ is continuous and $G(v_h')>G(v_l')$, $G(\cdot)$ must have a strictly positive mass in the interior of $I'$. Furthermore, $u(\cdot)$ is weakly convex on $I'$, and if it were also non-affine, then $L(v)>u(v)$ for $v\in (v_l',v_h')$. Thus non-affine $u(\cdot)$ implies $E_{\widetilde G_\epsilon}[u(v)]-E_G[u(v)]=\epsilon\int_{v_l'}^{v_h'}[L(v)-u(v)]dG(v)>0$, and $\widetilde{G}_\epsilon$ is a profitable deviation. Therefore, $u(\cdot)$ must be affine on $I'$. The argument extends to $I$ by exhaustion, using nested closed sub-intervals $\{I_n'\}_{n=1}^{\infty}$ whose union is $I$. \end{proof}

\noindent\textbf{\underline{Claim 4}}: Let $0<v_1<v_2<1$ such that $G(v_1)=G(v_2)=k\in (0,1)$ and $v_1=\min\{v:G(v)=k\}$ and $v_2=\max\{v:G(v)=k\}$. Then, $v_2=r$ and $k=G(r)$.

\begin{proof}[Proof Sketch] The complete proof is in the Online Appendix (pg. 3). Suppose, to the contrary, that $v_2\neq r$. For arbitrarily small $\epsilon_1,\epsilon_2>0$, extend the gap $[v_1,v_2]$ to an interval $I=[\underline v,\overline v]$, where $\underline v$ is the largest value satisfying $G(\underline v)=G(v_1-\epsilon_1)$ and $\overline v$ is the smallest value satisfying $G(\overline v)=G(v_2+\epsilon_2)$. By construction, $G(\underline v)<G(v)<G(\overline v)$ for every $v\in(\underline v,\overline v)$. Let $h(\cdot)$ be the concave envelope of $u(\cdot)$ on $I$ and $l(\cdot)$ be the endpoint chord of $u(\cdot)$ on $I$ (as in \eqref{eq_chord}). Because $u(\cdot)$ is upper semicontinuous, it coincides with its concave envelope at the endpoints, i.e., $h(\underline v)=u(\underline v)=l(\underline v)$ and $h(\overline v)=u(\overline v)=l(\overline v)$. Thus, if $h(\cdot)$ is affine, then $h(\cdot)=l(\cdot)$ on $I$. Since $G(\cdot)$ is constant on $[v_1,v_2]$ and $v_2\neq r$, there exists $v_2'<v_2$ such that $u(v_2')=u(v_2)$. Note that $l(v_2')$ is a convex combination of $u(\underline v)$ and $u(\overline v)$. Because $\underline v<v_1$, we have $u(\underline v)<u(v_1)\leq u(v_2)$. Furthermore, as $\epsilon_2\rightarrow0$, $\overline v\rightarrow v_2$ and $u(\overline v)\rightarrow u(v_2)$. By implication, $l(v_2')$ approaches a convex combination of $u(\underline v)<u(v_2)$ and $u(v_2)$ with positive weight on $u(\underline{v})$. Therefore, $l(v_2')<u(v_2')$ for small $\epsilon_2$. Thus, $l(\cdot)$ does not majorize $u(\cdot)$ and cannot be the concave envelope.  Therefore, $h(\cdot)$ is non-affine, contradicting \underline{Claim 1}.
\end{proof}

 \noindent {\textbf{\underline{Claim 5}}: $v_B\in\{0,r\}$.

\begin{proof}
Suppose, to the contrary, that $v_B\notin\{0,r\}$. With $G\in MPC(F)$ and $F(\cdot)$ supported on $[0,1]$, we must have $v_B\geq 0$. Thus, $v_B\neq 0$ implies $v_B>0$. Furthermore, with $G(v)=0$ for  $v\leq v_B$, we have $H(v_B)=\int_0^{v_B}F(x)dx>0$. By \underline{Claim 4}, any gap in the support of $G(\cdot)$ ends at $r$, and therefore at most one such gap exists. Because $v_B$ is the lower bound of the support and $G(\cdot)$ is atomless, this gap, if it exists, must begin strictly above $v_B$. Hence, $G(\cdot)$ is strictly increasing on some right neighborhood of $v_B$. By continuity of $G(\cdot)$ and $H(\cdot)$, an $\epsilon>0$ exists with $v_l=v_B-\epsilon>0$, $v_h=v_B+\epsilon$, and $I=(v_l,v_h)$, such that {\it (i)} $I$ lies within a partial disclosure interval with $H(v)>0$ for $v\in (v_l,v_h)$, and {\it (ii)} $G(\cdot)$ is strictly increasing on $(v_B,v_h)$. By equation \eqref{eq_u} and \underline{Claim 2}, $u(\cdot)$ is strictly increasing and weakly convex on $(v_B,v_h)$. Furthermore, equation \eqref{eq_u} implies that $u(\cdot)$ is constant on $(v_l,v_B]$ and continuous at $v_B$. Combining, we have that $u(\cdot)$ is weakly convex and non-affine on $I=(v_l,v_h)$; because $G(v_h)>G(v_l)$, this contradicts \underline{Claim 3}. Therefore, $v_B\in\{0,r\}$.
\end{proof}

\noindent {\textbf{\underline{Claim 6}}: $v_T>r$, or equivalently, $G(r)<1$. 
\begin{proof} 
Suppose, to the contrary, that $v_T\leq r$. Then, by eq. (\ref{eta}), $\eta=1$, and by eq. (\ref{alphatilde}), $\widetilde\alpha=\alpha$. Thus, eq. (\ref{eq_u}) gives $u(\cdot)=G(\cdot)^{n-1}$. By \underline{Claim 4} and \underline{Claim 5}, the support of $G(\cdot)$ begins at $v_B=0$ and contains no gaps, i.e., it is the interval $[0,v_T]$, where $v_T\leq r<1$. Note that $H(v)=H(1)-\int_{v}^1 [F(x)-G(x)]dx=\int_{v}^1 [G(x)-F(x)]dx$ since by Lemma \ref{lem:spread}, $H(1)=0$. Therefore, for $v\in[v_T,1)$, $H(v)=\int_v^1[1-F(x)] dx>0$. Thus, $v_T$ belongs to a partial-disclosure interval with left endpoint $v_l<v_T$. Because $G(\cdot)$ is strictly increasing on $(v_l,v_T)$, \underline{Claim 2} and \underline{Claim 3} imply that $G(\cdot)^{n-1}$ is affine on $(v_l,v_T)$. Furthermore, by Lemma \ref{lem:spread}, $H(v_l)=0$ implies $G(v_l)=F(v_l)$, while $G(v_T)^{n-1}=1>F(v_T)^{n-1}$. Strict convexity of $F(\cdot)^{n-1}$ therefore implies $G(v)^{n-1}>F(v)^{n-1}$ for every $v\in(v_l,v_T)$. Hence, $G(v)>F(v)$ throughout this interval, which, together with $H(v_l)=0$, implies $H(v)<0$ for $v\in(v_l,v_T)$. This contradicts $G\in MPC(F)$.\end{proof}

\smallskip
\noindent \textbf{\underline{Claim 7}}: Let $v_L\equiv \min\{v\in[0,1]:G(v)=G(r)\}$. Then, (i) $G(\cdot)$ is strictly increasing on $[0,v_L]$, (ii) $G(\cdot)$ is strictly increasing on $[r,v_T]$, and (iii) $v_L<r$. 

\begin{proof}[Proof Sketch] The complete proof is in the Online Appendix (pg. 4). \noindent\textit{Part (i):} If $G(r)=0$, then $v_L=0$ and the claim holds trivially. Suppose therefore that $G(r)>0$, which implies $v_B<r$. By \underline{Claim 5}, $v_B=0$. If $G(\cdot)$ were not strictly increasing on $[0,v_L]$, there would exist a gap $(v_1,v_2)\subset[0,v_L]$ in the support of $G(\cdot)$, with $G(v_1)=G(v_2)=k\in(0,G(r))$. By \underline{Claim 6}, $G(r)<1$, so $k\in(0,1)$. By \underline{Claim 4}, any such gap must end at $r$ and satisfy $G(v_1)=G(r)$. With $v_1<v_L$, this contradicts the definition of $v_L$. \noindent\textit{Part (ii):} analogous to Part (i), applying \underline{Claim 4} to any gap in $[r,v_T]$.
\noindent\textit{Part (iii):} Suppose that $v_L=r$. Combining parts (i) and (ii) and continuity of $G(\cdot)$, we have that $G(\cdot)$ is strictly increasing on $[0,v_T]$. Next, note that $r>0$ by assumption and $r<v_T$ by \underline{Claim 6}, i.e. $r\in(0,v_T)$. Thus, $G(\cdot)$ is strictly increasing on $[r-\epsilon,r+\epsilon]$, for sufficiently small $\epsilon>0$. By \underline{Claim 2}, $u(\cdot)$ must be weakly convex on $[r-\epsilon,r+\epsilon]$. But \eqref{eq_u} implies that $u(\cdot)$ has an upward jump discontinuity at $r$, violating weak convexity.\end{proof}

\noindent\textbf{\underline{Claim 8}:}} Let $I$ be a partial disclosure interval. Then, $I\cap(v_L,r)\neq\varnothing$.

\begin{proof}[Partial Proof] Suppose $I=(v_1,v_2)$ is a partial disclosure interval and $I\cap(v_L,r)$ is empty, i.e, either  $0\leq v_1<v_2\leq v_L$ (case a) or $r\leq v_1<v_2$ (case b). We present a proof that covers both cases in the Online Appendix (pg. 5). To convey the idea, here we focus on case (a). 

\noindent\textit{Step 1:} We show that $u(\cdot)$ is affine on $I$. In case (a), $v_1<v_2\leq v_L$. By \underline{Claim 7}, $G(\cdot)$ is strictly increasing on $I$. By \underline{Claim 2}, $u(\cdot)$ is weakly convex on $I$. By assumption, $I$ is a partial disclosure interval. From \underline{Claim 3}, $u(\cdot)$ is affine on $I$.

Let $u_F(\cdot)$ be constructed by maintaining $\eta$, $r$, and $\widetilde{\alpha}$, but replacing $G(\cdot)$ with $F(\cdot)$, 
\begin{equation}\label{eq_uF}
u_F(v)=
\begin{cases}
(\frac{\widetilde{\alpha}}{\eta}+(1-\widetilde{
\alpha}))F(v)^{n-1} & \text{if } v< r\\
\widetilde{\alpha}+(1-\widetilde{\alpha})F(v)^{n-1} & \text{if } v\geq r.
\end{cases}
\end{equation}

\noindent\textit{Step 2:} We show that $u(v_1)=u_F(v_1)$ and $u(v_2)= u_F(v_2)$. Because $v_1$ is an endpoint of a partial disclosure interval, we have $v_1\in \mathcal{D}$, which implies $F(v_1)=G(v_1)$ by Lemma \ref{lem:spread}. Comparing \eqref{eq_u} and \eqref{eq_uF}, we have $u(v_1)=u_F(v_1)$. The argument is identical for $v_2$.

\noindent\textit{Step 3:} We show that $u_F(\cdot)$ is strictly convex on $I$. In case (a), $v_1<v_2\leq v_L$, and $v_L<r$ by \underline{Claim 7}. On $I$, therefore, \(u_F(\cdot)\) is a positive rescaling of \(F(\cdot)^{n-1}\) and is strictly convex.

\noindent\textit{Step 4:} To complete the proof, we show that $H(\cdot)<0$ on $I$, contradicting $G\in MPC(F)$. Let $l(\cdot)$ be a chord joining $(v_1,u_F(v_1))$ and $(v_2,u_F(v_2))$. From Step 3, $l(v)>u_F(v)$ for $v\in I=(v_1,v_2)$. By Step 1 and 2, $u(\cdot)$ is affine on $I$, and  $u(v_1)=u_F(v_1)$, $u(v_2)=u_F(v_2)$. Thus, $u(v)=l(v)>u_F(v)$ for all $v\in I$. Together with \eqref{eq_uF} and $v_2\leq v_L<r$, $u(v)>u_F(v)$ implies $F(v)<G(v)$ for $v\in I$. Therefore, $H'(\cdot)<0$ on $I$. With $H(v_1)=0$, we have $H(\cdot)<0$ on $I$.\end{proof}

Let $v_\dagger\equiv\max\{v\leq v_L:H(v)=0\}$ and $v_H\equiv\min\{v\geq r:H(v)=0\}$. Continuity of $H(\cdot)$ implies that the sets are closed, while $H(0)=H(1)=0$ implies that they are nonempty, and thus $v_\dagger$ and $v_H$ exist. By construction, $0\leq v_\dagger\leq v_L<r\leq v_H\leq 1$ ($v_L<r$ by \underline{Claim 7}). 

 \noindent\textbf{\underline{Claim 9}:} Suppose $v_\dag<v_L$. Then (i) $G(v)=(F(v_\dag)^{n-1}+\gamma (v-v_\dag))^{\tfrac{1}{n-1}}$ for $v\in[v_\dag,v_L]$, and (ii)  $0<\gamma\leq [F(v_\dag)^{n-1}]'$. 
\begin{proof}\textit{Point (i)}: By construction, $v_\dag<v_L$ implies $H(\cdot)>0$ on $(v_\dag,v_L)$, which is therefore a sub-interval of a partial disclosure interval. By \underline{Claim 7}, the only gap in the support of $G(\cdot)$ is $(v_L,r)$, and therefore $G(\cdot)$ is strictly increasing on $(v_\dag,v_L)$. Applying \underline{Claims 2 and 3}, $u(\cdot)$ must be affine on $(v_\dag,v_L)$, i.e. $u(v)=u(v_\dag)+\widetilde{\gamma}(v-v_\dag)$. Note that $v_\dag<v_L<r$, so only the branch of the payoff where $v<r$ is relevant. Using the definition of $u(\cdot)$, we have $u(v_\dag)=A G(v_\dag)^{n-1}$, where $A\equiv \tfrac{\widetilde{\alpha}}{\eta}+(1-\widetilde{\alpha})$. By construction, $H(v_\dag)=0$, which in turn implies $G(v_\dag)=F(v_\dag)$, and therefore $u(v_\dag)=A F(v_\dag)^{n-1}$. For $v\in (v_\dag,v_L)$, we similarly have $AG(v)^{n-1}=AF(v_\dag)^{n-1}+\widetilde{\gamma}(v-v_\dag)$, which in turn implies $G(v)=(F(v_\dag)^{n-1}+\gamma(v-v_\dag))^{\tfrac{1}{n-1}}$, where $\gamma=\widetilde{\gamma}/A$. The expression extends to $v_L$ by continuity. 

\noindent\textit{Point (ii):}  By \underline{Claim 7}, $G(\cdot)$ is strictly increasing on $[0,v_L]$, which implies $\gamma>0$. Suppose $\gamma>[F(v_\dagger)^{n-1}]'$. Then an $\epsilon>0$  exists such that $
\gamma>[F(v)^{n-1}-F(v_\dagger)^{n-1}]/[v-v_\dagger]$ 
for all $v\in(v_\dagger,v_\dagger+\epsilon)$. As $v>v_\dagger$, we have
$
F(v)^{n-1}<F(v_\dagger)^{n-1}+\gamma(v-v_\dagger)=G(v)^{n-1},
$
which implies $F(v)<G(v)$. Thus, $H'(v)<0$ on $(v_\dagger,v_\dagger+\epsilon)$. Since $H(v_\dagger)=0$, we have $H(v)<0$ on $(v_\dagger,v_\dagger+\epsilon)$, violating $G\in MPC(F)$.\end{proof}

\noindent\textbf{\underline{Claim 10}:}  We show (i) $H(v)>0$ if and only if $v\in(v_\dag,v_H)$, (ii) $v_H>r$, and (iii) $G(v)=F(v)$ for all $v\in[0,v_\dag]\cup[v_H,1]$.
\begin{proof} 
\noindent\textit{Step 1:} We show that $H(v)=0$ for $v<v_\dagger$ and $v>v_H$. Suppose, to the contrary, that $H(v^*)>0$ for some $v^*<v_\dagger$ or $v^*>v_H$, i.e., $v^*$ lies in a partial disclosure interval $(v_1,v_2)$. Because $H(v_\dagger)=H(v_H)=0$, we must have either (a) $v_2\leq v_\dagger\leq v_L$ or (b) $v_1\geq v_H\geq r$. By \underline{Claim 8}, every partial disclosure interval must intersect $(v_L,r)$, so that $(v_1,v_2)\cap(v_L,r)\neq\varnothing$.  Thus, $v_L<v_2$ and $v_1<r$, which contradicts (a) or (b).

\noindent\textit{Step 2:} We show that if  $v_\dag=v_L$, then $H(v)>0$ for all $v\in(v_\dag,r]$. By Step 1, $H(v)=0$ on $[0,v_\dag]$, and hence $G(v_\dag)=F(v_\dag)$. Since $v_\dag=v_L$, $G(\cdot)$ is constant on $(v_\dag,r)$, while $F(\cdot)$ is strictly increasing. Thus, $H'(v)=F(v)-G(v)>0$ for all $v\in(v_\dag,r)$. With $H(v_\dag)=0$, we have $H(v)>0$ on $(v_\dag,r]$.

\noindent\textit{Step 3:} We show that if $v_\dagger<v_L$, then $H(v)>0$ for all $v\in(v_\dagger,r]$. With $F(\cdot)^{n-1}$ strictly convex, it lies strictly above its tangent line at $v_\dag$, i.e. $F(v)^{n-1}>F(v_\dag)^{n-1}+[F(v_\dag)^{n-1}]'(v-v_\dag)$. From \underline{Claim 9}, we have $[F(v_\dag)^{n-1}]'\geq \gamma$, and hence $F(v)^{n-1}>F(v_\dag)^{n-1}+\gamma (v-v_\dag)=G(v)^{n-1}$  on $(v_\dag,v_L]$. It follows that $F(\cdot)>G(\cdot)$ on $(v_\dag,v_L)$. Furthermore, because $G(\cdot)$ is continuous and constant on $[v_L,r]$ while $F(\cdot)$ is increasing, we have $H'(v)=F(v)-G(v)>0$ on $(v_\dag,r]$. With $H(v_\dag)=0$, we have $H(v)>0$ on $(v_\dag,r]$.

\noindent\textit{Step 4}: We show that $v_H>r$. From Steps 2 and 3, we have $H(v)>0$ for all $v\in(v_\dagger,r]$, and in particular, $H(r)>0$. By construction, $H(v_H)=0$ and $v_H\geq r$. Therefore, $v_H>r$.

\noindent Steps 1-4 prove (i) and (ii). Point (iii) follows from point (i) and Lemma \ref{lem:spread}.\end{proof}

\noindent\textbf{\underline{Claim 11}:}  $G(v)=\min\{(G(v_L)^{n-1}+\beta (v-r)),1\}^{\tfrac{1}{n-1}}$ for $v\in[r,v_H]$, with $\beta>0$.
 \begin{proof}[Proof Sketch] Analogous to \underline{Claim 9} point (i), see Online Appendix (pg. 6).\end{proof}

\noindent Together \underline{Claims  7, 9, 10, and 11}} establish Proposition \ref{proposition_A1}.\end{proof}

\begin{proposition}
If $G(\cdot)$ is a symmetric equilibrium, then $v_{\dagger}=v_L$.
\label{proposition_A2}
\end{proposition}

\begin{proof} By Proposition \ref{proposition_A1}, $G(\cdot)$ satisfies
eq. (\ref{eq_G_0}). The corresponding payoff $u(\cdot)$ in eq. (\ref{eq_u}) satisfies the regularity condition (i) of \cite{dworczak2019simple}: it is upper semicontinuous, has only finitely many one-sided jump discontinuities, and has bounded slope between
them.\footnote{Given eq. (\ref{eq_G_0}), $u(\cdot)$ has a single upward
jump at $r$. On intervals of partial disclosure, $G(\cdot)^{n-1}$, and
hence $u(\cdot)$, is affine with bounded slope. Outside the partial-disclosure region
$G(\cdot)=F(\cdot)$ and the slope of $u(\cdot)$ is bounded because $f(\cdot)$ is bounded on the entire support.} By Theorem 2 of \cite{dworczak2019simple}, for every optimal $G$ there
exists a continuous, weakly convex function $\phi(\cdot)$ such that (DM1) $\phi(\cdot)$ is continuous and weakly convex, (DM2) $\phi(v)\geq u(v)$ for all $v\in[0,1]$, (DM3) $\text{support}(G)\subseteq \{v\in [0,1]: u(v)=\phi(v)\}$, (DM4) $\int_0^1 \phi(v)dG=\int_0^1 \phi(v)dF$. We prove the proposition by showing that if $v_{\dagger}<v_L$, there is no $\phi(\cdot)$ that satisfies these necessary conditions.

\noindent {\it Step 1:} We show that $\phi(\cdot)$ must be affine on $[v_{\dagger},v_H]$. By eq. (\ref{eq_G_0}), $G(v)=F(v)$ for $v\in [0,v_\dagger]\cup [v_H,1]$, and thus DM4 reduces to $\int_{v_\dagger}^{v_H} \phi(v)dG(v)=\int_{v_\dagger}^{v_H} \phi(v)dF(v)$. Since by Proposition \ref{proposition_A1}, $H(v)>0$ for all $v\in(v_{\dagger},v_H)$, and by DM1 $\phi(\cdot)$ is continuous and convex, this equality holds iff $\phi(\cdot)$ is affine on $[v_{\dagger},v_H]$.\footnote{Since $H(v_{\dagger})=H(v_H)=0$, $H(v)>0$ for all $v\in (v_\dagger,v_H)$, and $\phi(\cdot)$ is convex on $[0,1]$, integration by parts yields $\int_{v_\dagger}^{v_H} \phi(v)dF(v)-\int_{v_\dagger}^{v_H} \phi(v)dG(v)=\int_{v_\dagger}^{v_H} H(v)d\phi'_{+}(v)=0$ iff $d\phi'_+(\cdot)=0$ on $(v_\dagger,v_H)$, i.e., $\phi(\cdot)$ is affine on $[v_\dagger,v_H]$. See also \citet[Appendix A.5, Lemma 2]{dworczak2019simple}.}

\noindent {\it Step 2:} We establish that 
\begin{equation}
(\alpha\eta+(1-\alpha))\int_{v_{\dagger}}^{v_H}\phi(v)dF(v)<\int_{v_{\dagger}}^{v_H} F(v)^{n-1}dF(v).
\label{eq_P2_1}
\end{equation}
DM3 and eqs. (\ref{eq_u}) and (\ref{eq_G_0}) yield $u(v)=\phi(v)=\left(\frac{\widetilde{\alpha}}{\eta}+(1-\widetilde{\alpha})\right)(F(v_{\dagger})^{n-1}+\gamma(v-v_{\dagger}))$ on $v\in [v_{\dagger},v_L]$. By {\it Step 1}, $\phi(\cdot)$ is affine on $[v_\dagger,v_H]$, so this expression extends to the entire interval $[v_\dagger,v_H]$. Substituting for $\widetilde{\alpha}$ from eq. (\ref{alphatilde}) and multiplying both sides by a factor $\alpha\eta+(1-\alpha)$ yields $(\alpha\eta+(1-\alpha))\phi(v)=F(v_{\dagger})^{n-1}+\gamma(v-v_{\dagger})$. Recall from the proof of Proposition \ref{proposition_A1} (Claim 9) that $\gamma\leq [F(v_{\dagger})^{n-1}]'$. Thus, the strict convexity of $F(\cdot)^{n-1}$ implies that for all $v\in(v_\dagger,v_H]$,$(\alpha\eta+(1-\alpha))\phi(v)=F(v_{\dagger})^{n-1}+\gamma(v-v_{\dagger})<F(v)^{n-1}$. Since $F(\cdot)$ has full support, integrating both sides yields (\ref{eq_P2_1}).

\noindent {\it Step 3:} We establish that
\begin{equation}
(\alpha\eta+(1-\alpha))
\int_{v_\dagger}^{v_H}u(v)dG(v)
\geq
\int_{v_\dagger}^{v_H}F(v)^{n-1}dF(v).
\label{eq_P2_2}
\end{equation}
By eqs. (\ref{alphatilde}) and (\ref{eq_u}),
\[
(\alpha\eta+(1-\alpha))u(v)
=
\begin{cases}
G(v)^{n-1}, & v<r,\\
G(v)^{n-1}+\alpha(\eta-G(v)^{n-1}), & v\geq r.
\end{cases}
\]
Therefore,
\[
(\alpha\eta+(1-\alpha))\int_{v_\dagger}^{v_H}u(v)dG(v)=\int_{v_\dagger}^{v_H}G(v)^{n-1}dG(v)+\alpha\int_r^{v_H}(\eta-G(v)^{n-1})dG(v).
\]
Note that $G(v)=F(v)$ for $v\in \{v_\dagger,v_H\}$ implies $\int_{v_\dagger}^{v_H}G(v)^{n-1}dG(v)=\frac{G(v_H)^n-G(v_\dagger)^n}{n}=\frac{F(v_H)^n-F(v_\dagger)^n}{n}=\int_{v_\dagger}^{v_H}F(v)^{n-1}dF(v)$. Moreover, by eq. (\ref{eta}), $\eta=\frac{1-G(r)^n}{n(1-G(r))}\geq \frac{G(v_H)^n-G(r)^n}{n(G(v_H)-G(r))}=\frac{\int_r^{v_H} G(v)^{n-1}dG(v)}{G(v_H)-G(r)}$, which in turn implies $\int_{r}^{v_H} (\eta-G(v)^{n-1})dG(v)\geq 0$. Combining the two observations gives (\ref{eq_P2_2}). 

\noindent Combining (\ref{eq_P2_1}) and (\ref{eq_P2_2}) with $\int_{v_{\dagger}}^{v_H} \phi(v)dF(v)=\int_{v_{\dagger}}^{v_H} \phi(v)dG(v)$ established in {\it Step 1}, yields $\int_{v_\dagger}^{v_H}u(v)dG(v)>\int_{v_\dagger}^{v_H}\phi(v)dG(v)$. This contradicts  $u(v)=\phi(v)$ for all $v\in Supp(G)$. Therefore, no such $\phi(\cdot)$ exists when $v_\dagger<v_L$, establishing that $v_\dagger=v_L$. \end{proof}

\begin{proof}[\textbf{Proof of Proposition \ref{prop_G_shape}}] Proposition \ref{prop_G_shape} follows from Propositions \ref{proposition_A1} and \ref{proposition_A2}.
\end{proof}

\begin{lemma}[Properties of $\beta^*$ and $v_H$] For each feasible pair $(v_L,r)$, let $\beta^*(v_L,r;n)$ and $v_H(v_L,r;n)$ denote the unique values associated with the distribution
$G(\cdot)$ in Lemma \ref{feas}, and define $\widetilde{\mu}=E_F[v\mid v\in(v_L,v_H)]$. Then:

\begin{itemize}

\item[1)] If $v_H<1$, then $\beta^*<[F(v_H)^{n-1}]^{\prime}$ and $\lim_{v_L\uparrow r} v_H=r$.

\item[2)] If $r<\mu$, then $\beta^*<\infty$. Furthermore, with $v_L=0$, we have $\lim_{r\uparrow \mu} \beta^*(0,r;n)=\infty$ and
\begin{itemize}
\item[a)] there exists $\widehat{n}$ such that for $n> \widehat{n}$, $v_H(0,r;n)<1$ and satisfies:

\begin{equation}
E_F[v|v<v_H]=\frac{1}{n}v_H+\frac{n-1}{n}r.
\label{eq_v_H_n}
\end{equation}

Moreover, $v_H(0,r;n)$ is strictly decreasing in $n$ and $\lim_{n\to \infty} v_H(0,r;n)\in (r,1)$.

\item[b)] $n\beta^*(0,r;n)$ decreases in $n$ with $\lim_{n\to \infty} n\beta^*(0,r;n)=0$. 
\end{itemize}
\item[3)] If $r\geq \mu$, there is a unique $\underline{v}_L\geq 0$ that solves $E_F[v|v\geq \underline{v}_L]=r$. Furthermore,  $\lim_{v_L\downarrow  \underline{v}_L} v_H(v_L,r;n)=1$ and $\lim_{v_L\downarrow \underline{v}_L} \beta^*(v_L,r;n)=\infty$.

\item[4)] $\frac{\partial \beta^*(v_L,r;n)}{\partial r}=\frac{\beta^*}{\widetilde{\mu}-r}>0$, and  $\frac{\partial \beta^*(v_L,r;n)}{\partial v_L}=-\frac{(r-v_L)\beta^*}{\widetilde{\mu}-r}\frac{f(v_L)}{F(v_H)-F(v_L)}-\frac{[F(v_L)^{n-1}]'}{\widetilde{\mu}-r}<0$ for $v_L>0$.
\end{itemize}
\label{beta_properties}
\end{lemma}
\begin{proof}
    See Online Appendix pg. 8.
\end{proof}

\noindent \textbf{Lemma \ref{feas}} and \textbf{Lemma \ref{optdu1}} are proved in the Online Appendix, pg. 6 and pg. 9, respectively.

\begin{proof}[\textbf{Proof of Proposition \ref{prop_r_exog}}]
By Lemma \ref{optdu1}, the equilibrium candidate is optimal if and only if  the associated multiplier $\phi(\cdot)$ from \eqref{eq_PHI} satisfies {\it(i)} $\phi(0)\geq 0$ for $v_L=0$, or {\it (ii)} $\phi(\cdot)$ is continuous and convex at $v_L$ for $v_L>0$. For each feasible pair $(v_L,r)$, Lemma \ref{feas} uniquely determines $\beta^*$ and $v_H$. Let $\widehat{\phi}(\cdot)=(\alpha\eta+(1-\alpha))\phi(\cdot)$, which preserves the shape of $\phi(\cdot)$. Define the tangent line to its lower branch at $v_L$ and the affine continuation of its upper branch by
\begin{equation}
\widehat{\phi}_l(v)=F(v_L)^{n-1}+[F(v_L)^{n-1}]'(v-v_L),\qquad
\widehat{\phi}_u(v)=\alpha\eta+(1-\alpha)\bigl(G(v_L)^{n-1}+\beta^*(v-r)\bigr).
\label{phi_hat}
\end{equation}
Thus, for $v_L>0$, continuity at $v_L$ is equivalent to 
\begin{equation}
Z(v_L,r):=\widehat{\phi}_u(v_L)-\widehat{\phi}_l(v_L)
=\alpha\left(\eta(v_L)-F(v_L)^{n-1}\right)-(1-\alpha)\beta^*(r-v_L)=0.
\label{eq_continuity}
\end{equation}
\noindent{\it Step 1:} We show that every root $v_L\geq 0$ of $Z(\cdot,r)$ satisfies $(1-\alpha)\beta^*>[F(v_L)^{n-1}]'$, ensuring convexity of $\phi(\cdot)$ at $v_L>0$. Note that $\widehat{\phi}_l(\cdot)$ and $\widehat{\phi}_u(\cdot)$ are both affine with $\widehat{\phi}_l(v_L)=\widehat{\phi}_u(v_L)=\widehat{\phi}(v_L)$. Then, the two functions either coincide or intersect only once at $v_L$. Thus, to show $(1-\alpha)\beta^*>[F(v_L)^{n-1}]'$, it suffices to show that there exists $v'>v_L$ such that $\widehat{\phi}_u(v')>\widehat{\phi}_l(v')$.

Consider $\widetilde{\mu}=E_F[v|v\in(v_L,v_H)]=E_G[v|v\in(v_L,v_H)]>r>v_L$, where the equality follows from $G\in MPC(F)$ over $[v_L,v_H]$ and the inequality follows from $G(\cdot)$ having no mass on $(v_L,r)$ (Proposition \ref{prop_G_shape}). Using $G(v_L)=F(v_L)$ and $\beta^*$ from eq. (\ref{eq_BET}) in eq. (\ref{phi_hat}), 
\begin{equation}
\widehat{\phi}_u(\widetilde{\mu})=\alpha\eta+(1-\alpha)E_F[F(v)^{n-1}|v\in (v_L,v_H)]\geq E_F[F(v)^{n-1}|v\in (v_L,v_H)],
 \label{phi_h_bound}
\end{equation}
where the last inequality follows from $G(r)=F(v_L)$ and $\eta=\frac{1-G(r)^n}{n(1-G(r))}=\frac{1-F(v_L)^n}{n(1-F(v_L))}=E_F[F^{n-1}(v)|v> v_L]\geq E_F[F(v)^{n-1}|v\in (v_L,v_H)]$. 

Turning to $\widehat{\phi}_l(v)$, note that since it corresponds to the tangent line of $F(\cdot)^{n-1}$ at $v_L$, the strict convexity of $F(\cdot)^{n-1}$ implies that $\widehat{\phi}_l(v)<F(v)^{n-1}$ for $v>v_L$. Thus, accounting for the affine shape of $\widehat{\phi}_l(\cdot)$ yields:

\begin{equation}
\widehat{\phi}_l(\widetilde{\mu})=E_F[\widehat{\phi}_l(v)\mid v\in(v_L,v_H)]<E_F[F(v)^{n-1}|v\in (v_L,v_H)].
\label{phi_l_bound}
\end{equation}
Combining (\ref{phi_h_bound}) and (\ref{phi_l_bound}) yields $\widehat{\phi}_u(\widetilde{\mu})> \widehat{\phi}_l(\widetilde{\mu})$, establishing $(1-\alpha)\beta^*>[F(v_L)^{n-1}]'$.

\noindent{\it Step 2:}  We next show that there exists a unique $\underline r\in(0,\mu)$ such that $v_L^{eq}(r)=0$ if and only if $r\leq \underline r$. By Lemma \ref{feas}, a candidate with $v_L=0$ exists if and only if $r<\mu$ and by {\it (i)} and eq. (\ref{eq_continuity}),
\[
\phi(0)\geq 0
\quad\Longleftrightarrow\quad
Z(0,r)=\frac{\alpha}{n}-(1-\alpha)\beta^*(0,r;n)r\geq 0.
\]
Note that $\lim_{r\downarrow0}Z(0,r)=\alpha/n>0$, while $Z(0,r)$ is strictly decreasing in $r$ because $\frac{\partial \beta^*(0,r;n)}{\partial r}>0$ by Lemma \ref{beta_properties} {\it 4)}. Moreover, by Lemma \ref{beta_properties}, $\lim_{r\uparrow \mu} \beta^*(0,r;n)=\infty$, hence $\lim_{r\uparrow \mu}Z(0,r)=-\infty$. Therefore, there exists a unique $\underline{r}\in(0,\mu)$ such that $Z(0,\underline{r})=0$. It follows that $v_L^{eq}(r)=0$ is an equilibrium for $r\leq\underline{r}$, while any equilibrium for $r>\underline{r}$ must satisfy $v_L^{eq}(r)>0$.

\noindent\textit{Step 3:} We next show that $v_L^{eq}(r)>0$ exists if and
only if $r>\underline{r}$, and $v_L^{eq}(r)$ is unique. By Step 1, at every root $v_L\geq 0$, $(1-\alpha)\beta^*>[F(v_L)^{n-1}]'$. Thus,
\[
Z_{v_L}(v_L,r)=\alpha\eta'(v_L)+(1-\alpha)\beta^*-\alpha[F(v_L)^{n-1}]'-(1-\alpha)(r-v_L)\frac{\partial\beta^*}{\partial v_L}>0,
\]
since $\eta'(v_L)\geq 0$ and $\partial\beta^*/\partial v_L\leq 0$ by Lemma \ref{beta_properties}(4). Therefore, every root is an upward crossing, so there is at most one root. For $r\leq\underline{r}$, a strictly positive root would require $Z(0,r)<0$, contrary to Step 2. Therefore, $v_L^{eq}(r)=0$ is the unique equilibrium for $r\leq \underline{r}$.  

For $r>\underline{r}$, $Z(\cdot,r)$ is negative near the lower feasible boundary: $Z(0,r)<0$ for $r<\mu$, while for $r\geq\mu$, Lemmas \ref{feas} and \ref{beta_properties} imply $Z(v_L,r)\to-\infty$ as $v_L\downarrow\underline v_L(r)$, where $E_F[v\mid v\geq\underline v_L(r)]=r$.
Moreover, by Lemma \ref{beta_properties}, $\beta^*$ remains bounded as $v_L\uparrow r$, so
$\lim_{v_L\uparrow r}Z(v_L,r)=\alpha(\eta(r)-F(r)^{n-1})>0$. Therefore, by continuity, there exists a unique root $v_L^{eq}(r)>0$ for $r>\underline{r}$. Finally, $Z_r(v_L,r)<0$ by Lemma \ref{beta_properties}(4), so $\frac{dv_L^{eq}(r)}{dr}=-\frac{Z_r(v_L^{eq},r)}{Z_{v_L}(v_L^{eq},r)}>0$ for $r>\underline{r}$.

\noindent{\it Step 4:} Next we establish the comparative statics of $\underline{r}(n)$. To see that $\underline{r}(n)\in (0,\mu)$ is increasing in $n$, recall from Lemma \ref{beta_properties} that $n\beta^*(0,r;n)$ decreases in $n$ and increases in $r$. Therefore, the result follows immediately from implicit differentiation of $Z(0,\underline{r})=\frac{1}{n}\left(\alpha-(1-\alpha)n\beta^*(0,\underline{r};n)\underline{r}\right)=0$. Moreover, by the same lemma, $\lim_{n\to \infty} n\beta^*(0,r;n)=0$ for $r<\mu$, which implies that $\lim_{n\to \infty} \underline{r}(n)=\mu$.

\noindent{\it Step 5:} Finally, we show the convergence to full disclosure at extreme values of $r$. As $r\to 1$, feasibility in Lemma \ref{feas} requires $E_F[v|v\geq v_L]>r$, which implies $v_L^{eq}(r)\to1$. Since $v_L<r<v_H\leq 1$, eq. (\ref{eq_G}) then implies $G(\cdot)\to F(\cdot)$ pointwise.

As $r\to 0$, $v_L^{eq}(r)=0$. Suppose, to the contrary, that along
some sequence $r\to0$, $v_H\to\bar v_H>0$. By eq. (\ref{eq_G}),
the limiting $G(\cdot)^{n-1}$ is concave on $[0,\bar v_H]$ and
coincides with $F(\cdot)^{n-1}$ at the endpoints. Since $F(\cdot)^{n-1}$ is strictly convex, this implies $G(v)^{n-1}>F(v)^{n-1}$ throughout the interior, and hence $G(v)>F(v)$, violating the MPC condition in Proposition \ref{prop_G_shape}. Therefore $v_H\to 0$, and eq. (\ref{eq_G}) again implies $G(\cdot)\to F(\cdot)$ pointwise. 
\end{proof}

\begin{lemma}(Search Equation). For each $v_L\in[0,1)$, let $r^{se}(v_L)$ be the unique value of $r$ that solves \eqref{eq_search2}. Let $r^{fi}$ be the reservation value in a benchmark with full disclosure, i.e., $r^{fi}$ is the unique value of $r$ that solves $\int_{r}^1(v-r)dF=s$. Then, (1) $r^{se}(0)=\mu-s$ and $r^{se}(r^{fi})=r^{fi}$; (2) $r^{se}(\cdot)$ is continuous, differentiable, and has a single peak at $v_L=r^{fi}$.
\label{lemma_search_eq}
\end{lemma}

\begin{proof}[\textbf{Proof of Lemma \ref{lemma_search_eq}}]
The first property is immediate from eq. (\ref{eq_search2}). The second property follows from the continuity and differentiability of eq. (\ref{eq_search2}) in $v_L$ and $r$, with implicit differentiation yielding 
\begin{equation}
\frac{d r^{se}(v_L)}{dv_L}=(r^{se}-v_L)\frac{f(v_L)}{1-F(v_L)}.
\label{eq_r_fi_slope}
\end{equation}
Hence, $r^{se}(v_L)$ is increasing for $v_L<r^{fi}$ and decreasing for $v_L>r^{fi}$.  
\end{proof}

\begin{proof}[\textbf{Proof of Proposition \ref{prop_r_endog}}]
We establish equilibrium existence and uniqueness in sequence. 

\medskip
\noindent\textbf{Existence.} An equilibrium satisfies $r^{se}(v_L^{eq}(r^*))=r^*$. Recall from Lemma \ref{lemma_search_eq} that $r^{se}(0)=\mu-s$. If $\mu-s\leq \underline{r}(n,\alpha)$, Proposition \ref{prop_r_exog} implies $v_L^{eq}(\mu-s)=0$, so $(r^*,v_L^*)=(\mu-s,0)$ is an equilibrium for $\mu-s\leq \underline{r}(n,\alpha)$ .

\noindent Suppose instead that $\mu-s>\underline r(n,\alpha)$. By Proposition \ref{prop_r_exog}, $0<v_L^{eq}(\mu-s)<\mu-s$. Since $r^{se}(\cdot)$ is increasing on $[0,r^{fi}]$ and $r^{se}(0)=\mu-s$, Lemma \ref{lemma_search_eq} implies $r^{se}(v_L^{eq}(\mu-s))>\mu-s$. Moreover, by Proposition \ref{prop_r_exog}, $v_L^{eq}(r^{fi})<r^{fi}$, so $r^{se}(v_L^{eq}(r^{fi}))<r^{se}(r^{fi})=r^{fi}$. Hence, by continuity, there exists $r^*\in(\mu-s,r^{fi})$ such that $r^{se}(v_L^{eq}(r^*))=r^*$. This establishes equilibrium existence for $\mu-s>\underline{r}(n,\alpha)$ with $r^*\in (\mu-s,r^{fi})$ and $v_L^*>0$ since $r^*>\underline{r}(n,\alpha)$.

\medskip
\noindent\textbf{Uniqueness.}
Let $r^{eq}(v_L)=[v_L^{eq}]^{-1}(v_L)$ denote the inverse on the interior branch. It suffices to show that at any interior intersection, $\frac{d r^{eq}(v_L^*)}{dv_L}>\frac{d r^{se}(v_L^*)}{dv_L}$. From the proof of Proposition \ref{prop_r_exog}, interior equilibria satisfy $Z(v_L,r^{eq}(v_L))=0$, where $Z(v_L,r)$ is given by eq. (\ref{eq_continuity}). Implicit differentiation  gives $\frac{d r^{eq}}{dv_L}=-\frac{Z_{v_L}(v_L,r^{eq}(v_L))}{Z_r(v_L,r^{eq}(v_L))}$. Define $\Psi(v_L):=Z(v_L,r^{se}(v_L))$. In equilibrium, $r^{se}(v_L^*)=r^{eq}(v_L^*)=r^*$, and hence $\Psi(v_L^*)=0$. Moreover,
\begin{equation}
\frac{d\Psi(v_L^*)}{dv_L}= Z_{v_L}(v_L^*,r^*)+Z_r(v_L^*,r^*)\frac{dr^{se}(v_L^*)}{dv_L}=-Z_r(v_L^*,r^*)\left(\frac{dr^{eq}(v_L^*)}{dv_L}-\frac{dr^{se}(v_L^*)}{dv_L}\right).
\label{psi_vL}
\end{equation}
Since by eq. (\ref{eq_continuity}), $Z_r(\cdot,r)<0$, it remains to show that $\frac{d\Psi(v_L^*)}{dv_L}>0$. Taking the total derivative of $Z(v_L^*,r^{se}(v_L^*))$ given by eq. (\ref{eq_continuity}) yields

\[
\frac{d\Psi(v_L^*)}{dv_L}=\alpha(\eta'(v_L^*)-[F(v_L^*)^{n-1}]')+(1-\alpha)\left(\beta^*\left(1-\frac{dr^{se}(v_L^*)}{dv_L}\right)-(r^*-v_L^*)\left(\frac{\partial \beta^*}{\partial v_L^*}+\frac{\partial \beta^*}{\partial r^*}\frac{d r^{se}(v_L^*)}{d v_L}\right)\right).
\]
By eq. (\ref{eq_r_fi_slope}), $\frac{d r^{se}(v_L^*)}{d v_L}=(r^*-v_L^*)\frac{f(v_L^*)}{1-F(v_L^*)}$. Moreover, $\eta(v_L)=\frac{\int_{v_L}^1 F(v)^{n-1}f(v)dv}{1-F(v_L)}$ and $\eta'(v_L)=\frac{f(v_L)}{1-F(v_L)}[\eta(v_L)-F(v_L)^{n-1}]$. Using the equilibrium condition $Z(v_L^*,r^{eq}(v_L^*))=0$, $\eta'(v_L^*)$ can be rewritten as
\[
\eta'(v_L^*)= \frac{1-\alpha}{\alpha}\beta^*(r^*-v_L^*)\frac{f(v_L^*)}{1-F(v_L^*)}.
\]
Substituting for $\frac{d r^{se}(v_L^*)}{d v_L}$ and $\eta'(v_L^*)$ in $\frac{d\Psi(v_L^*)}{dv_L}$, and simplifying yields
\[
\frac{d\Psi(v_L^*)}{dv_L}=(1-\alpha)\beta^*-\alpha [F(v_L^*)^{n-1}]'-(1-\alpha)(r^*-v_L^*)\left(\frac{\partial \beta^*}{\partial v_L}+\frac{\partial \beta^*}{\partial r}\frac{dr^{se}(v_L^*)}{dv_L}\right).
\]
By the convexity of $\phi(\cdot)$ at $v_L^*$ (Lemma \ref{optdu1}),$(1-\alpha)\beta^*> [F(v_L^*)^{n-1}]'>\alpha [F(v_L^*)^{n-1}]'$ and from Lemma \ref{beta_properties} {\it 4)} and eq. (\ref{eq_r_fi_slope}), the term in brackets is given by
\[
\frac{\partial \beta^*}{\partial v_L}+\frac{\partial \beta^*}{\partial r}\frac{dr^{se}}{dv_L}=-\frac{[F(v_L^*)^{n-1}]'}{\widetilde{\mu}-r^*}+\frac{\beta^*(r^*-v_L^*)}{\widetilde{\mu}-r^*}f(v_L^*)\left(\frac{1}{1-F(v_L^*)}-\frac{1}{F(v_H^*)-F(v_L^*)}\right).
\]
Since $F(v_H^*)\leq 1$, we have $\frac{1}{1-F(v_L^*)}-\frac{1}{F(v_H^*)-F(v_L^*)}\leq 0$, so the entire expression is negative. Hence, $\frac{d\Psi(v_L^*)}{dv_L}>0$, and by eq. (\ref{psi_vL}), $\frac{d r^{eq}(v_L^*)}{dv_L}>\frac{d r^{se}(v_L^*)}{dv_L}$, establishing the uniqueness of $r^*$. 
\end{proof}

\begin{proof}[\textbf{Proof of Proposition \ref{prop_large_n}}]
Let $G_n(\cdot)$ denote the equilibrium distribution given $n$ firms.

\noindent {\it 1) Competition and Disclosure:} Recall by Proposition \ref{prop_r_endog} that $v_L^*=0$ if and only if $\underline{r}(n,\alpha)\geq \mu-s$. By Proposition \ref{prop_r_exog}, $\underline{r}(n,\alpha)$ is increasing in $n$ with $\lim_{n\to \infty} \underline{r}(n,\alpha)=\mu$. Therefore, by continuity, there exists $\underline{n}\in [2,\infty)$ such that $\underline{r}(n,\alpha)\geq \mu-s$, yielding the result. 

\noindent {\it 2) Competition and Informativeness:} Let $n'>n>\underline{n}$. Then, by 1), $G_{n}(v)=G_{n'}(v)=0$ for $v\leq r^*(n)=r^*(n')=\mu-s$. Moreover, on the common affine region, $G_{n'}(v)/G_n(v)\propto (v-r^*)^{1/(n'-1)-1/(n-1)}$, which is strictly decreasing and diverges to infinity as $v\downarrow r^*$; hence, $G_{n'}(v)>G_{n}(v)$ immediately above $r^*$. Together with $v_H^*(n')\leq v_H^*(n)$ (Lemma \ref{beta_properties}) and eq. \eqref{eq_G}, this implies that $G_{n'}(\cdot)$ can cross $G_n(\cdot)$ at most once and only from above as illustrated in Figure \ref{fig_large_n}. Such a crossing must exist since the distributions have the same mean. Therefore, $G_{n'}\in MPS(G_n)$, i.e., higher $n$ increases informativeness.

 \noindent {\it 3) Paradox of Choice:} \Inexperienced{} consumers continue to search whenever $v<r^*$. From Proposition \ref{Gstruc},  $\Pr(v<r^*)=G(r^*)=G(v_L^*)>0\Leftrightarrow v_L^*>0$. From 1), $v_L^*>0\Leftrightarrow n<\underline{n}$.

\noindent {\it 4) Rich get Richer:} With $n'>n>\underline{n}$, $G_{n}\in MPC(G_{n'})$ (Part 2). By standard properties of convex order/order statistics, 
$\CSsav(n)\leq E_{G_{n'}}[\max\{v_i\}_{i=1}^n]<E_{G_{n'}}[\max\{v_i\}_{i=1}^{n'}]=\CSsav(n')$.

\noindent {\it 5) The Poor get Poorer:} Given $v_n\sim G_n(\cdot)$,  $r$ the associated equilibrium reservation value,
and $\widetilde{v}_n=\min\{v_n,r\}$, we have 
\begin{equation*}
E_{G_n}[\widetilde v_n]=\int_0^{v_L}v\,dF(v)+(1-F(v_L))r=\mu-\int_{v_L}^1(v-r)\,dF(v)=\mu-s.
\end{equation*}
For $n'\geq \underline{n}$, there is no disclosure at the bottom, i.e., the support of $G_{n'}(\cdot)$ is weakly above its reserve value $r'=\mu-s$. Thus, $\widetilde{v}=\mu-s$ (i.e., it is degenerate) and $\CSin(n')=\mu-s$. For $n<\underline{n}$, in contrast, $G_{n}(\cdot)$ includes mass below $r$, (on $[0,v_L]$), so $\widetilde{v}$ is nondegenerate. Thus, $\CSin(n)=E_{G_{n}}[\max\{\widetilde{v}_1,...,\widetilde{v}_{n}\}]>E[\widetilde{v}_i]=\mu-s=\CSin(n')$. \end{proof}

\begin{proof}[\textbf{Proof of Proposition \ref{com_stat_s}}]
Let $S(r)=\int_{v_L^{eq}(r)}^1(v-r)dF(v)$. By eq. \eqref{eq_search2}, $r^*(s)$ satisfies $S(r^*(s))=s$. First, we show that $\lim_{s\to 0} r^*(s)=1$. Suppose, to the contrary, that along some sequence $s_k\to0$, $r^*(s_k)\to\bar r<1$. Then, by continuity of $S(\cdot)$, $\lim_{k\to \infty} S(r^*(s_k))=S(\overline{r})=0$. By Lemma \ref{feas}, however, $E_F[v|v>v_L^{eq}(\overline{r})]>\overline{r}$ for $\overline{r}<1$, and thus $S(\overline{r})>0$, a contradiction. Therefore, $\lim_{s\to 0} r^*(s)=1$. By Proposition \ref{prop_r_exog}, $r^*(s)\to 1$ implies $G(\cdot)\to F(\cdot)$, establishing the pointwise convergence for $s\to 0$. To establish the other extreme, note that by Proposition \ref{prop_r_endog}, for $s\geq \mu-\underline{r}=\overline{s}$, $v_L^*(s)=0$ and $r^*(s)=\mu-s$. Then, clearly, $\lim_{s\to \mu} r^*(s)=0$ and by Proposition \ref{prop_r_exog}, $G(\cdot)\to F(\cdot)$ as $s\to \mu$. To show that informativeness increases with $s$ for $s>\overline{s}$, consider first $v_H^*<1$, where by Proposition \ref{proposition_A1}, $G(v_H^*)=F(v_H^*)\Longrightarrow \beta^*(v_H^*-r^*)=F(v_H^*)^{n-1}$. Implicit differentiation and applying Lemma \ref{beta_properties} 1), 4), and eq. \eqref{eq_v_H_n} yields
\[
\frac{dv_H^*}{ds}
=
\frac{n-1}{1-\frac{[F(v_H^*)^{n-1}]'}{\beta^*}}<0.
\]

\noindent Therefore, $r^*(s)$ and $v_H^*(s)$ decrease in $s$. Therefore, for any $s'>s>\overline{s}$, $G_{s'}(\cdot)$ single-crosses $G_s(\cdot)$ from above (see Figure \ref{fig_change_s}) and, since they have the same mean, $G_{s'}\in MPS(G_s)$.

Alternatively, if $v_H^*=1$, by eq. (\ref{eq_beta_thresh}), $\beta^*=\frac{1}{n(\mu-r^*)}=\frac{1}{ns}$ and thus $v_T^*=(n-1)s+\mu$. Since $r^*(s)$ decreases in $s$, while $v_T^*(s)$ increases in $s$ and the distribution $G_{s}(\cdot)^{n-1}$ is affine on the entire support, it is immediate that an increase in $s$ results in a wider support and a clockwise rotation of the distribution, which implies that $G_{s'}\in MPS(G_{s})$ for $s'>s>\overline{s}$. It follows immediately that $\CSsav(s')>\CSsav(s)$. In contrast, $\CSin=\mu-s$ for $s>\overline{s}$ since $v_L^*(s)=0$, and thus $\CSin$ is decreasing in s.
\end{proof}

\newpage
\section{Online Appendix}\setcounter{page}{1}

\subsection{Proofs}

\subsection*{Detailed Proofs for Claims in Proposition \ref{proposition_A1}}

\begin{proof}[Proof of \underline{Claim 1}] We prove that if an interval $I=[\underline{v},\overline{v}]$ exists such that $G(\underline{v})<G(v)<G(\overline{v})$ for any $v\in (\underline{v},\overline{v})$ and $h(\cdot)$ is non-affine on $I$, then there exists $\widetilde{G}(\cdot)\in MPC(F)$ such that $\mathbb{E}_{\widetilde{G}}[u(v)]>\mathbb{E}_{G}[u(v)]$, establishing a profitable deviation from $G(\cdot)$.

With $G(\cdot)$ continuous in its support, the payoff $u(\cdot)$ (from \eqref{eq_u}) is (i) bounded, (ii)  weakly increasing, with a single upward jump discontinuity at $v=r$, where it is right-continuous. Together (i) and (ii) imply that $u(\cdot)$ is upper semicontinuous.

Consider the chord joining $(\underline{v},u(\underline{v}))$ and $(\overline{v},u(\overline{v}))$ as in \eqref{eq_chord}.
Let $d(v)\equiv u(v)-l(v)$. Because $u(\cdot)$ is upper semicontinuous and $l(\cdot)$ continuous, $d(\cdot)$ is upper semicontinuous, attaining a maximum on compact interval $I$. Let $M\equiv \max_{v\in I}d(v)$ and $v^*\in\text{argmax}_{v\in I}{d(v)}$. 

\noindent \textit{Step 1:} We show that $M>0$ and $v^*\in(\underline{v},\overline{v})$. By construction, $d(\underline{v})=d(\overline{v})=0$, so $M\geq 0$. Suppose that $M=0$. Then $u(v)\leq l(v)$ for all $v\in I$, so $l(\cdot)$ is a concave majorant of $u(\cdot)$. Since $h(\cdot)$ is the smallest concave majorant, $h(v)\leq l(v)$ for all $v\in I$. Conversely, $h(\cdot)\geq u(\cdot)$ on $I$ implies $h(v)\geq l(v)=u(v)$ for $v\in \{\underline{v},\overline{v}\}$ and therefore $h(\underline{v})=l(\underline{v})$ and $h(\overline{v})=l(\overline{v})$. Hence, $l(\cdot)$ is  a chord of $h(\cdot)$. The concavity of $h(\cdot)$ then implies $h(v)\geq l(v)$ for all $v\in I$. Thus, $h(v)=l(v)$ on $I$, contradicting the assumption that $h(\cdot)$ is non-affine. Therefore, $M>0$. Since $d(\underline{v})=d(\overline{v})=0$, any $v^*\in\arg\max_{v\in I}d(v)$ must satisfy $v^*\in(\underline{v},\overline{v})$.

\noindent \textit{Step 2:} We show that $h(v^*)=u(v^*)$. Consider $L(v)=l(v)+M$. Note that $L(v)\geq l(v)+d(v)=u(v)$ for all $v\in I$ and $L(\cdot)$ is (weakly) concave. Thus $L(\cdot)$ is a concave majorant of $u(\cdot)$, which implies $h(v)\leq L(v)$ for all $v\in I$. In particular, $h(v^*)\leq L(v^*)=l(v^*)+M=l(v^*)+d(v^*)=u(v^*)$. Furthermore,  $h(\cdot)$ is the concave envelope of $u(\cdot)$, which implies $h(v^*)\geq u(v^*)$. Combining, we have $h(v^*)=u(v^*)$.

\noindent \textit{Step 3:} Let $\epsilon_1(t)=(1-t)(v^*-\underline{v})$, $\epsilon_2(t)=t(\overline{v}-v^*)$, $S(t)= (\underline{v},\underline{v}+\epsilon_1)\cup(\overline{v}-\epsilon_2,\overline{v})$ and $s(t)=E_G[v|v\in S(t)]$. We show that  $t^*\in(0,1)$ exists such that $s(t^*)=v^*$. First, note that $s(\cdot)$ is well-defined because the conditions $G(\overline{v})-G(v)>0$ and  $G(v)-G(\underline{v})>0$ ensure that $S(t)=(\underline{v},\underline{v}+\epsilon_1)\cup(\overline{v}-\epsilon_2,\overline{v})$ contains nonzero mass, i.e., $p(t)=\Pr(v\in S(t))>0$ for all $t\in[0,1]$. Second, note that $s(0)=E_G[v|v\in(\underline{v},v^*)]<v^*<s(1)=E_G[v|v\in (v^*,\overline{v})]$, and $s(\cdot)$ is continuous. The claim follows from the intermediate value theorem.

\noindent\textit{Step 4:} Let $t^*$ be such that $s(t^*)=v^*$. We show that $E_G[h(v)|v\in S(t^*)]<h(v^*)$. By construction, $E_G[v|v\in S(t^*)]=v^*$ and $h(\cdot)$ is concave. By construction, the convex hull of the support of $G(\cdot\mid v\in S(t^*))$ is $[\underline{v},\overline{v}]$. Furthermore, $h(\cdot)$ is non-affine on $[\underline{v},\overline{v}]$. Thus, Jensen's inequality is strict, $E_G[h(v)|v\in S(t^*)]< h(v^*)$.

\noindent\textit{Step 5.} To complete the proof, we construct $\widetilde{G}(\cdot)\in MPC(G)\subseteq MPC(F)$, such that $E_{\widetilde{G}}[u(v)]>E_G[u(v)]$. Define a new random variable $\widetilde{v}=v\mathbbm{1}\{v\notin S(t^*)\}+v^*\mathbbm{1}\{v\in S(t^*)\}$, and let $\widetilde{G}(\cdot)$ be the cdf of $\widetilde{v}$. Thus, $\widetilde{G}(\cdot)$ replaces all mass in $S(t^*)$ with its expected value $v^*$, leaving the distribution otherwise unchanged. By construction, $\widetilde{G}(\cdot)\in MPC(G)\subseteq MPC(F)$. Furthermore, $E_{\widetilde{G}}[u(v)]-E_{G}[u(v)]=p^*\{u(v^*)-E_G[u(v)|v\in S(t^*)]\}$, where $p^*=\Pr_G(v\in S(t^*))>0$. We must therefore show $E_G[u(v)|v\in S(t^*)]<u(v^*)$. Using Steps 2, 4, and $u(v)\leq h(v)$, we have \(E_G[u(v)|v\in S(t^*)]\leq E_G[h(v)|v\in S(t^*)]<h(v^*)=u(v^*)\).
\end{proof}

\begin{proof}[Proof of \underline{Claim 3}] Steps 1 and 2 establish the result for a closed subinterval $I'=[v_l',v_h']\subset I$ with $G(v_h')-G(v_l')>0$. Such a subinterval exists because $G(\cdot)$ is atomless and $G(v_h)>G(v_l)$, implying that the interior of $I$ has strictly positive $G$-mass.

Consider the following family of distributions indexed by $\epsilon\in(0,1)$,
 \begin{equation*}
\widetilde{G}_\epsilon(v)=
\begin{cases}
G(v) & \text{for } v< v_l'\\
(1-\epsilon)G(v)+\epsilon(\frac{1}{v'_h-v'_l})\int_{v'_l}^{v'_h}G(x)dx) & \text{for } v\in [v'_l,v'_h)\\
G(v) & \text{for } v\geq v_h'
\end{cases}
\end{equation*}
It is straightforward to verify that for $\epsilon\in(0,1)$, such $\widetilde{G}_\epsilon(\cdot)$ is a valid cdf. Such distributions reallocate a fraction $\epsilon$ of the mass inside $(v_l',v_h')$ to the interval's endpoints while preserving the mean, leaving the distribution otherwise unchanged.

\noindent\textit{Step 1:}   We show that for $\epsilon>0$ and sufficiently small, $\widetilde{G}_\epsilon\in MPC(F)$. To do so, we first consider $D_\epsilon(v)=\int_0^v\widetilde{G}_\epsilon(x)-G(x)dx$. By routine simplification, we have
\begin{align*}
D_\epsilon(v)=\begin{cases}
0 & \text{for } v< v_l'\\
\epsilon \Delta(v) &\text{for }v\in[v_l',v_h')\\
0 & \text{for } v\geq v_h'
\end{cases},
\quad
\Delta(v)=(v-v_l')\biggl(\frac{1}{v_h'-v_l'}\int_{v_l'}^{v_h'}G(x)dx-\frac{1}{v-v_l'}\int_{v_l'}^vG(x)dx\biggr)
\end{align*}
Let $\widetilde{H}_\epsilon(v)=\int_0^v F(x)-\widetilde{G}_\epsilon(x)dx=\int_0^v F(x)-G(x)dx-[\int_0^v\widetilde{G}_\epsilon(x)-G(x)dx]=H(v)-D_\epsilon(v)$. Because $[v_l',v_h']$ lies strictly inside a partial disclosure interval, we have $H(\cdot)>0$ and continuous on $[v_l',v_h']$, and therefore $\underline{H}\equiv \min_{[v_l',v_h']}H(v)>0$. Similarly, $\Delta(v)$ is continuous on $[v_l',v_h']$ and $\overline{\Delta}\equiv \max_{[v_l',v_h']}\Delta(v)<\infty$. Thus, on $[v_l',v_h']$, we have
$\widetilde H_\epsilon(v)\geq\underline{H}-\epsilon\overline{\Delta}> 0$
for all sufficiently small $\epsilon>0$. For $v\notin[v'_l,v'_h]$, we have $D_\epsilon(v)=0$, which implies $\widetilde{H}_\epsilon(v)=H(v)\geq 0$. To complete Step 1, we must show $\widetilde{H}_\epsilon(1)=0$. With $v_h'<v_h\leq 1$, we have $D_\epsilon(1)=0$, and thus $\widetilde H_\epsilon(1)=H(1)$. Noting that $G\in MPC(F)$ implies $H(1)=0$.

\noindent\textit{Step 2:} We show that $u(\cdot)$ must be affine on $I'$. Following Step 1, select $\epsilon>0$ sufficiently small that $\widetilde{G}_\epsilon\in MPC(F)$  and omit the subscript. Since $\widetilde{G}(\cdot)$ is a feasible deviation from $G(\cdot)$, it should not be strictly profitable, i.e., $E_{\widetilde{G}}[u(v)]\leq E_G[u(v)]$. Let $L(\cdot)$ be the chord of $u(\cdot)$ that joins $(v_l',u(v_l'))$ and $(v_h',u(v_h'))$ (similar in form to \eqref{eq_chord}). Extend the domain of $L(\cdot)$ to the entire interval $[0,1]$. Because $G\in MPC(F)$ and $\widetilde{G}\in MPC(F)$, we have $E_G[v]=E_{\widetilde{G}}[v]$. Because $L(\cdot)$ is affine, we have $E_G[L(v)]=E_{\widetilde{G}}[L(v)]$. Therefore, 
\[E_{\widetilde{G}}[u(v)]-E_G[u(v)]=E_{\widetilde{G}}[u(v)-L(v)]-E_G[u(v)-L(v)]=\int_{0}^{1}[u(v)-L(v)]d\widetilde{G}-\int_{0}^{1}[u(v)-L(v)]dG.\]
Continuing by direct substitution, we have
\[E_{\widetilde{G}}[u(v)]-E_G[u(v)]=-\epsilon\int_{v_l'}^{v_h'}[u(v)-L(v)]dG+\epsilon(q-G(v_l'))[u(v_l')-L(v_l')]+\epsilon(G(v_h')-q)[u(v_h')-L(v_h')],\]
where
 \(q\equiv\frac{1}{v'_h-v'_l}\int_{v'_l}^{v'_h}G(v)dv.\)
By construction, $L(\cdot)$ is the chord of $u(\cdot)$ on $[v_l',v_h']$, which implies $u(v_l')=L(v_l')$ and $u(v_h')=L(v_h')$. Substituting and rearranging,
\[
E_{\widetilde{G}}[u(v)]-E_G[u(v)]=\epsilon\int_{v_l'}^{v_h'}[L(v)-u(v)]dG.
\]
Suppose, contrary to the claim, that $u(\cdot)$ is not affine on $I'$. Because $u(\cdot)$ is weakly convex on $I$, it is weakly convex on $I'$. Hence its endpoint chord satisfies $L(v)>u(v)$ for every $v\in(v_l',v_h')$. Moreover, because $G(\cdot)$ is atomless and $G(v_h')>G(v_l')$, there exists a closed subinterval $[v_l'',v_h'']\subset(v_l',v_h')$ such that $G(v_h'')>G(v_l'')$. Since $L(\cdot)-u(\cdot)$ is continuous and strictly positive on $[v_l'',v_h'']$,
\[
\delta\equiv\min_{v\in[v_l'',v_h'']}[L(v)-u(v)]>0.
\]
It follows that
\[
E_{\widetilde{G}}[u(v)]-E_G[u(v)]=\epsilon\int_{v_l'}^{v_h'}[L(v)-u(v)]dG\geq \epsilon\delta[G(v_h'')-G(v_l'')]>0,
\]
which, in turn contradicts the optimality of $G(\cdot)$. Therefore, $u(\cdot)$ is affine on $I'$.

\noindent\textit{Step 3.} We next extend the local result established in Steps 1 and 2 to the entire interval $I$. Because $G(\cdot)$ is atomless and $G(v_h)>G(v_l)$, there exists a closed interval $I_1'=[v_l',v_h']$ lying strictly inside $I$ that has strictly positive $G$-mass. Let $\{I_n'\}_{n=1}^\infty$ be an increasing sequence of closed intervals lying strictly inside $I$ such that $I_1'\subseteq I_n'\subset I$ and $\bigcup_n I_n'=I$. Because the intervals are nested, every $I_n'$ contains $I_1'$ and therefore has strictly positive mass. Moreover, weak convexity of $u(\cdot)$ on $I$ implies weak convexity of $u(\cdot)$ on every $I_n'$. Therefore, by Steps 1 and 2, $u(\cdot)$ is affine on every $I_n'$. Because the intervals are nested, their affine representations coincide on their intersections. Hence, $u(\cdot)$ is affine on their union $I$.
\end{proof}

\begin{figure}
\centering

\begin{minipage}{0.48\textwidth}
\centering
\begin{tikzpicture}[xscale=0.85,yscale=0.25]

\def\a{2}
\def\vone{3.2}
\def\r{5}
\def\vtwo{6.8}
\def\b{7.2}

\def\ya{0.8}
\def\ylow{2.048}
\def\yr{5}
\def\yhigh{5.85}
\def\yb{5.65}
\def\vtwoprime{5.8}

\draw[blue, line width=1pt, domain=0.5*\vone:\vone]
    plot (\x,{0.2*\x^2});
\draw[blue, line width=1pt] (\vone,\ylow) -- (\r,\ylow);
\draw[blue, line width=1pt, dashed] (\r,\ylow) -- (\r,\yr);
\draw[blue, line width=1pt] (\r,\yr) -- (\vtwo,\yr);
\draw[blue, line width=1pt, domain=\vtwo:\vtwo+0.6]
    plot (\x,{4*(\x-\vtwo)^2 + \yr});

\draw[red, line width=1pt]
    (\a,\ya) -- (\r,\yr) -- (\b,\yb);
\draw[violet, line width=1pt]
    (\a,\ya)  -- (\b,\yb);

\foreach \x/\lab in {\a/{\underline{v}},\vone/{v_1},\r/{r},\vtwo/{v_2},\b/{\overline{v}}}{
    \node[below] at (\x,-0.1) {{$\lab$}};
}

\draw[black, dotted, line width=0.5pt] (\a,0) -- (\a,\ya);
\draw[black, dotted, line width=0.5pt] (\vone,0) -- (\vone,\ylow);
\draw[black, dotted, line width=0.5pt] (\r,0) -- (\r,\yr);
\draw[black, dotted, line width=0.5pt] (\vtwo,0) -- (\vtwo,\yr);
\draw[black, dotted, line width=0.5pt] (\b,0) -- (\b,\yb);

\draw[line width=0.5pt] (1.4,0) -- (1.4,6.5);
\draw[line width=0.5pt] (1.4,0) -- (7.7,0);

\node[circle, fill=violet, inner sep=0pt,
      minimum size=3pt] at (\a,\ya) {};

\node[circle, fill=violet, inner sep=0pt,
      minimum size=3pt] at (\b,\yb) {};

\end{tikzpicture}
\end{minipage}
\begin{minipage}{0.48\textwidth}
\centering
\begin{tikzpicture}[xscale=0.85,yscale=0.25]

\def\a{2}
\def\vone{3.2}
\def\r{3.2}
\def\vtwo{6.8}
\def\b{7.2}

\def\ya{0.6}
\def\ylow{4.1}
\def\yr{4.1}
\def\yhigh{4.1}
\def\yb{4.7}
\def\vtwoprime{5.8}

\draw[blue, line width=1pt, domain=0.5*\vone:\vone]
    plot (\x,{0.039*\x^4});
\draw[blue, line width=1pt] (\vone,\ylow) -- (\r,\ylow);
\draw[blue, line width=1pt, dashed] (\r,\ylow) -- (\r,\yr);
\draw[blue, line width=1pt] (\r,\yr) -- (\vtwo,\yr);
\draw[blue, line width=1pt, domain=\vtwo:\vtwo+0.6]
    plot (\x,{4*(\x-\vtwo)^2 + \yr});

\draw[red, line width=1pt]
    (\a,\ya) -- (\r,\yr) -- (\b,\yb);
\draw[violet, line width=1pt]
    (\a,\ya)  -- (\b,\yb);

\foreach \x/\lab in {\a/{\underline{v}},\vone/{v_1},\vtwo/{v_2},\b/{\overline{v}}}{
    \node[below] at (\x,-0.1) {{$\lab$}};
}

\draw[black, dotted, line width=0.5pt] (\a,0) -- (\a,\ya);
\draw[black, dotted, line width=0.5pt] (\vone,0) -- (\vone,\ylow);
\draw[black, dotted, line width=0.5pt] (\vtwo,0) -- (\vtwo,\yr);
\draw[black, dotted, line width=0.5pt] (\b,0) -- (\b,\yb);

\draw[line width=0.5pt] (1.4,0) -- (1.4,6.5);
\draw[line width=0.5pt] (1.4,0) -- (7.7,0);

\node[circle, fill=violet, inner sep=0pt,
      minimum size=3pt] at (\a,\ya) {};

\node[circle, fill=violet, inner sep=0pt,
      minimum size=3pt] at (\b,\yb) {};
      \end{tikzpicture}
\end{minipage}
\caption{\textit{Illustration of Claim 4, $\underline{v}=v_1-\epsilon_1$ and $\overline{v}=v_2+\epsilon_2$. In both panels, $u(\cdot)$ is in blue, $l(\cdot)$ in purple, and $h(\cdot)$ in red. Evidently, $l(\cdot)$ does not majorize $u(\cdot)$ and $h(\cdot)$ is non-affine. In the left panel $r\in(v_1,v_2)$, in the right panel $r\notin(v_1,v_2)$.}}
\label{fig_claim4}
\end{figure}

\begin{proof}[Proof of \underline{Claim 4}] Note that $v_1,v_2$ are well-defined because $G(\cdot)$ is continuous on a bounded domain, and therefore its level sets are compact. For any small $\epsilon_1>0,\epsilon_2>0$, consider $\underline{v}=\max\{v|G(v)=G(v_1-\epsilon_1)\}$ and $\overline{v}=\min\{v|G(v)=G(v_2+\epsilon_2)\}$. Such $\underline{v}$ and $\overline{v}$ are well-defined because $G(\cdot)$ is continuous and its domain is bounded, so that its level sets are compact. Consider the interval $I=[\underline{v},\overline{v}]$ constructed in this way.
    
    \noindent\textit{Step 1:}  For any small $\epsilon_1>0,\,\epsilon_2>0$, we show that $G(v)-G(\underline{v})>0$ and $G(\overline{v})-G(v)>0$ for all $v\in (\underline{v},\overline{v})$. If $v>\underline{v}$, then $G(v)\geq G(\underline{v})$, as $G(\cdot)$ is weakly increasing. Furthermore, the construction of $\underline{v}$ requires $G(v)\neq G(\underline{v})$, as $\underline{v}$ is the largest value satisfying $G(\underline{v})=G(v_1-\epsilon_1) $. Consequently, $G(v)>G(\underline{v})$. The proof for $G(\overline{v})-G(v)>0$ is analogous.

\noindent\textit{Step 2:} Let $h(\cdot)$ be the concave envelope of $u(\cdot)$ on $I$. We show that if $v_2\neq r$ and $\epsilon_2$ is sufficiently small, then $h(\cdot)$ is non-affine. Suppose, to the contrary, that $h(\cdot)$ is affine on $I$. Recall that with $G(\cdot)$ continuous in its support, the payoff $u(\cdot)$ (from \eqref{eq_u}) is (i) bounded, (ii)  weakly increasing, with a single upward jump discontinuity at $v=r$, where it is right-continuous. Together (i) and (ii) imply that $u(\cdot)$ is upper semicontinuous. It is well known that the concave envelope of an upper semicontinuous function on a closed interval coincides with the function at the endpoints, \(h(\underline v)=u(\underline v)\) and \(h(\overline v)=u(\overline v)\). Therefore, if $h(\cdot)$ is affine, then it must be the endpoint chord $l(\cdot)$ defined in \eqref{eq_chord}. We will show that the endpoint chord $l(\cdot)$ is not a majorant of $u(\cdot)$, which implies that $l(\cdot)$ is not the concave envelope of $u(\cdot)$, which in turn implies that $h(\cdot)$ is non-affine. Let
\(
\tilde v
=
\mathbbm{1}\{r\notin(v_1,v_2)\}v_1
+
\mathbbm{1}\{r\in(v_1,v_2)\}r.
\)
Because $v_1<v_2$ and $v_2\neq r$, we have $\tilde v\in[v_1,v_2)$, and interval $(\tilde{v},v_2]$ is non-empty. Furthermore, $r\notin(\tilde v,v_2]$ by construction. Since $G(\cdot)=k$ on $[v_1,v_2]$, \eqref{eq_u} implies that \(u(v)=u(v_2)\) for all \(v\in(\tilde v,v_2]\). Moreover, with continuous $G(\cdot)$, the only discontinuity in $u(\cdot)$ is at $v=r\notin(\tilde{v},v_2)$. With $v_2\neq r$, we have $\overline{v}\rightarrow v_2$ and $u(\overline{v})\rightarrow u(v_2)$ as $\epsilon_2\rightarrow 0$. By implication, for \(v\in(\tilde v,v_2)\),
\[
\lim_{\epsilon_2\rightarrow 0}l(v)
=
\frac{v_2-v}{v_2-\underline v}u(\underline v)
+
\frac{v-\underline v}{v_2-\underline v}u(v_2)
<
u(v_2),
\]
where the strict inequality follows from $u(\underline v)<u(v_1)\leq u(v_2)$. Thus, for sufficiently small $\epsilon_2>0$ and $v\in(\tilde{v},v_2)$, we have 
\(l(v)<u(v)=u(v_2)\). Therefore, $l(\cdot)$ is not a majorant of $u(\cdot)$.

\noindent\textit{Step 3:} To complete the proof of \underline{Claim 4}, suppose to the contrary that $v_2\neq r$. Construct $I$ for small $\epsilon_1>0$ and choose $\epsilon_2>0$ that is small enough to ensure that the concave envelope of $u(\cdot)$ on $I$ is non-affine. From Step 2, such $\epsilon_2$ exists. From Step 1, $G(\underline{v})<G(v)<G(\overline{v})$ for all $v\in (\underline{v},\overline{v})$. Existence of such an interval $I$ contradicts \underline{Claim 1}.\end{proof}

\begin{proof}[Proof of \underline{Claim 7}]
 \noindent\textit{Part (i):} If $G(r)=0$, then $v_L=0$ and the claim holds trivially. Suppose therefore that $G(r)>0$, which implies $v_B<r$. By \underline{Claim 5}, $v_B=0$. Suppose, to the contrary, that $G(\cdot)$ is not strictly increasing on $[0,v_L]$. Then there exist $0\leq v_1<v_2\leq v_L$ such that $G(v_1)=G(v_2)=k$, where $v_1$ is taken to be the smallest such value and $v_2$ the largest. Since $v_B=0$ and $v_2>0$, we have $k=G(v_2)>0$, which in turn implies $v_1>0$. Moreover, by the definition of $v_L$, we have $k=G(v_1)<G(r)<1$, where the last inequality follows from \underline{Claim 6}. Thus, $G(v_1)=G(v_2)=k\in(0,G(r))$ and $0<v_1<v_2\leq v_L\leq r<1$. Applying \underline{Claim 4}, $v_2=r$ and $k=G(r)$, contradicting $k<G(r)$.

 \noindent\textit{Part (ii):} Suppose, to the contrary, that $G(\cdot)$ is not strictly increasing on $[r,v_T]$. Then there exist $r\leq v_1<v_2\leq v_T$ such that $G(v_1)=G(v_2)=k$, where $v_1$ is again the smallest such value and $v_2$ the largest. Since $v_1<v_2\leq v_T$, we have $v_1<v_T$, and hence $k=G(v_1)<1$. Furthermore, $v_B\leq r<v_2$ implies $k=G(v_2)>0$. Since $G(v_2)=k<1=G(1)$, we also have $v_2<1$. Thus, $G(v_1)=G(v_2)=k\in(0,1)$ and $0<r\leq v_1<v_2<1$. Applying \underline{Claim 4}, $v_2=r$, contradicting $r\leq v_1<v_2$.

 \noindent\textit{Part (iii):} Suppose that $v_L=r$. Combining parts (i) and (ii) and continuity of $G(\cdot)$, we have that $G(\cdot)$ is strictly increasing on $[0,v_T]$. Next, note that $r>0$ by assumption and $r<v_T$ by \underline{Claim 6}, i.e. $r\in(0,v_T)$. Thus, $G(\cdot)$ is strictly increasing on $[r-\epsilon,r+\epsilon]$, for sufficiently small $\epsilon>0$. By \underline{Claim 2}, $u(\cdot)$ must be weakly convex on $[r-\epsilon,r+\epsilon]$. But \eqref{eq_u} implies that $u(\cdot)$ has an upward jump discontinuity at $r$, violating weak convexity.\end{proof}

 \begin{proof}[Proof of \underline{Claim 8}]  Suppose $I=(v_1,v_2)$ is a partial disclosure interval and $I\cap(v_L,r)$ is empty, i.e, either  $0\leq v_1<v_2\leq v_L$ (case a) or $r\leq v_1<v_2$ (case b). To streamline the exposition of these two cases, let $v'_2\equiv\min\{v_2,v_T\}$ and $I'\equiv(v_1,v'_2)$. Note that in case (a), we have $v_1<v_2\leq v_L<r<v_T$, where $v_L<r<v_T$ follows from \underline{Claims 6 and 7}. By implication, $v_2'=v_2$ in case (a).

\noindent\textit{Step 1}: We show that  $v_1<v_2'$  i.e., $I'$ is non-empty. In case (a), we have  $v_1<v_2'=v_2$. Consider case (b), and suppose that $v_1\geq v_T$, so that, $G(v_1)=1$. By Lemma \ref{lem:spread}, $v_1\in \mathcal{D}$ and therefore $F(v_1)=G(v_1)=1$. It follows that $F(\cdot)=G(\cdot)=1$ on $I=(v_1,v_2)$, and hence $H'(\cdot)=0$ on $I$. By Lemma \ref{lem:spread}, $H(v_1)=0$, and therefore, $H(\cdot)=0$ on $I$. Thus, $I$ is not a partial disclosure interval. Therefore, it must be that $v_1<v_T$. With $v_1<v_2$ by assumption, we have $v_1<v_2'=\min\{v_2,v_T\}$.

\noindent\textit{Step 2:} We show that $u(\cdot)$ is affine on $I'$. In case (a), $0\leq v_1<v_2'=v_2\leq v_L$. By implication $I'=(v_1,v_2')\subset [0,v_L]$. In case (b), $r\leq v_1<v_2'\leq v_T$. By implication, $I'=(v_1,v_2')\subset [r,v_T]$. In both cases, $G(\cdot)$ is strictly increasing on $I'$ by \underline{Claim 7}. By \underline{Claim 2}, $u(\cdot)$ is weakly convex on $I'$. Therefore, $I'$ is an open subinterval of partial disclosure interval $I$, and $u(\cdot)$ is weakly convex on $I'$. By \underline{Claim 3}, $u(\cdot)$ is affine on $I'$.

Let $u_F(\cdot)$ be an auxiliary function constructed by maintaining $\eta$, $r$, and $\widetilde{\alpha}$, but replacing $G(\cdot)$ with $F(\cdot)$, 
\begin{equation}\label{eq_uFOA}
u_F(v)=
\begin{cases}
(\frac{\widetilde{\alpha}}{\eta}+(1-\widetilde{
\alpha}))F(v)^{n-1} & \text{if } v< r\\
\widetilde{\alpha}+(1-\widetilde{\alpha})F(v)^{n-1} & \text{if } v\geq r.
\end{cases}
\end{equation}

\noindent\textit{Step 3:} We show that $u(v_1)=u_F(v_1)$ and $u(v_2')\geq u_F(v_2')$. Because $v_1$ is an endpoint of a partial disclosure interval, we have $v_1\in \mathcal{D}$ and $F(v_1)=G(v_1)$ by Lemma \ref{lem:spread}. Comparing \eqref{eq_u} and \eqref{eq_uFOA}, we have $u(v_1)=u_F(v_1)$. When $v_2\leq v_T$ so that $v_2'=v_2$, the argument is analogous. When $v_T<v_2$ so that $v_2'=v_T$, we have $G(v_2')=1$, and hence $u(v_2')=1= \max_{v\in[0,1]}u_F(v)\geq u_F(v_2')$.

\noindent\textit{Step 4:} We show that $u_F(\cdot)$ is strictly convex on $I'$. In case (a), we have $v_1<v_2\leq v_L<r$, where the last inequality follows from \underline{Claim 7}. In case (b), $r\leq v_1<v_2'$, where the last inequality is Step 1. In either case, \(r\notin I'\). By \eqref{eq_uFOA}, on $I'$, \(u_F(\cdot)\) is therefore a positive affine transformation of  strictly convex \(F(\cdot)^{n-1}\). Thus, \(u_F(\cdot)\) is strictly convex on \(I'\).

\noindent\textit{Step 5:} To complete the proof, we show that $H(\cdot)<0$ on $I'$, contradicting $G\in MPC(F)$. Let $l(\cdot)$ be a chord joining $(v_1,u_F(v_1))$ and $(v_2',u_F(v_2'))$. From Step 4, $l(v)>u_F(v)$ for $v\in I'=(v_1,v_2')$. By Steps 2 and 3, $u(\cdot)$ is affine on $I'$ and weakly dominates $u_F(\cdot)$ at both endpoints. Thus, $u(v)\geq l(v)>u_F(v)$ for all $v\in I'$. In case (a) $v_2'=v_2\leq v_L<r$ we have \((\frac{\widetilde{\alpha}}{\eta}+(1-\widetilde{\alpha}))F(v)^{n-1}<(\frac{\widetilde{\alpha}}{\eta}+(1-\widetilde{
\alpha}))G(v)^{n-1}\Rightarrow F(v)<G(v)\) for all $v\in I'$. In case (b) $r\leq v_1<v_2'$, we have 
 \(\widetilde{\alpha}+(1-\widetilde{\alpha})F(v)^{n-1}<\widetilde{\alpha}+(1-\widetilde{\alpha})G(v)^{n-1}\Rightarrow F(v)<G(v)\) for all $v\in I'$. Therefore $H'(\cdot)<0$ for $v\in I'$. Combined with $H(v_1)=0$, we have $H(\cdot)<0$ on $I'$.\end{proof}

\begin{proof}[Proof of \underline{Claim 11}] Let $v_T$ be the top of the support of $G(\cdot)$. Note that both $v_T>v_H$ and $v_T\leq v_H$ are possible. Let $v_a\equiv\min\{v_T,v_H\}$. By \underline{Claim 6}, $v_T>r$. By \underline{Claim 10}, $v_H>r$. Thus, $v_a>r$. By \underline{Claim 7}, $G(\cdot)$ is strictly increasing on $(r,v_a)$. Furthermore,  by \underline{Claim 10}, $H(\cdot)>0$ on $(r,v_a)$, i.e. it is a sub-interval of a partial disclosure interval. By \underline{Claims 2 and 3}, $u(\cdot)$ is affine on $(r,v_a)$, i.e. $u(v)=u(r)+\widetilde{\beta}(v-r)$ for some $\widetilde{\beta}>0$. Using \eqref{eq_u}, we find that on $(r,v_a)$ we have $\widetilde{\alpha}+(1-\widetilde{\alpha})G(v)^{n-1}=\widetilde{\alpha}+(1-\widetilde{\alpha})G(r)^{n-1}+\widetilde{\beta}(v-r)$, and hence $G(v)=[G(r)^{n-1}+\beta(v-r)]^{\frac{1}{n-1}}$, where $\beta=\widetilde{\beta}/(1-\widetilde{\alpha})$. Noting that $G(\cdot)$ has a gap on $(v_L,r)$, we have $G(r)=G(v_L)$, and therefore $G(v)=[G(v_L)^{n-1}+\beta(v-r)]^{\frac{1}{n-1}}$ on $(r,v_a)$, extending to the endpoints by continuity of $G(\cdot)$. By the definition $v_a=\min\{v_T,v_H\}$ and continuity of $G(\cdot)$, the preceding expression holds as $v$ increases from $r$, until either $v_H$ is reached or $G(\cdot)$ reaches one at $v_T$. Since $G(v)=1$ for all $v\geq v_T$, it follows that
$G(v)=\min\{G(v_L)^{n-1}+\beta(v-r),1\}^{\frac{1}{n-1}}$ for all $v\in[r,v_H]$. \end{proof}

\subsection*{Proofs of Lemmas}

\begin{proof}[\textbf{Proof of Lemma \ref{lem:spread}}]
 Since $F(\cdot)$ and $G(\cdot)$ are continuous, $H(\cdot)$ is differentiable, with $H'(v)=F(v)-G(v)$. Since $G\in MPC(F)$, $H(\cdot)\geq0$, so every $v\in\mathcal D\cap(0,1)$ minimizes $H(\cdot)$. Hence $H'(v)=0$ and $F(v)=G(v)$. Equality also holds at the boundaries since $F(0)=G(0)=0$ and $F(1)=G(1)=1$, proving (i). Continuity of $H(\cdot)$ implies that $\mathcal D=\{v:H(v)=0\}$ is closed and $\mathcal P=\{v:H(v)>0\}$ is open. Hence $\mathcal P$ is a union of disjoint open intervals $I^k$. By continuity, $H(\cdot)=0$ at the endpoints of each such interval, proving (ii).
\end{proof}

\begin{proof}[\textbf{Proof of Lemma \ref{feas}}]
Fix $r$ and $v_L<r$. By Proposition \ref{prop_G_shape}, $G(\cdot)$ in the partial disclosure region $\mathcal{P}$ is given by 

\begin{equation}
G_\beta(v)=
\begin{cases}
F(v_L) & \text{for } v\in (v_L,r)\\
\min\left\{\left(F(v_L)^{n-1}+\beta(v-r)\right),1\right\}^{\frac{1}{n-1}} & \text{for } v\geq r,
\end{cases}
\label{eq_G_trunk}
\end{equation}
where $v_L$ is the lower bound of $\mathcal{P}$. Define
\[
D(v,\beta)=\min\left\{\left(F(v_L)^{n-1}+\beta(v-r)\right),1\right\}-F(v)^{n-1},
\]
\[
H(v,\beta)=\int_{v_L}^v F(z)dz-F(v_L)(r-v_L)-\int_{r}^v G_{\beta}(z)dz.
\]
Feasibility requires $D(v_H,\beta)=0$ (i.e, $G_\beta(v_H)=F(v_H)$) and $H(v_H,\beta)=0$. Note that $D(\cdot,\beta)$ is continuous and thus attains a maximum on $[r,1]$. Let $M(\beta)=\underset{v\in[r,1]}{\max } D(v,\beta)$. Notice that $M(0)=D(r,0)<0$, while $G_{\beta}(v)=1$ for $v\geq v_T=\frac{1-F(v_L)^{n-1}}{\beta}+r$, and thus $M(\beta)\geq D(v_T,\beta)>0$ for $\beta>\frac{1-F(v_L)^{n-1}}{1-r}$. Moreover, since $D(v,\beta)$ is jointly continuous and increasing in $\beta$ for every $v$,
$M(\beta)$ is continuous and increasing in $\beta$. Therefore, there exists $\underline{\beta}>0$ such that $M(\underline{\beta})=0$ and $M(\beta)>0$ iff $\beta>\underline{\beta}$. Consequently, a solution $v_H(\beta)$ that solves $D(v_H(\beta),\beta)=0$ exists iff $\beta\geq \underline{\beta}$. Next, we focus on $\beta\geq \underline{\beta}$ and establish the existence and uniqueness of $(v_H(\beta),\beta)$ iff $E_F[v|v\geq v_L]>r$.     

\bigskip

\noindent {\it Step 1:} We establish that for a fixed $\beta>\underline{\beta}$, the highest root of $D(v,\beta)$ is the unique candidate for $v_H(\beta)$. First, note that $D(\cdot,\beta)$ is concave since $G_\beta(\cdot)^{n-1}$ is concave and $F(\cdot)^{n-1}$ is convex. Moreover, $D(r,\beta)<0$ and $D(1,\beta)\leq 0$ and $M(\beta)=\underset{v\in[r,1]}{\max } D(v,\beta)>0$ for $\beta>\underline{\beta}$. Hence, $D(\cdot,\beta)$ admits exactly two roots $v_{l}(\beta)<v_h(\beta)\leq 1$ (see Figure \ref{fig_G_star}). For all $v\in [r, v_l(\beta))$, $D(v,\beta)<0$, and thus $G(v)<F(v)$. Therefore, $H(v_l(\beta),\beta)>0$, ruling out $v_l(\beta)$ as a solution for $v_H(\beta)$. Thus, the higher root is the unique candidate for $v_H(\beta)$.

\noindent {\it Step 2:} We show that the solution for $\beta$, if such exists, must be unique. Note that for $\beta>\underline{\beta}$, $H(v_H(\beta),\beta)$ is continuous in $\beta$ since $G_\beta(\cdot)$ and $v_H(\cdot)$ are continuous functions. Thus, differentiating $H(v_H(\cdot),\cdot)$  and using the boundary condition $F(v_H(\beta))=G_{\beta}(v_H(\beta))$ yields 
\[
\frac{d}{d\beta}H(v_H(\beta),\beta)=-\int_{r}^{v_H(\beta)} \frac{\partial G_\beta(v)}{\partial \beta}dv<0.
\]
Therefore, $H(v_H(\beta),\beta)$ is monotone decreasing in $\beta$, implying a unique solution. 

\noindent {\it Step 3:} We show that $\beta$ solving $H(v_H(\beta),\beta)=0$ exists iff $E_F[v|v\geq v_L]>r$. First, note that $M(\underline{\beta})=0$ implies $D(v,\underline{\beta})\leq 0$, and thus $F(v)\geq G(v)$, for all $v\in [r,1]$. Therefore, $H(v_H(\underline{\beta}),\underline{\beta})>0$. For $\beta>\underline{\beta}$, by {\it Step 2}, $H(v_H(\beta),\beta)$ is decreasing in $\beta$. Thus, to establish the existence and uniqueness of $\beta$, it remains to show that $H(v_H(\beta),\beta)<0$ for sufficiently high $\beta$. Note that $\lim_{\beta\to \infty} G_{\beta}(v)=1$ for all $v>r$ and thus $\lim_{\beta\to \infty} v_H(\beta)=1$. Hence,
\[
\lim_{\beta\to \infty} H(v_H(\beta),\beta)=\int_{v_L}^1 F(v)dv-F(v_L)(r-v_L)-(1-r)=(1-F(v_L))(r-E_F[v|v\geq v_L]).
\]
Therefore, there exists a unique $\beta\in (\underline{\beta},\infty)$ that solves $H(v_H(\beta),\beta)=0$ iff $E_F[v|v\geq v_L]>r$. By Proposition \ref{prop_G_shape}, the unique solution for $\beta$ and $v_H(\beta)$ imply that $G(\cdot)$ is unique. 

As a final step, we establish the functional form of $\beta$. From eq. (\ref{eq_G_trunk}) we obtain:
\[
\int_{r}^{v_H} G(v)^{n-1}dG(v)=\int_{r}^{v_H} [F(v_L)^{n-1}+\beta(v-r)]dG(v),
\]
where, if $G(\cdot)$ reaches $1$ before $v_H$, it is constant thereafter,
so this region carries no $G$-mass. Recall that $\int_{r}^{v_H} G(v)^{n-1}dG(v)=\frac{G(v_H)^{n}-G(r)^{n}}{n}=\frac{F(v_H)^{n}-F(v_L)^{n}}{n}=\int_{v_L}^{v_H} F(v)^{n-1}dF(v)$ since $G(r)=F(v_L)$ and $G(v_H)=F(v_H)$, and by Bayes' plausibility on $[v_L,v_H]$, $\int_{r}^{v_H}vdG(v)=\int_{v_L}^{v_H}vdF(v)$. Therefore, the above equation can be re-written as:
\[
\int_{v_L}^{v_H} F(v)^{n-1}dF(v)=(F(v_L)^{n-1}-\beta r)(F(v_H)-F(v_L))+\beta\int_{v_L}^{v_H}vdF(v).
\]
Dividing by $F(v_H)-F(v_L)$ and rearranging terms yields \[\beta=\beta^*\equiv\frac{E_F[F(v)^{n-1}|v\in (v_L,v_H)]-F(v_L)^{n-1}}{E_F[v|v\in(v_L,v_H)]-r}\].\end{proof}

\begin{proof}[\textbf{Proof of Lemma \ref{beta_properties}}] Let $\beta(v_L,r;n,v_H)$ denote the right-hand side of eq. (\ref{eq_BET}), so that $\beta^*(v_L,r;n)=\beta(v_L,r;n,v_H(v_L,r;n))$. We establish the listed properties in sequence.

\medskip

\noindent \textbf{\it 1)} For the first property, recall that when $v_H<1$, $v_H$ is the higher root of $D(\cdot,\beta^*)$ defined in the proof of Lemma \ref{feas}. Since $D(v,\beta^*)$ is strictly concave, it is decreasing in $v$ at $v=v_H$. Hence, $D_v(v_H,\beta^*)=\beta^*-[F(v_H)^{n-1}]'<0$.

\noindent For the second property, consider $v_L\uparrow r$. Suppose, to the contrary, that $v_H$ does not converge to $r$. Then, along a convergent subsequence, $v_H\to\bar{v}_H>r$. In the limit, $G(\cdot)^{n-1}$ is concave on $[r,\bar{v}_H]$ and coincides with $F(\cdot)^{n-1}$ at the two endpoints. Since
$F(\cdot)^{n-1}$ is strictly convex, $G(v)^{n-1}>F(v)^{n-1}$ for all $v\in(r,\bar v_H)$. Therefore, $G(v)>F(v)$ on this interval, which contradicts the MPC condition in Proposition \ref{prop_G_shape}. Therefore, $\lim_{v_L\uparrow r}v_H(v_L,r;n)=r$.

\medskip

\noindent \textbf{\it 2)} If $r<\mu$, then for every $v_L<r$, $E_F[v\mid v\geq v_L]\geq\mu>r$. Hence, by Lemma \ref{feas}, $\beta^*(v_L,r;n)<\infty$. 

\noindent Consider next $v_L=0$. By eq. (\ref{eq_BET}),
\begin{equation}
\beta(0,r;n,v_H)=
\frac{F(v_H)^n}
{n\int_0^{v_H}(v-r)f(v)dv}.
\label{eq_beta_thresh}
\end{equation}

\noindent From the proof of Lemma \ref{feas}, the top of $G$'s support $v_T<1$ for $\beta>\frac{1}{1-r}$. This, in turn, implies $v_H=1$ for $\beta>\frac{1}{1-r}$. Then, eq. (\ref{eq_beta_thresh}) reduces to $\beta(0,r;n,1)=\beta^*(0,r;n)=\frac{1}{n(\mu-r)}$. Thus, $\lim_{r\uparrow \mu} \beta^*(0,r;n)=\infty$.

\noindent For part {\it a)}, note that $\beta(0,r;n,1)=\frac{1}{n(\mu-r)}<\frac{1}{1-r}$ for $n>\frac{1-r}{\mu-r}=n_0$. Then, $v_H(0,r;n)<1$ for $n> n_0$ and $v_H(0,r;n)=1$ for $n\leq n_0$. Accounting for the fact that $n$ is integer valued, let $\widehat{n}=\lfloor n_0\rfloor$, and thus $v_H(0,r;n)<1$ for $n>\widehat{n}$ and $v_H(0,r;n)=1$ for $n\leq \widehat{n}$. To establish the comparative statics below, we treat $n$ as a continuous variable.   

\noindent For $n>n_0$, $v_H<1$, and the endpoint condition $G_{\beta^*}(v_H)=F(v_H)$ yields

\begin{equation}
\beta(0,r;n,v_H)=\frac{F^{n-1}(v_H)}{v_H-r}.
\label{eq_beta_thresh_2}
\end{equation}

\noindent Equating (\ref{eq_beta_thresh}) and (\ref{eq_beta_thresh_2}) yields

\begin{eqnarray}
\int_0^{v_H} (v-r)f(v)dv&=&\frac{F(v_H)}{n}(v_H-r). \label{eq_beta_equiv}
\end{eqnarray}

\noindent Rearranging terms in eq. (\ref{eq_beta_equiv}) yields eq. (\ref{eq_v_H_n}). To show that $v_H(0,r;n)$ decreases in $n$ for $n>n_0$, we implicitly differentiate eq. (\ref{eq_beta_equiv}) to obtain

\begin{eqnarray*}
&&\left(\frac{F(v_H)}{n}-\frac{n-1}{n}f(v_H)(v_H-r)\right)\frac{d v_H}{d n}-\frac{F(v_H)(v_H-r)}{n^2}=0\\
&&\Longleftrightarrow\frac{F(v_H)}{n}\left(1-\frac{(n-1)f(v_H)F(v_H)^{n-2}}{\beta^*}\right)\frac{d v_H}{d n}-\frac{F(v_H)(v_H-r)}{n^2}=0,
\end{eqnarray*}

\noindent where the second line substitutes for $(v_H-r)$ from eq. (\ref{eq_beta_thresh_2}). Note that $(n-1)f(v_H)F(v_H)^{n-2}=[F(v_H)^{n-1}]'$ and by $1)$, $\beta^*<[F(v_H)^{n-1}]'$. Therefore,

\begin{equation}
\frac{d v_H(0,r;n)}{d n}=\frac{(v_H-r)}{n\left(1-\frac{[F(v_H)^{n-1}]'}{\beta^*}\right)}<0.
\label{eq_vH_n}
\end{equation}

\noindent Finally, by eq. (\ref{eq_v_H_n}), $\lim_{n\to \infty} E_{F}[v|v<v_H]=r<\mu$, implying $\lim_{n\to \infty} v_H(0,r;n)\in(r,1)$. 

\noindent To establish b), first note that from part a), $v_H(0,r;n)=1$ for $n\leq n_0$, and eq. (\ref{eq_beta_thresh}) yields $n\beta(0,r;n,1)=n\beta^*(0,r;n)=\frac{1}{\mu-r}$, which is constant in $n$. 

\noindent For $n>n_0$, $v_H(0,r;n)<1$ from part a). Differentiating (\ref{eq_BET}) with respect to $v_H$ and using the boundary condition $G_{{\beta^*}}(v_H)^{n-1}=F(v_L)^{n-1}+\beta^*(v_H-r)=F(v_H)^{n-1}$ yields
\begin{equation}
\frac{\partial\beta(v_L,r;n,v_H)}{\partial v_H}=\frac{f(v_H)}{\int_{v_L}^{v_H}(v-r)f(v)dv}\left[F(v_H)^{n-1}-F(v_L)^{n-1}-\beta^*(v_H-r)\right]=0.
\label{eq_BET_v_H}
\end{equation}
Therefore, when differentiating the equilibrium value of $\beta^*$ with respect to $n$, the effect through $v_H$ vanishes. Consequently, $\frac{d\beta^*(0,r;n)}{dn}=\frac{\partial \beta(0,r;n,v_H)}{\partial n}$, and by eq. (\ref{eq_beta_thresh}) 
\begin{eqnarray*}
\frac{d\left(n\beta^*(0,r;n)\right)}{dn}=\frac{F^{n}(v_H(0,r;n))\ln(F(v_H(0,r;n)))}{\int_{0}^{v_H(0,r;n)}(v-r)f(v)dv}<0,
\end{eqnarray*}
where the last inequality follows from $F(v_H(0,r;n))<1$ for $n>n_0$. This establishes that $n\beta^*(0,r;n)$ decreases in $n$. Finally, to establish that $n\beta^*(0,r;n)$ converges to zero as $n\rightarrow \infty$, note that by eq. (\ref{eq_beta_thresh_2}), 

\begin{equation}
\lim_{n\to \infty} n\beta^*(0,r;n)=\lim_{n\to \infty}\frac{nF^{n-1}(v_H(0,r;n))}{v_H(0,r;n)-r}.
\end{equation}

\noindent From a), $\lim_{n\to\infty} v_H(0,r;n)\in (r,1)$, implying that $\lim_{n\to \infty} nF^{n-1}(v_H(0,r;n))=0$, while the denominator is strictly positive in the limit. Therefore, $\lim_{n\to \infty} n\beta^*(0,r;n)=0$.

\medskip

\noindent \textbf{\it 3)} 
For $r\geq\mu$, the uniqueness of $\underline v_L$ satisfying $E_F[v|v\geq\underline v_L]=r$ follows from the strict monotonicity of $E_F[v|v\geq v_L]$ in $v_L$. Consider $v_L\downarrow\underline v_L$. If $v_H$ did not converge to $1$, then there would exist a sequence of $v_L$'s approaching $\underline v_L$ along which $v_H$ remains bounded away from $1$. Along a convergent subsequence, let $v_H\to\overline v_H<1$. Then $\widetilde{\mu}\to E_F[v|v\in(\underline{v}_L,\overline{v}_H)]<E_F[v|v\geq \underline{v}_L]=r$. This contradicts Lemma \ref{feas}, which requires $\widetilde{\mu}>r$. Hence, $\lim_{v_L\downarrow\underline{v}_L}v_H(v_L,r;n)=1$ and therefore $\widetilde{\mu}\to r$. By eq. (\ref{eq_BET}), the denominator of $\beta^*$ converges to zero from above, while its numerator converges to $E_F[F(v)^{n-1}\mid v\geq\underline v_L]-F(\underline v_L)^{n-1}>0$. Thus, $\lim_{v_L\downarrow\underline v_L}\beta^*(v_L,r;n)=\infty$.

\medskip

\noindent \textbf{\it 4)} The proof follows by differentiating $\beta(v_L,r;n,v_H)$ given by eq. (\ref{eq_BET}) with respect to $r$ and $v_L$, accounting for $\frac{\partial \beta(v_L,r;n,v_H)}{\partial v_H}=0$ as established by eq. (\ref{eq_BET_v_H}) for $v_H<1$. \end{proof}

\begin{proof}[\textbf{Proof of Lemma \ref{optdu1}}]
Recall from the proof of Proposition \ref{proposition_A2} that DM1-DM4 are necessary and sufficient for the optimality of $G(\cdot)$. We first show that DM3-DM4, together with the equilibrium properties of $G(\cdot)$ identified in Proposition \ref{prop_G_shape}, imply that the multiplier supporting the equilibrium $G(\cdot)$ must satisfy eq. \eqref{eq_PHI}. Then, we show that the remaining conditions DM1-DM2 reduce to {\it (i)} and {\it (ii)}. 

\noindent \textit{Step 1}: We show that if $G(\cdot)$ is an equilibrium, then any multiplier $\phi(\cdot)$ that supports $G(\cdot)$ must be \eqref{eq_PHI}. Recall that DM3 requires $\phi(\cdot)=u(\cdot)$ on the support of $G(\cdot)$, while Step 1 in the proof of Proposition \ref{proposition_A2} establishes that, for any supporting multiplier, DM4 holds if and only if $\phi(\cdot)$ is affine on the entire partial-disclosure interval $(v_L,v_H)$. Together, these conditions uniquely determine $\phi(\cdot)$ as given by eq. (\ref{eq_PHI}).

\noindent \textit{Step 2}: We show that for $G(\cdot)$ that satisfies Proposition \ref{prop_G_shape} and $\phi(\cdot)$ that satisfies eq. \eqref{eq_PHI}, DM1-DM4 reduce to (C1) $\phi(v)\geq u(v)$ on $[v_L,r)$; and (C2)  $\phi(\cdot)$ is continuous and convex at $v_L$ whenever $v_L>0$. From Step 1, any supporting multiplier must satisfy eq. \eqref{eq_PHI}. Conversely, eq. \eqref{eq_PHI} ensures DM4 and implies that DM3 holds except possibly at $v_L$ when $v_L>0$. In that case, continuity at $v_L$ ensures $\phi(v_L)=u(v_L)$, so (C2) also guarantees DM3. Therefore, it remains to show that DM2 is equivalent to (C1) and DM1 is equivalent to (C2).

To establish that DM2 reduces to (C1), note that by construction $\phi(v)=u(v)$ for $v\in [0,v_L)$, and $\phi(\cdot)\geq u(\cdot)$ for any $v\geq r$. The latter inequality follows from $u(v)=\widetilde{\alpha}+(1-\widetilde{\alpha})\min\{F(v_L)^{n-1}+\beta(v-r),1\}\leq \phi(v)$ on $[r,v_H]$. Therefore, DM2 can fail only on $[v_L,r)$, i.e., DM2 if and only if (C1).

To establish that DM1 reduces to (C2), note from eq. \eqref{eq_PHI} that the only possible points of discontinuity are $v_L$, when $v_L>0$, and $v_H$ when $v_H<1$. The continuity at $v_H$ follows by the feasibility condition $G(v_H)=F(v_H)$ (Proposition \ref{prop_G_shape}). Similarly, the only possible non-convexity occurs at $v_L$, when $v_L>0$, and at $v_H$ when $v_H<1$. The convexity at $v_H$ follows from Lemma \ref{beta_properties} 1), which states that $\beta<[F(v_H)^{n-1}]'$. Therefore, DM1 can fail only at $v_L$ whenever $v_L>0$, i.e., DM1 if and only if (C2).    

\noindent \textit{Step 3}: Suppose that $v_L=0$. We show that $G(\cdot)$ given by \eqref{eq_G} is an equilibrium if and only if $\phi(0)\geq 0$. By Step 2, (C1) is the only remaining condition. Note that $u(v)=0$ for $v<r$ while $\phi(\cdot)$ is increasing and affine on $[0,r]$. Thus, a necessary and sufficient condition for $\phi(\cdot)\geq u(\cdot)$ on $[0,r)$ is $\phi(0)\geq 0$. This establishes {\it (i)}.   

\noindent \textit{Step 4}: Suppose that $v_L>0$. We show that $G(\cdot)$ given by \eqref{eq_G} is an equilibrium if and only if $\phi(\cdot)$ is continuous and convex at $v_L$. By Step 2, (C1) and (C2) are the only necessary and sufficient conditions that remain. Recall that by Proposition \ref{prop_G_shape}, $G(v_L)=G(r)=F(v_L)$. Hence, $u(\cdot)$ is constant on $[v_L,r)$, whereas $\phi(\cdot)$ given by \eqref{eq_PHI} is strictly increasing on the same interval. Moreover, since $u(\cdot)$ is continuous on $[0,r)$ and $\phi(v)=u(v)$ for $v<v_L$, continuity of $\phi(\cdot)$ at $v_L$ implies $\phi(v_L)=u(v_L)$ and therefore ensures that $\phi(v)\geq u(v)$ for every $v\in[v_L,r)$. Thus, (C2) implies (C1). Thus, the continuity and convexity at $v_L$ (i.e., C2) is the only remaining necessary and sufficient condition. This establishes {\it (ii)}.
\end{proof}

\subsection{Additional Cost Heterogeneity}
In this section we prove Proposition \ref{prop_hetero}. 
\paragraph{Single-Cost Benchmark.} Consider a \textit{single-cost benchmark} in which the \inexperienced{} consumer's cost is always $s_1$, which is the smallest cost in the support of $K(\cdot)$. The Proof of Proposition \ref{prop_large_n}, shows that for all $n$ sufficiently large, the unique equilibrium is 
\begin{equation*}
G(v|s_1,n)=
\begin{cases}
0&\text{ for } v\in[0,r_1)\\
(\beta^*(n))^{\frac{1}{n-1}}(v-r_1)^{\frac{1}{n-1}} & \text{for } v\in[r_1,v^*_H(n)]\\
F(v) & \text{for } v\in (v_H^*(n),1],
\end{cases}
\end{equation*}
where
\begin{align*}
\beta^*(n)=\frac{F(v^*_H(n))^{n-1}}{v^*_H(n)-r_1},\qquad r_1=\mu-s_1, \qquad\text{and}\qquad E_F[v|v<v_H^*(n)]=\frac{v_H^*(n)}{n}+\frac{n-1}{n}r_1.
\end{align*}
This equilibrium is supported by multiplier 
\begin{align}\label{Amult}
    \phi(v)\equiv
    \begin{cases}
\widetilde{\alpha}+(1-\widetilde{\alpha})\beta^*(n)(v-r_1) & \text{if } v\leq v^*_H(n)\\
\widetilde{\alpha}+(1-\widetilde{\alpha})F(v)^{n-1} & \text{if } v\in (v_H^*(n),1]\end{cases}.
\end{align}
Furthermore, because the \inexperienced{} consumer always stops at the first firm ($G(r_1|s_1,n)=0$), we have 
\begin{align*}
    \eta=\frac{1}{n},\qquad\qquad\qquad\widetilde{\alpha}=\frac{\alpha(\frac{1}{n})}{\alpha(\frac{1}{n})+(1-\alpha)}.
\end{align*}
The associated payoff function $u(\cdot|s_1,n)$ is zero below $r_1$ and coincides with $\phi(\cdot)$ above, i.e., 
\begin{align*}
    u(v|s_1,n)=
    \begin{cases}
    0& \text{if } v\in[0,r_1)\\
\widetilde{\alpha}+(1-\widetilde{\alpha})\beta^*(n)(v-r_1) & \text{if } v\in[r_1, v^*_H(n)]\\
\widetilde{\alpha}+(1-\widetilde{\alpha})F(v)^{n-1} & \text{if } v\in (v_H^*(n),1]\end{cases}.
\end{align*}
The multiplier and payoff function of the single-cost benchmark are illustrated in the left panel of Figure \ref{fig_hetero}.

\medskip

\paragraph{Introducing Cost Heterogeneity.} Consider a distribution of search costs for the \inexperienced{} consumer, $K(\cdot)$, satisfying the model's assumptions. Suppose that firms continue to use disclosure strategy $G(\cdot|s_1,n)$ from the single-cost benchmark. We wish to derive the associated payoff function $u_K(\cdot)$ in the presence of cost heterogeneity for the \inexperienced{} consumer.

\begin{claim}\label{Aa} If firms use disclosure strategy $G(\cdot|s_1,n)$, then the \inexperienced{} consumer stops at the first visited firm for all realizations of the search cost in the support of $K(\cdot)$. Furthermore, the reservation value of a \inexperienced{} consumer with search cost $s\in[s_1,s_k]$ is $r(s)=\mu-s$.
\end{claim}
The claim follows from elementary properties of the reservation value. Note that reservation value is decreasing in search cost. Because $s_1$ is the minimum search cost in the support of $K(\cdot)$ the \inexperienced{} consumer's search cost is weakly higher and the reservation value is weakly smaller than the reservation value of type-$s_1$, i.e. for any $s\in(s_1,s_k]$, we have $r(s)<r(s_1)=\mu-s_1$. Because the support of $G(\cdot|s_1,n)$ is $[\mu-s_1,1]$, every realized value weakly exceeds $r(s)$ for all $s\in[s_1,s_k]$. Consequently, the \inexperienced{} consumer stops at the first visited firm regardless of the realized search cost. Furthermore, because the entire support of $G(\cdot|s_1,n)$ is above $r(s)$ for $s\in[s_1,s_k]$, the search equation implies $r(s)=\mu-s$.

\begin{claim}\label{Ab}Suppose costs are distributed according to $K(\cdot)$ and firms use disclosure strategy $G(\cdot|s_1,n)$. The probability of visit by a \inexperienced{} consumer is $\eta=1/n$, and  the probability that a consumer is \inexperienced{} conditional on visit is 
\begin{align*}
\widetilde{\alpha}=\frac{\alpha(\frac{1}{n})}{\alpha(\frac{1}{n})+(1-\alpha)},
\end{align*}
both identical to the single-cost benchmark.
\end{claim}
The claim follows immediately from Claim \ref{Aa}. Because the \inexperienced{} type stops at the first visited firm regardless of realized cost, the probability of visit is $1/n$. The posterior belief follows immediately from Bayes' rule.

\begin{claim}\label{Ac}Suppose costs are distributed according to $K(\cdot)$ and firms use disclosure strategy $G(\cdot|s_1,n)$. The probability of sale with realized posterior mean $v$ is 
\begin{align*}
    u_K(v)=
    \begin{cases}
    \widetilde{\alpha}(1-K((\mu-v)^-))& \text{if } v\in[0,r_1)\\
\widetilde{\alpha}+(1-\widetilde{\alpha})\beta^*(n)(v-r_1) & \text{if } v\in[r_1, v^*_H(n)]\\
\widetilde{\alpha}+(1-\widetilde{\alpha})F(v)^{n-1} & \text{if } v\in (v_H^*(n),1]\end{cases}.
\end{align*}
\end{claim}
The calculation is parallel to the one in the main text.  
\begin{align*}
    u_K(v)=\widetilde{\alpha}\Pr(\text{chosen}\mid v\cap \text{visit}\cap\text{casual})+(1-\widetilde{\alpha})\Pr(\text{chosen}\mid v\cap \text{visit}\cap\text{\savvy{}}).
\end{align*}
For $v\geq r_1$, we have 
\begin{align*}
    &\Pr(\text{chosen}\mid v\cap \text{visit}\cap\text{casual})=1\\
    &\Pr(\text{chosen}\mid v\cap \text{visit}\cap\text{\savvy{}})=G(v|s_1,n)^{n-1}.
\end{align*}
Therefore, for $v\geq r_1$, we have $u_K(v)=\widetilde{\alpha}+(1-\widetilde{\alpha})G(v|s_1,n)^{n-1}=u(v|s_1,n)$. Next, consider $v<r_1$. Note that 
\begin{align*}
    \Pr(\text{chosen}\mid v\cap \text{visit}\cap\text{\savvy{}})=G(v|s_1,n)^{n-1}=0.
\end{align*}
Furthermore, 
\begin{align*}
\Pr(\text{chosen}\mid\text{visit}\cap v\cap\text{casual})=\frac{\Pr(\text{sale}\cap\text{visit}\mid v\cap\text{casual})}{\Pr(\text{visit}\mid v\cap\text{casual})}.
\end{align*} 
Note that a product is only chosen by a \inexperienced{} consumer if it is the first encountered, and if the consumer's reservation value $r(s)\leq v$. Therefore $\Pr(\text{chosen}\cap\text{visit}\mid v\cap\text{casual})=\eta\Pr(r(s)\leq v)=\eta\Pr(\mu-s\leq v)=\eta (1-K((\mu-v)^-))$. Noting that $\Pr(\text{visit}\mid v\cap\text{casual})=\eta$, we have the expression in the claim. 

Note that graphically, the payoff function $u_K(\cdot)$ is a horizontal reflection of $K(\cdot)$ around $v=\mu/2$, followed by a vertical reflection around $y=1/2$, and then a scaling down the $y$ axis to $\widetilde{\alpha}$. The right panel of Figure \ref{fig_hetero} illustrates.

\begin{figure}
\begin{minipage}{0.5\textwidth}
\begin{tikzpicture}[scale=0.6]

\pgfmathsetmacro{\al}{0.1}
\pgfmathsetmacro{\et}{0.4}
\pgfmathsetmacro{\vL}{0}
\pgfmathsetmacro{\r}{5}
\pgfmathsetmacro{\vH}{7.6}
\pgfmathsetmacro{\z}{1}
\pgfmathsetmacro{\n}{2}
\draw[line width=0.5pt] (0,0) -- (0,10.5); 
\draw[line width=0.5pt] (0,0) -- (10.5,0); 
\draw[dotted, line width=0.5pt] (0,10) -- (10,10) -- (10,0);

\fill (\r,0) node[below] {\footnotesize{$r=\mu-s_1$}};
\fill (\vH,0) node[below] {\footnotesize{$v_H$}};

\draw[dotted, line width=0.5pt] (\vH,0)--
  (\vH,{10*\al + 10*(1-\al)*(0.1*\vH)^\n});

 \draw[line width=1pt, teal] (\vL,0.1*\vL^\n)--(\r,0.1*\vL^\n);
 \draw[line width=1pt, teal, dashed] 
   (\r,0.1*\vL^\n) --
   (\r,{0.1*(\al/\et+1-\al)*\vL^\n+\z+(10*\al+10*(1-\al)*(0.1*\vH)^\n -\z- 0.1*(\al/\et+1-\al)*\vL^\n)/(\vH-\vL)*(\r-\vL)});

\draw[decorate, decoration={brace,amplitude=10pt}]
  (\r,{0.1*(\al/\et+1-\al)*\vL^\n+\z+(10*\al+10*(1-\al)*(0.1*\vH)^\n -\z - 0.1*(\al/\et+1-\al)*\vL^2)/(\vH-\vL)*(\r-\vL)})
  -- (\r,0.1*\vL^\n);

\draw[violet, line width=1pt, domain=\vL:\vH] 
  plot (\x,{0.1*(\al/\et+1-\al)*\vL^\n+\z + (10*\al + 10*(1-\al)*(0.1*\vH)^\n-\z - (0.1*(\al/\et+1-\al)*\vL^\n))/(\vH-\vL)*(\x-\vL)});
\draw[violet, line width=1pt, domain=\vH:10] 
  plot (\x,{10*\al + 10*(1-\al)*(0.1*\x)^\n});

\fill (0,0) node[left] {\footnotesize{$0$}};
\fill (0,10) node[left] {\footnotesize{$1$}};
\fill (\r+0.5,2.4) node[right]{\footnotesize{$\widetilde{\alpha}$}};
 \fill (2,0) node[above] {\footnotesize{$u(\cdot)$}};
\fill (8,8.2) node[above] {\footnotesize{$\phi(\cdot)$}};
\fill (6,0) node[above] {\footnotesize{$I_{a}$}};
\node[rotate=35] at (3,3.5) {\scriptsize{slope $(1-\widetilde{\alpha})\beta^*(n)$}};

\draw[green, line width=3pt, opacity=0.5] (\r,0)--(\vH,0);

\end{tikzpicture}
\end{minipage}
\begin{minipage}{0.5\textwidth}
\begin{tikzpicture}[scale=0.6]

\pgfmathsetmacro{\al}{0.1}
\pgfmathsetmacro{\et}{0.4}
\pgfmathsetmacro{\vL}{0}
\pgfmathsetmacro{\r}{5}
\pgfmathsetmacro{\vH}{7.6}
\pgfmathsetmacro{\z}{1}
\pgfmathsetmacro{\n}{2}
\pgfmathsetmacro{\alt}{0.1*(\al/\et+1-\al)*\vL^\n+\z+(10*\al+10*(1-\al)*(0.1*\vH)^\n -\z- 0.1*(\al/\et+1-\al)*\vL^\n)/(\vH-\vL)*(\r-\vL)}
\pgfmathsetmacro{\ra}{0.75*\r}
\pgfmathsetmacro{\rb}{0.47*\r}
\pgfmathsetmacro{\rc}{0.25*\r}
\pgfmathsetmacro{\ha}{0.25*\alt}
\pgfmathsetmacro{\hb}{0.2*\alt}
\pgfmathsetmacro{\hc}{0.15*\alt}

\draw[line width=0.5pt] (0,0) -- (0,10.5); 
\draw[line width=0.5pt] (0,0) -- (10.5,0); 
\draw[dotted, line width=0.5pt] (0,10) -- (10,10) -- (10,0);

\fill (\r,0) node[below] {\footnotesize{$r_1$}};
\fill (\vH,0) node[below] {\footnotesize{$v_H$}};

\draw[dotted, line width=0.5pt] (\vH,0)--
  (\vH,{10*\al + 10*(1-\al)*(0.1*\vH)^\n});

 \draw[line width=1pt, teal] (0,0)--(\rc,0);
 \draw[line width=1pt, teal] (\rc,\hc)--(\rb,\hc);
 \draw[line width=1pt, teal] (\rb,\hc+\hb)--(\ra,\hc+\hb);
 \draw[line width=1pt, teal] (\ra,\hc+\hb+\ha)--(\r,\hc+\hb+\ha);
 
 \draw[line width=1pt, teal, dashed] 
   (\r,\ha+\hb+\hc) --
   (\r,\alt);
   
\draw[line width=0.5pt, black, dotted] 
   (\r,\ha+\hb+\hc) --(\r,0);

   \draw[line width=1pt, teal, dashed] 
   (\ra,\ha+\hb+\hc) --
   (\ra,\hb+\hc);
   \filldraw[color=teal, fill=teal] (\ra,\ha+\hb+\hc) circle (2pt);
\draw[line width=0.5pt, black, dotted] 
   (\ra,\hb+\hc) --(\ra,0);
   \fill (\ra,0) node[below]{\footnotesize{$r_2$}};

    \draw[line width=1pt, teal, dashed] 
   (\rb,\hc+\hb) --
   (\rb,\hc);
\draw[line width=0.5pt, black, dotted] 
   (\rb,\hc) --(\rb,0);
   \fill (\rb,0) node[below]{\footnotesize{$r_3$}};
   \filldraw[color=teal, fill=teal] (\rb,\hb+\hc) circle (2pt);
 
\draw[line width=1pt, teal, dashed] 
   (\rc,\hc) --
   (\rc,0);
   \fill (\rc,0) node[below]{\footnotesize{$r_4$}};
   \filldraw[color=teal, fill=teal] (\rc,\hc) circle (2pt);

\draw[line width=0.5pt, red] 
   (\rc,\hc) --
   (\r,\alt);

\node[rotate=45] at (\rc+0.5,\hc+0.8) {\tiny{slope $b(v)$}};
   
\draw[decorate, decoration={brace,amplitude=10pt}]
  (\r,\alt)
  -- (\r,0);

  \draw[decorate, decoration={brace,amplitude=5pt}]
  (\ra,\ha+\hb+\hc) --
   (\ra,\hb+\hc) node[midway,xshift=12pt] {\scriptsize{$\widetilde{\alpha}p_2$}};;

\draw[decorate, decoration={brace,amplitude=5pt}]
  (\rb,\hb+\hc) --
   (\rb,\hc) node[midway,xshift=12pt] {\scriptsize{$\widetilde{\alpha}p_3$}};;

\draw[decorate, decoration={brace,amplitude=5pt}]
  (\rc,\hc) --
   (\rc,0) node[midway,xshift=12pt] {\scriptsize{$\widetilde{\alpha}p_4$}};;

\draw[violet, line width=1pt, domain=\vL:\vH] 
  plot (\x,{0.1*(\al/\et+1-\al)*\vL^\n+\z + (10*\al + 10*(1-\al)*(0.1*\vH)^\n-\z - (0.1*(\al/\et+1-\al)*\vL^\n))/(\vH-\vL)*(\x-\vL)});
\draw[violet, line width=1pt, domain=\vH:10] 
  plot (\x,{10*\al + 10*(1-\al)*(0.1*\x)^\n});

\fill (0,0) node[left] {\footnotesize{$0$}};
\fill (0,10) node[left] {\footnotesize{$1$}};
\fill (\r+0.5,\alt/2) node[right]{\footnotesize{$\widetilde{\alpha}$}};
 
\fill (8,8.2) node[above] {\footnotesize{$\phi(\cdot)$}};
\fill (6,0) node[above] {\footnotesize{$I_{a}$}};
\node[rotate=35] at (3,3.5) {\scriptsize{slope $(1-\widetilde{\alpha})\beta^*(n)$}};

\draw[green, line width=3pt, opacity=0.5] (\r,0)--(\vH,0);

\end{tikzpicture}
\end{minipage}

\caption{\label{fig_hetero} 
\textit{Additional Cost Heterogeneity.}
Left panel: A single \inexperienced{} type $s_1$.
Right panel: Multiple \inexperienced{} types.
Multiplier in violet, payoff in teal with $p_i=\Pr(s=s_i)$. The red line interpolates $(v,u_K(v))$ and $(r_1,u_K(r_1))$. We denote its slope $b(v)$.}
\end{figure}
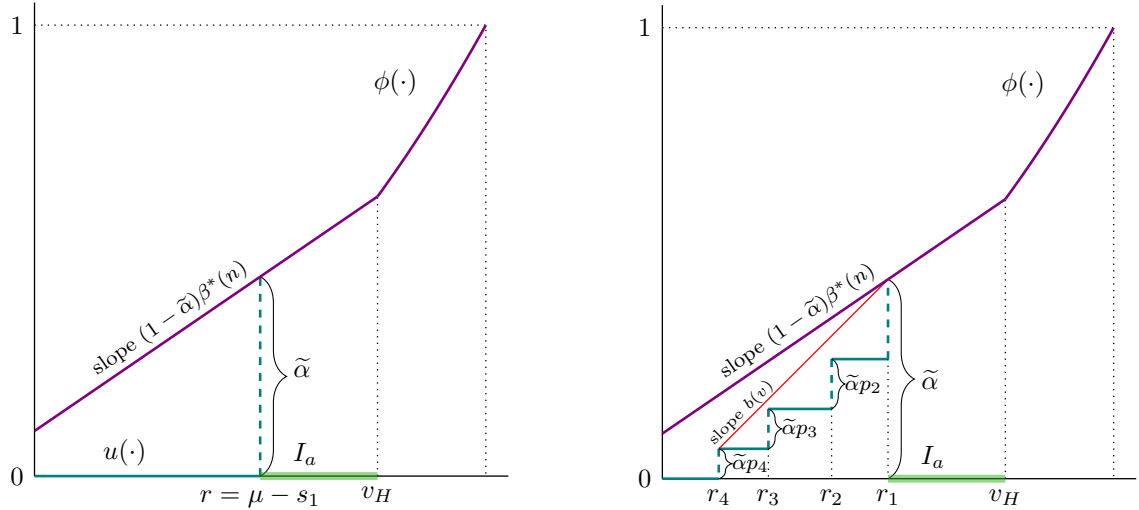

\paragraph{Equilibrium Construction.}
\begin{claim}\label{Ad} $G(\cdot\mid s_1,n)$ is an equilibrium in the model with cost heterogeneity if the associated multiplier $\phi(\cdot)$ in \eqref{Amult} lies above $u_K(\cdot)$ on $[0,r_1]$, i.e., $\phi(v)\geq u_K(v)$ for all $v\in[0,r_1]$.
\end{claim}
The only difference between the problem characterizing the firm's best response in the single-cost benchmark and the model with heterogeneous cost is the payoff function on interval $[0,r_1]$. By implication, the only optimality condition that $\phi(\cdot)$ might fail to satisfy is (DM-2),  which requires $\phi(v)\geq u_K(v)$ for $v\in[0,r_1)$.

Let $b(v)$ be the slope of the line interpolating $(v,u_K(v))$ and $(r_1,u_K(r_1))$, 
\begin{align*}
    b(v)\equiv \frac{u_K(r_1)-u_K(v)}{r_1-v}=\widetilde{\alpha}\frac{K((\mu-v)^-)}{r_1-v}.
\end{align*}
Recall that the multiplier is linear on $[0,r_1]$. Clearly, if the slope of the multiplier $\phi(\cdot)$ is equal to $b(v)$, then the multiplier crosses $u_K(\cdot)$ at $v$ (consult Figure \ref{fig_hetero}). Furthermore, if the slope of the multiplier is smaller than $b(v)$, then the multiplier lies strictly above the payoff at $v$, i.e., 
\begin{align*}
(1-\widetilde{\alpha})\beta^*(n)< b(v)\Rightarrow \phi(v)>u_K(v)
\end{align*}
Therefore, a sufficient condition for $\phi(v)>u_K(v)$ for all $v\in[0,r_1)$ is 
\begin{align}\label{Aineq}
    &(1-\widetilde{\alpha})\beta^*(n)<\inf_{v\in[0,r_1)}b(v)\iff
    (1-\widetilde{\alpha})\beta^*(n)<\widetilde{\alpha}\inf_{v\in[0,r_1)}\frac{K((\mu-v)^-)}{r_1-v}.
\end{align}
\begin{claim}\label{Ae} Given the assumptions on $K(\cdot)$,
\begin{align*}
b^*\equiv \inf_{v\in[0,r_1)}\frac{K((\mu-v)^-)}{r_1-v}>0.
\end{align*}
\end{claim}
First note for $v\in[0,r_1)$, we have $\mu-v>\mu-r_1=s_1$. Therefore, for all such $v$, we have $K((\mu-v)^-)>0$. By implication, for all such $v$, the ratio $K((\mu-v)^-)/(r_1-v)>0$, and because $K(\cdot)$ is non-decreasing, this ratio is bounded away from 0 on interval $[0,r_1-\epsilon]$ for any $\epsilon>0$. Thus, the infimum over $[0,r_1)$ is strictly positive as long as 
\begin{align*}
    \lim_{v\rightarrow r_1^-}\frac{K((\mu-v)^-)}{r_1-v}>0.
\end{align*}
First, consider a discrete distribution. In this case, $\lim_{v\rightarrow r_1^-}K((\mu-v)^-)=K(s_1)>0$. That is, the numerator is simply the probability of the lowest type in the support, $s_1$. Because the denominator approaches 0, the limit in question is infinite.

Second, consider a continuous distribution. In this case, the numerator approaches $K(s_1)=0$, and the denominator approaches $r_1-r_1=0$. With continuous $k(\cdot)$, we evaluate the limit by L'Hospital's rule,
\begin{align*}
    \frac{-k(\mu-r_1)}{-1}=k(s_1)>0,
\end{align*}
and the last inequality follows from assumption (b-ii).

To complete the Proof of Proposition \ref{prop_hetero}, we show that for sufficiently large $n$, the inequality in \eqref{Aineq} is satisfied. By implication, $\phi(\cdot)$ supports $G(\cdot|s_1,n)$ as an equilibrium, even with cost heterogeneity for  the \inexperienced{} consumer.
\begin{claim}\label{Af} For sufficiently large $n$, $(1-\widetilde{\alpha})\beta^*(n)<\widetilde{\alpha}b^*.
$
\end{claim}
Consider the following equivalent inequalities
\begin{align*}
    (1-\widetilde{\alpha})\beta^*(n)<\widetilde{\alpha}b^*\iff
    \frac{(1-\widetilde{\alpha})}{\widetilde{\alpha}}\beta^*(n)<b^*\iff
    \frac{(1-\alpha)}{\alpha\eta}\beta^*(n)<b^*\iff
    n\beta^*(n)<\frac{\alpha}{1-\alpha}b^*.
\end{align*}
Next, note that for all $n$ sufficiently large, $v_H^*(n)-r_1>0$. Furthermore, as $n\rightarrow\infty$, $v_H^*(n)\rightarrow v_H^\infty$ such that $E_F[v|v<v_H^\infty]=r_1$. By implication $v_H^{\infty}\in(r_1,1)$. It follows that
\begin{align*}
    \lim_{n\rightarrow\infty}n\beta^*(n)=\lim_{n\rightarrow\infty}\frac{nF(v_H^*(n))^{n-1}}{v^*_H(n)-r_1}=0.
\end{align*}
Hence, the desired inequality holds for all $n$ sufficiently large.

\section{Search Technology: extended result}

Section 5.2 discusses the non-monotonicity of information provision with changes in the \inexperienced{} consumers' search cost. In particular, it shows that the equilibrium must approach full disclosure as the search cost approaches the extremes of zero and the prior mean $\mu$. Section 5.2 further shows that, for high search costs, a small decrease in the search cost  reduces informativeness and has a disparate effect on \inexperienced{} and \savvy{} consumers. 

For completeness, here we extend the characterization by describing the comparative statics for small search costs. The equilibrium structure in this case is more complicated as it includes disclosure at the bottom. While we are unable to establish a monotonic relationship between informativeness and search costs in this case, we can still show that, consistent with standard intuition, a sufficient drop in the search cost  increases information provision and the welfare of both consumers.

 \begin{proposition} (Small Search Costs).
 For each $(n,\alpha,F(\cdot))$, there exists $\underline{s}>0$ such that if $s'<s<\underline{s}$, then the equilibrium with $s'$ cannot be less informative than the equilibrium with $s$. Furthermore, for each $s\in(0,\underline{s})$, there exists $\widetilde{s}\in(0,s]$ such that, for all $s''\in(0,\widetilde{s})$, equilibrium informativeness is strictly higher, and the surplus of both consumer types is strictly larger with $s''$ than with $s$.
  
\label{com_stat_s_2}
\end{proposition}

\begin{proof}
{\it Step 1:} We first show that for sufficiently small $s$, the equilibrium satisfies $v_H^*(s)=1$ and $v_T^*(s)<1$. Since by Proposition \ref{com_stat_s}, $v_L^*(s)\to 1$ as $s\to 0$ and by Lemma \ref{optdu1}, $\phi(\cdot)$ is convex at $v_L^*(s)$, it follows that there exists $\underline{s}>0$ such that:

\[
[F(1)^{n-1}]'\leq \frac{1}{1-\alpha}[F(v_L^*(s))^{n-1}]'\leq \beta^*(s) \quad \text{for all } s\leq \underline{s}.
\]
By the strict convexity of $F(\cdot)^{n-1}$, $\beta^*(s)>[F(v)^{n-1}]'$ for all $s\leq \underline{s}$ and $v<1$.  Then, by Lemma \ref{beta_properties} 1), $v_H^*(s)=1$ for $s\leq \underline{s}$. Moreover, if $v_T^*(s)=1$, then $G_s(v)^{n-1}=F(v_L^*(s))^{n-1}+\beta^*(s)(v-r^*(s))$ for $v\in [r^*(s),1]$. Then, $G_s(1)=F(1)=1$, and $\beta^*(s)>[F(v)^{n-1}]'$ implies $G_s(v)<F(v)$ for $v\in [r^*(s),1)$. This, in turn, yields $H(r^*(s))=\int_{r^*(s)}^1[G_s(v)-F(v)]dv<0$, contradicting $G_s\in MPC(F)$. Therefore, $v_T^*(s)<1$ for $s\leq \underline{s}$.

\noindent {\it Step 2:} Fix $s'<s<\underline s$ and define $\Delta_{s',s}(v):=G_{s'}(v)^{n-1}-G_s(v)^{n-1}$. We next show that $G_s\notin MPS(G_{s'})$. Implicitly differentiating eq. (\ref{eq_search2}) and evaluating at the equilibrium $(r^*,v_L^*)$ yields
\[
\frac{dr^*}{ds}=\frac{1}{(r^*-v_L^*)f(v_L^*)\frac{d v_L^{eq}(r^*)}{d r}-(1-F(v_L^*))}<0,
\]
where the inequality follows from $\frac{d v_L^{eq}(r^*)}{dr}=\left(\frac{d r^{eq}(v_L^*)}{dv_L}\right)^{-1}$ and $\frac{dr^{eq}(v_L^*)}{dv_L}>\frac{dr^{se}(v_L^*)}{dv_L}=\frac{(r^*-v_L^*)f(v_L^*)}{1-F(v_L^*)}$, where the inequality is established in the proof of Proposition \ref{prop_r_endog} and the equality follows from eq. (\ref{eq_r_fi_slope}).  By Proposition \ref{prop_r_exog}, $v_L^{eq}(r)$ increases in $r$ implying that $v_L^*(s)$ also decreases in $s$. Hence, by eq. (\ref{eq_G}), $G_{s'}(v)=G_s(v)$ ($\Delta_{s',s}(v)=0$) for $v\in [0,v_L^*(s)]$ and $G_{s'}(v)>G_s(v)$ ($\Delta_{s',s}(v)>0$) for $v\in (v_L^*(s),r^*(s)]$ since $G_s(\cdot)$ is constant on this interval, while $v_L^*(s')>v_L^*(s)$ and $r^*(s')>r^*(s)$. Thus, $\int_{0}^{r^*(s)} [G_{s'}(v)-G_s(v)]dv>0$, which implies $G_s\notin MPS(G_{s'})$.

\noindent {\it Step 3:} Next, note that $\lim_{s\to 0} v_T^*(s)=1$ and $\lim_{s\to 0} r^*(s)=1$ (since $G(\cdot)\to F(\cdot)$), and hence there exists $\widetilde{s}\in(0,s]$ such that for all $s''\in(0,\widetilde{s})$, $v_T^*(s)<r^*(s'')<v_T^*(s'')$. Clearly, $G_{s}(\cdot)$ and $G_{s''}(\cdot)$ have the same mean. We next show that $G_{s''}\in MPS(G_s)$ by establishing that $G_{s''}(\cdot)$ crosses $G_s(\cdot)$ from above exactly once on $(v_L^*(s),v_T^*(s''))$.

As above, $\Delta_{s'',s}(v)=0$ on $[0,v_L^*(s)]$ and $\Delta_{s'',s}(v)>0$ on $(v_L^*(s),r^*(s)]$. On $(r^*(s),v_T^*(s))$, $G_s(\cdot)^{n-1}$ is affine with slope $\beta^*(s)$, while $G_{s''}(\cdot)$ either coincides with $F(\cdot)$ or is constant for $v<r^*(s'')$. Therefore, for $v\in(r^*(s),v_T^*(s))$,
\[
\Delta_{s'',s}'(v)=
\begin{cases}
[F(v)^{n-1}]'-\beta^*(s) & \text{for } v<v_L^*(s''),\\
-\beta^*(s) & \text{for } v>v_L^*(s'').
\end{cases}
\]
Recall from above that $[F(v)^{n-1}]'<\beta^*(s)$ for all $v<1$, implying that $\Delta_{s'',s}(\cdot)$ is strictly
decreasing on $(r^*(s),v_T^*(s))$. Moreover, $\Delta_{s'',s}(r^*(s))>0$ and $\Delta_{s'',s}(v)=G_{s''}(v)^{n-1}-1<0$ for $v\in [v_T^*(s),v_T^*(s''))$. Therefore, $G_{s''}(\cdot)$ crosses $G_s(\cdot)$ from above exactly once on $(r^*(s),v_T^*(s''))$ and, since the distributions have the same mean, $G_{s''}\in MPS(G_s)$.

\noindent {\it Step 4:} We next show that consumer surplus of both consumer types is strictly larger with $s''$ than with $s$. To establish $\CSsav(s'')>\CSsav(s)$, note that 
\[
\CSsav(s'')-\CSsav(s)=\int_0^1 v dG_{s''}^n-\int_0^1 v dG_{s}^n=\int_{0}^1[G_s(v)^n-G_{s''}(v)^n]dv.
\]
Let $c$ denote the unique crossing point that satisfies $G_s(c)=G_{s''}(c)$, and define $W(v):=\sum_{k=0}^{n-1}G_s(v)^{n-1-k}G_{s''}(v)^k$. Then, $G_s(\cdot)^n-G_{s''}(\cdot)^n=[G_s(\cdot)-G_{s''}(\cdot)]W(\cdot)$, and accounting for the equality of means, we obtain 
\[
\int_0^1[G_s(v)^n-G_{s''}(v)^n]dv=\int_0^1[G_s(v)-G_{s''}(v)][W(v)-W(c)]dv.
\]
Note that $G_s(\cdot)-G_{s''}(\cdot)$ and $W(\cdot)-W(c)$ have the same sign since, as established in {\it Step 3}, $G_{s''}(\cdot)$ crosses $G_{s}(\cdot)$ exactly once from above at $c$, and $W(\cdot)$ is an increasing function. Therefore, the integrand is nonnegative everywhere and strictly positive wherever the two distributions differ. Since the distributions differ on intervals of positive length, the integral is strictly positive, establishing $\CSsav(s'')>\CSsav(s)$.

To establish the same inequality for the \inexperienced{} consumers, let $M(v_L,r)$ denote the expected value of $\widetilde v=\min\{v,r\}$, where
\[
M(v_L,r):=E_{\widetilde G(\cdot\mid v_L,r)}[\widetilde v]=\int_0^{v_L}v\,dF(v)+(1-F(v_L))r=\mu-\int_{v_L}^1(v-r)\,dF(v).
\]
By eq. (\ref{eq_search2}), $M(v_L^*(s),r^*(s))=\mu-s$. Note that $\frac{\partial M}{\partial v_L}<0$ and $\frac{\partial M}{\partial r}>0$. Since $v_L^*(s)<v_L^*(s'')$ and $r^*(s)<r^*(s'')$,
\[
M(v_L^*(s),r^*(s))=\mu-s<\mu-s''=M(v_L^*(s''),r^*(s''))<M(v_L^*(s),r^*(s'')).
\]
By the intermediate value theorem, there exists $\overline r\in(r^*(s),r^*(s''))$ such that $M(v_L^*(s),\overline r)=\mu-s''$. Let $\overline G:=\widetilde G(\cdot\mid v_L^*(s),\overline r)$. By construction, $\overline G$ first-order stochastically dominates $\widetilde G(\cdot\mid v_L^*(s),r^*(s))$. Moreover, $\overline G$ is obtained from $\widetilde G(\cdot\mid v_L^*(s''),r^*(s''))$ by pooling all mass above $v_L^*(s)$ at its conditional mean $\overline r$, and hence is an MPC of that distribution. Therefore,
\[
\begin{aligned}
\CSin(s)=E_{\widetilde G(\cdot\mid v_L^*(s),r^*(s))}[\widetilde v^{\max}]<E_{\overline G}[\widetilde v^{\max}] \leq E_{\widetilde G(\cdot\mid v_L^*(s''),r^*(s''))}[\widetilde v^{\max}]=\CSin(s''),
\end{aligned}
\]
where the first inequality is strict because the upper atom rises from $r^*(s)$ to $\overline r$, and the second follows from convex order.
\end{proof}

\end{document}